\documentclass[a4paper,10pt]{article}
\usepackage{authblk}

\usepackage[ngerman,english]{babel}
\usepackage[latin1]{inputenc}

\usepackage{csquotes}
\usepackage{amsfonts,amsmath,amsthm}
\usepackage{graphicx}
\usepackage{booktabs}
\usepackage{empheq}
\usepackage[titletoc,title]{appendix}
\usepackage[backend=bibtex8,doi=false,eprint=false,firstinits=true,isbn=false,style=numeric-comp,maxnames=99]{biblatex}
\makeatletter
\def\blx@maxline{77}
\makeatother

\bibliography{mpm_bib.bib}
\AtEveryBibitem{\clearlist{language}}

\usepackage{cases}
\usepackage{mathabx}
\usepackage{bbm}
\usepackage{xfrac}

\usepackage{fancyhdr}
\usepackage{color}
\usepackage[colorinlistoftodos,textsize=small,backgroundcolor=white,bordercolor=magenta,linecolor=magenta]{todonotes}
\usepackage[colorlinks]{hyperref}
\definecolor{blue75}{rgb}{0,0,.75}
\definecolor{green75}{rgb}{0,.75,0}
\hypersetup{colorlinks=true, urlcolor=blue75,linkcolor=blue75,citecolor=green75,pdfstartview=FitB,bookmarksopen=true,bookmarksopenlevel=1}
\usepackage[a4paper, left=2.5cm, right=2.5cm, top=2.5cm,bottom=2cm]{geometry}
\usepackage{constants}
\newcommand{\parenthezises}[1]{\arabic{#1}}
\newconstantfamily{C}{
symbol=C,
format=\parenthezises,
}
\newconstantfamily{M}{
symbol=M,
format=\parenthezises,
}
\newconstantfamily{B}{
symbol=B,
format=\parenthezises,
}
\newconstantfamily{b}{
symbol=b,
format=\parenthezises,
}
\newconstantfamily{S}{
symbol=S,
format=\parenthezises,
}
\usepackage{enumerate}
\usepackage{graphicx}
\graphicspath{{images/} }
\usepackage{wrapfig}
\usepackage{figbib}
\allowdisplaybreaks
\usepackage[capitalise]{cleveref}

\crefdefaultlabelformat{{\it #2#1#3}}

\crefname{equation}{}{}

\crefname{enumi}{}{}
\creflabelformat{enumi}{{(#2#1#3)}}

\crefname{section}{{\it Section}}{{\it Sections}}
\crefname{subsection}{{\it Subsection}}{{\it Subsections}}
\crefname{subsubsection}{{\it Paragraph}}{{\it Paragraphs}}

\newtheorem{Theorem}{Theorem}[section]
\crefname{Theorem}{{\it Theorem}}{{\it Theorems}}

\crefname{Definition}{{\it Definition}}{{\it Definitions}}
\newtheorem{Lemma}[Theorem]{Lemma}
\crefname{Lemma}{{\it Lemma}}{{\it Lemmas}}

\crefname{Proposition}{{\it Proposition}}{{\it Propositions}}

\crefname{Assumption}{{\it Assumption}}{{\it Assumptions}}

\crefname{Assumptions}{{\it Assumptions}}{{\it Assumptions}}

\theoremstyle{definition}
\newtheorem{Remark}[Theorem]{Remark}
\crefname{Remark}{{\it Remark}}{{\it Remarks}}
\newtheorem{Notation}[Theorem]{Notation}
\crefname{Notation}{{\it Notation}}{{\it Notations}}
\newtheorem{Example}[Theorem]{Example}
\crefname{Example}{{\it Example}}{{\it Examples}}

\renewbibmacro{doi+eprint+url}{%
    \printfield{doi}%
    \newunit\newblock%
    \iftoggle{bbx:eprint}{%
        \usebibmacro{eprint}%
    }{}%
    \newunit\newblock%
    \iffieldundef{doi}{%
        \usebibmacro{eprint}}%
        {}%
    }
\usepackage{subcaption}
\begin{document}

\newcommand{\cb}{\color{blue}}
 \newcommand{\red}[1]{\textcolor{red}{#1}}
\newcommand{\cmg}[1]{\textcolor{magenta}{#1}}

\newcommand{\D}{\mathbb{D}}
\newcommand{\E}{\mathbb{E}}
\newcommand{\T}{\mathbb{T}}
\newcommand{\PP}{\mathbb{P}}
\newcommand{\R}{\mathbb{R}}
\newcommand{\N}{\mathbb{N}}
\newcommand{\F}{\mathbb{F}}
\newcommand{\V}{\mathbb{V}}
\newcommand{\ve}{\varepsilon}
\newcommand{\I}{\{c,m,n,w\}}
\newcommand{\numb}{4}
\def\diag{\operatorname{diag}}
\def\diam{\operatorname{diam}}
\def\dist{\operatorname{dist}}
\def\diver{\operatorname{div}}
\def\ess{\operatorname{ess}}
\def\inner{\operatorname{int}}
\def\osc{\operatorname{osc}}
\def\sign{\operatorname{sign}}
\def\supp{\operatorname{supp}}
\def\tr{\operatorname{trace}}
\newcommand{\BMO}{BMO(\Omega)}
\newcommand{\LOne}{L^{1}(\Omega)}
\newcommand{\LOnen}{(L^{1}(\Omega))^d}
\newcommand{\LTwo}{L^{2}(\Omega)}
\newcommand{\Lq}{L^{q}(\Omega)}
\newcommand{\Lp}{L^{2}(\Omega)}
\newcommand{\Lpn}{(L^{2}(\Omega))^d}
\newcommand{\LInf}{L^{\infty}(\Omega)}
\newcommand{\HOneO}{H^{1,0}(\Omega)}
\newcommand{\HTwoO}{H^{2,0}(\Omega)}
\newcommand{\HOne}{H^{1}(\Omega)}
\newcommand{\HTwo}{H^{2}(\Omega)}
\newcommand{\HmOne}{H^{-1}(\Omega)}
\newcommand{\HmTwo}{H^{-2}(\Omega)}

\newcommand{\LlogL}{L\log L(\Omega)}

\def\avint{\mathop{\,\rlap{-}\!\!\int}\nolimits} 

\newcommand{\om}{\omega }
\newcommand{\Om}{\Omega }

\newtheorem{proofpart}{Step}
\makeatletter
\@addtoreset{proofpart}{Theorem}
\makeatother
\numberwithin{equation}{section}
\title{Multiphase modelling of glioma pseudopalisading under acidosis}
	
\author{Pawan Kumar\thanks{Felix-Klein-Zentrum für Mathematik, Technische Universität Kaiserslautern, Paul-Ehrlich-Str. 31, 67663 Kaiserslautern, Germany, \href{mailto:surulescu@mathematik.uni-kl.de}{kumar@mathematik.uni-kl.de}},\ 
    Christina Surulescu\thanks{Felix-Klein-Zentrum für Mathematik, Technische Universität Kaiserslautern, Paul-Ehrlich-Str. 31, 67663 Kaiserslautern, Germany, \href{mailto:surulescu@mathematik.uni-kl.de}{surulescu@mathematik.uni-kl.de}}, \ and Anna Zhigun\thanks{School of Mathematics and Physics, Queen's University Belfast, University Road, Belfast BT7 1NN, Northern Ireland, UK, \href{mailto:A.Zhigun@qub.ac.uk}{A.Zhigun@qub.ac.uk}}
   }

\date{}
\maketitle
\begin{abstract}
We propose a multiphase modeling approach to describe glioma pseudopalisade patterning under the influence of acidosis. The phases considered at the model onset are glioma, normal tissue, necrotic matter, and interstitial fluid in a void-free volume with acidity represented by proton concentration. We start from mass and momentum balance to characterize the respective volume fractions and deduce reaction-cross diffusion equations for the space-time evolution of glioma, normal tissue, and necrosis. These are supplemented with a reaction-diffusion equation for the acidity dynamics and lead to formation of patterns which are typical for high grade gliomas. Unlike previous works, our deduction also works in higher dimensions and involves less restrictions. We also investigate the existence of weak solutions to the obtained system of equations and perform numerical simulations to illustrate the solution behavior and the pattern occurrence. 

\end{abstract}
\section{Introduction}

Glioblastoma is the most common type of primary brain tumors in adults, with a dismal prognosis. The histological features include characteristic patterns called pseudopalisades, which exhibit garland-like structures made of aggregates of glioma cells stacked in rows at the periphery of regions with low pH and high necrosis surrounding the occlusion site(s) of one or several capillaries \cite{Wippold2037}. Such patterns are used to grade the tumor and are essential for diagnosis \cite{Brat2003,Kleihues1995}. 

The few continuous mathematical models proposed for the description of pseudopalisade patterns  \cite{Alfonso,kumar2020flux,kumar2020multiscale,martinez2012hypoxic}, involve systems of ODEs and PDEs which were set up in a heuristic manner or obtained from lower scale dynamics. They account for various aspects like phenotypic switch between proliferating and migrating glioma as a consequence of vasoocclusion and nutrient suppression associated therewith \cite{Alfonso}, interplay between normoxic/hypoxic glioma, necrotic matter, and oxygen supply \cite{martinez2012hypoxic}, or tissue anisotropy and repellent pH-taxis, with \cite{kumar2020flux} or without \cite{kumar2020multiscale} vascularisation. The latter two works deduced effective equations for glioma dynamics from models set on the lower, microscopic and mesoscopic scales in the kinetic theory of active particles (KTAP) framework explained e.g., in \cite{bellomo2008}. Those deductions are in line with previous works concerning glioma invasion in anisotropic tissue \cite{CONTE2021,CEKNSSW,Dietrich2020,engwer2015glioma,engwer2015effective,EKS,Hunt2016,PH13}, which use parabolic scaling to obtain the equations for glioma evolution on the macroscale from descriptions of subcellular and mesoscopic dynamics. Here we propose yet another approach, relying on the interpretation of the relevant components (glioma, normal tissue, necrotic matter, and acidity) as phases in a mixture - with the exception of acidity, which is characterised by proton concentration in the volume occupied by the phases. 

Multiphase models in the framework of mixture theory \cite{Atkin1976, Drew1998} have been considered in connection with cancer growth and invasion, (arguably) starting with \cite{Breward02,BYRNE2003} and followed by many others, involving different mathematical and biological aspects; see e.g.  \cite{Garcke2018,Jackson2002307,Preziosi2008, Preziosi2011,Scium2013,TOSIN2010969,WALDELAND2018268,Wise2008} and references therein. Particularly \cite{Scium2013} provides a comprehensive discussion of the classes of multiphase models, along with their strengths and drawbacks.  Such models employ a population level description with conservation laws analogous to balance equations for single cells, with supplementary terms accounting for interphase effects. Thereby, the living cells and tissues are most often seen as viscous fluids, while the interstitial fluid is inviscid. Few of the existing models (e.g., \cite{Garcke2018,Scium2013,Wise2008}) are explicitly handling necrosis, which is, however, an important component of advanced tumors. Likewise, there are relatively few models accounting for chemoattracting or chemorepellent agents (e.g., \cite{Garcke2018,WALDELAND2018268}), which are known to bias in a decisive way the expansion of the neoplasm and therewith associated dynamics of tumor cells and surrounding tissue. To our knowledge there are no multiphase models for glioma pseudopalisade development. In this context we propose such a model comprising necrotic matter as one of the phases in the mixture (although there are a few  issues related to this approach, as mentioned, e.g. in \cite{Scium2013}). As our main aim is to describe glioma behavior in an acidic environment fastly leading to extensive necrosis, we explicitly include these influences in our model. From the mass and momentum balance equations written for the phases composing the neoplasm and supplemented with appropriate constitutive relations, we then deduce, under some simplifying assumptions, a system of reaction-cross diffusion equations with repellent pH-taxis for the glioma cells and normal and necrotic tissues, also including a solenoidality constraint on the total flux. Our deduction has some similarities with that in \cite{Jackson2002307}, however is done in N dimensions instead of 1D, it accounts explicitly for the evolution of necrotic matter and the effects of acidity, and does not require the drag coefficients between the phases to be equal. The obtained equations are able to reproduce qualitatively the typical pseudopalisade patterns with all their aspects related to the glioma aggregates, acidity, necrotic inner region, normal tissue dynamics.

The rest of this paper is organized as follows: Section \ref{Prelim} contains some notations and conventions. Section \ref{Setup} contains the model setup with the considered phases (glioma cells, normal tissue, necrotic matter, and interstitial fluid) and the corresponding mass and momentum balance equations, along with their associated constitutive relations. That description is used to deduce the announced reaction-cross diffusion equations. Thereby, we investigate a model with an immovable component and observe that enforcing one of the phases to be fixed prevents the simultaneous validity of all basic conservation laws needed in the setting. Then we neglect the interstitial fluid, thus reducing the number of phases, and obtain the PDEs with reaction, nonlinear diffusion and taxis terms, while still ensuring all balance equations. The total flux needs to be given, satisfying a solenoidality constraint. Section \ref{Existence} is dedicated to the existence of weak solutions to the obtained cross diffusion system coupled with the reaction-diffusion equation for acidity dynamics in the volume of interest. Finally, in Section \ref{SecNum} we provide numerical simulations for that system, studying several parameter scenarios in order to put in evidence the effect of acidity and of different drag coefficients on the obtained patterns exhibiting pseudopalisade formation.

\section{Preliminaries}\label{Prelim}
We will use the following notation throughout this paper:

\begin{Notation}\label{NotVec} \begin{enumerate}
\item By vector we always mean a column vector.
\item We denote by $e$ the vector of length $\numb$ and all components equal to one. As usual, $I_l$ stands for the identity matrix of size $l$.
\item For two vectors $w^{(1)}$ and $w^{(2)}$ of the same length we denote by $w^{(1)}.w^{(2)}$ the vector with elements $w^{(1)}_iw^{(2)}_i$. Similarly, for a vector $w$ with nonzero elements we write $w.^{-1}$ meaning the vector with coordinates $w_i^{-1}$.
\item For a vector $w$ we denote by $\diag(w)$ the diagonal matrix with $(\diag(w))_{ii}=w_i$.
\item As usual, $\cdot$ denotes the scalar product.
\item As usual, $\nabla$ refers to the gradient with respect to the spatial variable $x$.
\end{enumerate}

\end{Notation}
\begin{Notation}
 Throughout the paper we often skip the arguments of coefficient functions.
\end{Notation}
\section{Modelling}

\subsection{Model setup}\label{Setup}
\paragraph{Model variables}
Motivated by the multiphase approach developed in \cite{Jackson2002307} we view the tumour and its environment as a saturated mixture of several components. We assume these components to be:  tumour cells, normal tissue (mainly the extracellular matrix, but also normal cells), necrotic tissue (was not included in \cite{Jackson2002307}), and interstitial fluid. The latter, in turn, has several constituents, among which are protons. Depending on their concentration, the environment can be more or less acidic. pH levels drop due to enhanced glycolytic activity of neoplastic cells.    
We assume that neither cells nor living tissue are produced due to an 
 already very  acidic, and hence very  unfavourable,  environment. Thus we assume that the total volume of the mixture is preserved, the phases only transferring  from
one  into another.
The main variables in our  models, all depending on time $t\geq0$ and position in space $x\in\Omega\subset \R^N$, $\Omega$ a domain, are:
\begin{itemize}
\item vector $u:[0,\infty)\times\Omega\rightarrow[0,1]^{\numb}$ of volume fractions of \begin{itemize}
       \item  tumour cells, $u_c$, 
\item   normal tissue, $u_m$,
\item   necrotic tissue, $u_n$,
\item   interstitial fluid, $u_w$.
      \end{itemize}
\item  acidity (concentration of protons), $h:[0,\infty)\times\Omega\rightarrow[0,\infty)$.
\end{itemize}
Unlike \cite{Jackson2002307} we do not require the space dimension, $N\in\N$, to be one. 
 The main goal of this Section is to derive  equations for the  volume fractions based primarily on the conservation laws for mass and momentum as well as  additional assumptions on some of the phases. 
To write down the  physical laws, we introduce additional variables that are subsequently eliminated. These are:
\begin{itemize}
 \item  fluxes of the components, $J_i:[0,\infty)\times\Omega\rightarrow\R^{N}$, $i\in \I$,
 \item common pressure, $p:[0,\infty)\times\Omega\rightarrow\R$.
\end{itemize}
\paragraph{Model parameters} 
The equations we develop below involve the following set  of parameters: 
\begin{itemize}
\item matrix of drag coefficients associated with each pair of  components, $K\in     \R^{\numb\times\numb}$;
 \item additional, isotropic pressures by the components, $\tau_i:[0,1]^{\numb}\times\R\rightarrow\R$, $i\in \I$;
 \item reaction terms, $f_i: [0,1]^{\numb}\times\R\rightarrow\R$, $i\in \{c,m,n,w,h\}$;
 \item diffusion coefficient of the protons, $D_h>0$.
\end{itemize}
\paragraph{Assumptions on $K$.} We assume throughout that 
\begin{align*}
 &K_{ij}>0\quad\text{and}\quad K_{ij}=K_{ji}\quad\text{for }\quad i,j\in\I,\ i\neq j.
\end{align*}
\paragraph{Assumptions on $f_i$'s.}
In order to ensure that the sum of all volume fractions  in the mixture is always one, we require
 \begin{align}
 \underset{i\in\I}{\sum}f_{ i}\equiv0.\label{sumf}
\end{align}
Since  $u_i$'s and $u_h$ should be nonnegative, we require further that 
\begin{align*}
 &f_i\geq 0\qquad \text{for }u_i=0,\ i\in\I,\\
 &f_h\geq0\qquad\text{for }h=0.
\end{align*}

\noindent
Possible choices for the reaction terms are:
\begin{Example} $f_c(u_c,u_n,h)=-c_1u_cu_n(h-h_{max})$, $f_m(u_m,h)=-\frac{c_2u_m(h-h_{max})_+}{1+u_m+h}-c_3u_m$, $f_n=-f_c-f_m$, $f_w=0$. 
		This accounts for the fact that the amount of both tissue and viable cancer cells (the latter being in interaction with the necrotic matter embedded in the acidic environment) decreases when the proton concentration $h$ exceeds a certain maximum threshold, leading to acidosis and hypoxia. The tissue infers degradation due to causes other than direct influence of acidity or tumor cells, and, on the other hand, it experiences a certain amount of self-regeneration, which is limited by the already available tissue and low pH. Thereby, $c_1,c_2,c_3>0$ are constant rates; they could, however, also include further dependencies. For the reaction term in the acidity equation we choose e.g., $f_h=-a_1u_wh/(1 + u_wh/h_{max}) + a_2u_c - a_3h$, with $a_i\ge 0$ ($i=1,2,3$) constants. This choice ensures proton uptake by the interstitial fluid (with saturation), production by glioma cells, and decay with rate $a_3$.
 \end{Example}
\paragraph{Assumptions on $\tau_i$'s.} 
We assume that the non-living components (necrotic matter and extracellular fluid) acquire a tendency to reach equilibrium, so that only living matter exerts additional pressure, thus
\begin{align}
 \tau_n\equiv\tau_w\equiv0.\label{Atau}
\end{align}
The mechanical properties of living matter (tumour cells and normal tissue) are different; among others, they are able to generate both intra- and interphase forces in response to changes in the local environment. The influences include variations in the volume fraction of the cell phase and the presence of chemical cues, of which we account for the local acidity (via proton concentration). The corresponding forces manifest themselves as an additional intraphase pressure. It is assumed that the pressure in the tumour and normal tissue phases increases with their respective densities; moreover, glioma cells exert supplementary (isotropic) pressure on the normal tissue. Taking these effects into account, we choose for the additional pressure terms
\begin{align}
 &\tau_c(u_c,h):=\alpha_cu_c+g(h),\qquad g(h):=\frac{\chi { h}}{1+\frac{h}{h_{max}}},\label{tau3}\\
 &\tau_m(u_c,u_m):=\alpha_m(1+\theta u_c)u_m,\label{tau4}
\end{align}
where  $\alpha_c,\alpha_m,\theta,h_{max},\chi>0$ are constants. In \cref{tau3} we combine   the influence of local cell mass and acidity, both adding to cancer cell pressure which, in turn, enhances glioma motility. 
Cell stress increases with growing cell mass, pushing the cells away from overcrowded regions. Stress due to acidity 
saturates for large acidity levels. Indeed, in highly acidic regions cell ion channels and pH-sensing receptors are mostly occupied, thus making cells insensitive to the presence of protons in their environment.   
As in \cite{Dietrich2020,kolbe_etal21,kumar2020multiscale,kumar2020flux}, we call this a repellent pH-tactic behaviour.
Similarly, the choice in \eqref{tau4} means that there is some intraspecific tissue stress (compression) further accentuated by the interaction between glioma cells and their fibrous environment. The forms proposed in \eqref{tau3}, \eqref{tau4} are reminiscent of those chosen in \cite{Jackson2002307}.  

\paragraph{Main equations} To derive our models we will rely on a set of equations describing  physical laws. They are as follows.
\begin{itemize}
 \item Since the the components of $u$ are volume fractions, we have 
\begin{align}
\underset{i\in\I}{\sum}u_i=1.\label{novoid} 
\end{align}
This is the so-called  'no void' condition. 
   \item Mass conservation 
for $i$th component is given by 
  \begin{align}
    \partial_t u_i=-\nabla\cdot J_i+f_i.\label{CL}                                                                                                                          \end{align}
    It would be reasonable to assume that this equation should hold for all phases. However, this may come into conflict with additional assumptions, see subsequent \cref{SecHapto}.
  \item Conservation of momentum  for $i$th component  is given by 
  \begin{align}
   -\nabla\cdot(u_i\sigma_i)=F_i,\label{Stl1}
  \end{align}
     where 
     \begin{align}
&\sigma_i=-(p+\tau_i)I_N\label{sigm}
\end{align} is the stress tensor involving the common pressure $p$ and the additional pressure $\tau_i$. It is assumed that the response of the tumour to stress is elastic and isotropic, which implies that the deformation induced by the applied force on each of the considered cellular and tissue components is limited. An unlimited deformation would correspond to a viscoelastic material \cite{Jones2000}, which is not considered here. The forms of the stress tensors depend on the material properties of each phase and their response to mechanical and chemical cues in the environment. As in \cite{Jackson2002307,lemon2007multiphase} the interstitial fluid is considered to be inert and isotropic and viscous effects within each phase are neglected.

The right hand side in \eqref{Stl1} represents the force acting on the $i$th phase and, neglecting any inertial effects and exterior body forces, it takes the following form: 
\begin{align}
 F_i=p\nabla u_i+\sum_{j\in\I\backslash\{i\}}K_{ij}(u_iJ_j-u_jJ_i).\label{Force}
\end{align}
The first term on the right hand side in \eqref{Force} accounts as usual (see, e.g.  \cite{Breward02,Jackson2002307,lemon2007multiphase}) for the pressure distribution at the interface between phases. The remaining term represents viscous drag between the phases, with drag coefficients $K_{ij}$.

\noindent
Plugging \cref{sigm} and \cref{Force} into \cref{Stl1} yields an equation for the variables we introduced above:
 \begin{align}
   \nabla\cdot(u_i(p+\tau_i)I_N)=p\nabla u_i+\sum_{j\in\I\backslash\{i\}}K_{ij}(u_iJ_j-u_jJ_i).\label{Stl}
  \end{align} 
  Equation \cref{Stl} is assumed to hold for all four phases. 
  \item The proton concentration satisfies the reaction-diffusion equation
  \begin{align}
   \partial_th=D_h\Delta h+f_h.\label{EqProt}
  \end{align}

\end{itemize}
In the remainder of this section we derive two models based on these laws. To begin with, we simplify equation \cref{Stl}. Given that $p$ and $\tau_i$ are scalar functions, we have for all $i\in\I$
\begin{align}
 \sum_{j\in\I\backslash\{i\}}K_{ij}(u_iJ_j-u_jJ_i)
 =&\nabla\cdot(u_i(p+\tau_i)I_N)-p\nabla u_i\nonumber\\=&\nabla(u_i(p+\tau_i))-p\nabla u_i\nonumber\\
 =&\nabla(u_i\tau_i)+u_i\nabla p.\label{MomCL1}
\end{align}
Since $K$ is symmetric, adding together the left-hand sides of \cref{MomCL1} for all $i\in\I$, we find that
\begin{align}
 \sum_{i\in\I}\sum_{j\in\I\backslash\{i\}}K_{ij}(u_iJ_j-u_jJ_i)=0.\label{sumKuij}
\end{align}
Combining \cref{MomCL1} with \cref{novoid} and \cref{sumKuij}, we obtain
\begin{align}
 0=&\sum_{i\in\I}\nabla(u_i\tau_i)-u_i\nabla p\nonumber\\
 =&\nabla(u\cdot\tau)+\nabla p,\nonumber
\end{align}
so that
\begin{align}
 \nabla p=-\nabla(u\cdot\tau).\label{px}
\end{align}
Plugging \cref{px} into \cref{MomCL1} we arrive at a system which no longer involves $p$:
\begin{align}
 \sum_{j\in\I\backslash\{i\}}K_{ij}(u_iJ_j-u_jJ_i)=-u_i\nabla(u\cdot\tau)+\nabla(u_i\tau_i),\qquad i\in\I,\label{Syst}
\end{align}
and it holds that
\begin{align}
 \underset{i\in\I}{\sum}\left(-u_i\nabla (u\cdot\tau)+\nabla (u_i\tau_i)\right)=0.\label{gsum}
\end{align}
Next, we introduce a matrix function
\begin{align}
 A:[0,1]^{\numb}\rightarrow\R^{\numb\times\numb},\qquad A_{im}:=\begin{cases}
          -\sum_{j\in\I\backslash\{i\}}K_{ij}u_j,&i=m,\\
          K_{im}u_i,&i\neq m,
         \end{cases}\qquad i,m\in\I.
\end{align}
The symmetry of $K$ once again implies that 
\begin{align}
 \underset{i\in\I}{\sum}A_{im}\equiv0,\qquad m\in\I.\label{Sumr}
\end{align}
System \cref{Syst} can now be written in the form
\begin{align}
 &AJ^{(l)}=-u\partial_{x_l} (u\cdot\tau)+\partial_{x_l} (\tau.u),\qquad l\in\{1,\dots,N\},\label{EqSM}
\end{align}
where we denote by $J^{(l)}$  the vector made up of the  $l$th coordinates of each $J_i$. Due to \cref{gsum,Sumr} for each $l$ (at  least) one equation is redundant.  We exploit this in more detail in the next Subsections.
\subsection{A model with an immovable component}\label{SecHapto}
Recall that {$u_m$} and $u_n$ correspond to normal and necrotic tissues, respectively. A standard modelling assumption is  that any  tissue is  completely  immovable. This means that its flux is  a zero function, turning the mass  conservation law \cref{CL} into an ODE. However, as we show in this Section, already presupposing, e.g.
\begin{align}
 J_n\equiv 0\label{Jn0}
\end{align}
is problematic. 

Singling out the $n$th phase, we use the following convenient notation.
\begin{Notation}
  For a vector $y$ with components corresponding to the four phases we denote by $\tilde y$  the vector which is obtained by removing the component which corresponds to $n$th phase. Similarly, for a matrix $A$ we denote by $\tilde A$ the submatrix which results from removing the  row and column of $A$ corresponding to that phase.
\end{Notation}
Recall that  fluxes $J_i$, $i\in\I$, satisfy system \cref{EqSM} which is underdetermined. 
We notice that the submatrix $\tilde A$ is diagonal-dominant
due to \cref{Sumr} and the nonnegativity of $u_i$'s and $K_{ij}$'s. Using the Gershgorin circle theorem, we infer that the real parts of the eigenvalues of $\tilde A$ do not exceed $-\rho_A(u_i)$, where
\begin{align}
 \rho_A:=u_n\min_{i\in\{c,m,w\}}K_{ni}.
\end{align}
Therefore, for $u_n>0$ the submatrix $\tilde A$ of $A$ is invertible. Consequently, system \cref{EqSM}-\cref{Jn0} is uniquely solvable.

To simplify the calculations, let us now assume that the  following technical assumption holds:
  \begin{align}
K=kk^T\quad\text{for some}\quad k\in[0,1]^{\numb}. \label{Kkk}       
\end{align}
In other words, the symmetric matrix $K$ has rank one. One readily verifies the following formulas:
\begin{Lemma} Let assumption \cref{Kkk} hold. Then:
\begin{align}
 \tilde{A}=-(k\cdot u)\left(I_3-\frac{1}{k\cdot u}(\tilde k.\tilde u)\tilde e^T\right)\diag(\tilde k)
\end{align}
and
\begin{align}
 \tilde{A}^{-1}=-\frac{1}{k\cdot u}\diag^{-1}(\tilde k)\left(I_3+\frac{1}{k_{n}u_{n}}(\tilde k.\tilde u)\tilde e^T\right), 
\end{align}
so that
\begin{align}
 \tilde{A}^{-1}\tilde w
 =&-\frac{1}{k\cdot u}\left(\tilde k.^{-1}.\tilde w+\frac{1}{k_{n}u_{n}}\tilde u(\tilde e\cdot\tilde w)\right).\label{Amw}
\end{align}
\end{Lemma}
\noindent Using \cref{Amw} we can resolve system  \cref{EqSM}-\cref{Jn0} with respect to the components of the fluxes.  For $l\in\{1,\dots,N\}$ and $i\in\{c,m,w\}$ we obtain
\begin{align}
  J^{(l)}_i=&-\frac{1}{k\cdot u}\left(\frac{1}{k_i} (-u_i\partial_{x_l} (u\cdot\tau)+\partial_{x_l} (u_i\tau_i))+\frac{1}{k_{n}u_{n}}u_i\tilde e\cdot(-\tilde u\partial_{x_l} (u\cdot\tau)+\partial_{x_l} (\tilde\tau.\tilde u))\right)\nonumber\\
  =&-\frac{1}{k\cdot u}\left(\frac{1}{k_i} (-u_i\partial_{x_l} (u\cdot\tau)+\partial_{x_l} (u_i\tau_i))+\frac{1}{k_{n}u_{n}}u_i(-(\tilde e\cdot\tilde u)\partial_{x_l} (u\cdot\tau)+\partial_{x_l} (\tilde u\cdot\tilde\tau))\right)\nonumber\\
  =&-\frac{1}{k\cdot u}\left(\frac{1}{k_i} (-u_i\partial_{x_l} (u\cdot\tau)+\partial_{x_l} (u_i\tau_i))+\frac{1}{k_{n}u_{n}}u_i(-(1-u_{n})\partial_{x_l} (u\cdot\tau)+\partial_{x_l} (\tilde u\cdot\tilde\tau))\right)
  \nonumber\\
  =&-\frac{1}{k\cdot u}\left(\frac{1}{k_i} (-u_i\partial_{x_l} (\tilde u\cdot\tilde\tau)+\partial_{x_l} (u_i\tau_i))+\frac{1}{k_{n}}u_i\partial_{x_l} (\tilde u\cdot\tilde\tau)\right)\nonumber\\
  =&-\frac{1}{k\cdot u}\left(\frac{1}{k_i}\partial_{x_l} (u_i\tau_i)-\left(\frac{1}{k_i}-\frac{1}{k_n}\right) u_i\partial_{x_l} (\tilde u\cdot\tilde\tau)\right).\label{Jl}
\end{align}
 We used \cref{novoid} and \cref{Atau} in the third and fourth equalities, respectively. 
Plugging \cref{Jl} into \cref{CL} and using \cref{Atau}, we obtain  
\begin{subequations}\label{syst3}
\begin{align}
 &\partial_t u_i=\nabla\cdot\left(\frac{1}{k\cdot u}\left(\frac{1}{k_i}\nabla(u_i\tau_i)-\left(\frac{1}{k_i}-\frac{1}{k_n}\right) u_i\nabla(u_c\tau_c+u_m\tau_m)\right)\right)+f_i,\qquad i\in\{c,m\},\\
 &\partial_t u_w=\nabla\cdot\left(\frac{1}{k\cdot u}\left(-\left(\frac{1}{k_w}-\frac{1}{k_n}\right) u_w\nabla(u_c\tau_c+u_m\tau_m)\right)\right)+f_w.
\end{align}
\end{subequations}
For variable $u_n$ we have due to \cref{CL,Jn0} that it satisfies an ODE:
\begin{align}
 \partial_tu_n=f_n.\label{eqn3}
\end{align}
However, the structure of the fluxes in \cref{syst3} would ensure \cref{novoid} only if  
\begin{align}
 \nabla\cdot\sum_{i\in\{c,m,w\}}J_i\equiv0.\nonumber
\end{align}
This condition fails to hold in general, unless for $\tau_c$ and $\tau_m$  such that
\begin{align}
 \sum_{i\in\{c,m,w\}}\left(\frac{1}{k_i}\nabla_u(u_i\tau_i)-\left(\frac{1}{k_i}-\frac{1}{k_n}\right) u_i\nabla_u(u_c\tau_c+u_m\tau_m)\right)=0\ \ \text{for all }0\leq u_c,u_m,u_n,\ \sum_{i\in\{c,m,w\}}u_i\leq1.
\end{align}
It is not fulfilled by the coefficients we have in mind, see \cref{Setup}. 

Let us assume that \cref{syst3} holds only for $i\in\{c,w\}$. Assume further that   
\begin{align}
u_w\equiv0,      
\end{align}
i.e. that the liquid phase is negligible.
In this case we obtain a  haptotaxis model:
\begin{subequations}\label{hapto}
\begin{align}
  &\partial_tu_c=\nabla\cdot\left(\frac{1}{k\cdot u}\left(\frac{1}{k_c}\nabla(u_c\tau_c(u,h))-\left(\frac{1}{k_c}-\frac{1}{k_n}\right) u_c\nabla(u_c\tau_c(u,h)+(1-(u_c+u_n))\tau_m(u,h))\right)\right)\nonumber\\
  &\qquad\quad +f_c(u,h),\\
  &\partial_tu_n=f_n(u,h),\\
  &u_m=1-(u_c+u_n),\\
  &u=(u_c,u_m,u_n).
\end{align}
\end{subequations}
Unless $\tau_c$ is zero for $u_c+u_n=1$, this model cannot ensure  $u_m\geq0$.

Altogether we see that trying to explicitly enforce an immovable phase  leads to a situation where not all basic conservation laws can hold at the same time.   
\subsection{A model with total flux control}
Let us now assume that the number of components in the mixture is three since
\begin{align*}
 u_w\equiv 0,
\end{align*}
thus, unlike previous multiphase models in a related context (see e.g. \cite{Jackson2002307,WALDELAND2018268}), but compare also \cite{TOSIN2010969}, we neglect the interstitial fluid and only take into account the more 'solid' components, namely cancer cells, necrotic matter, and normal tissue, all of which are assumed to be more or less heterogeneously interspersed within the volume of interest. 
 In this Subsection we assume all physical laws from \cref{Setup} to hold in full, so that, in particular, the mass conservation law \cref{CL} holds for each $i\in\{c,m,n\}$. 
Due to \cref{sumf} and \cref{novoid} we can replace \cref{CL} for $i=n$ by
\begin{align}
 &\nabla\cdot J_{sum}=0,\label{Jsum1}
\end{align}
where
\begin{align}
 J_{sum}:=\sum_{i\in\{c,m,n\}}J_{i}.\nonumber
\end{align}
As we have seen previously,  system  \cref{EqSM} is underdetermined. However, if  $J_{sum}$ is  given, then the system is equivalent to the following matrix equation:
\begin{align}
 \bar A(J_c,J_m,J_n)^T=(z_c,z_m,J_{sum})^T,\label{systm}
\end{align}
where 
\begin{align}
 \bar A=\left(\begin{matrix}
                  -(K_{cm}u_m + K_{cn}u_n)& K_{cm}u_c& K_{cn}u_c\\
 K_{cm}u_m& -(K_{cm}u_c + K_{mn}u_n)& K_{mn}u_m\\
 1& 1& 1                 \end{matrix}
\right)
\end{align}
and 
\begin{align*}
z_{i}:=-u_i\nabla(u\cdot\tau)+\nabla (u_i\tau_i).\end{align*}
One can readily verify that 
\begin{align}\label{inv-matrix}
 \bar A^{-1}=\frac{1}{K_{cm}^2\Cr{Sc}}\left(
\begin{matrix}
 -K_{cm} {u_c}-K_{mn} ({u_m}+{u_n}) & {u_c} (K_{cn}-K_{cm}) & {u_c}
   K_{cm}^2\Cr{Sc} \\
 {u_m} (K_{mn}-K_{cm}) & -K_{cm} {u_m}-K_{cn} ({u_c}+{u_n}) & {u_m}
   K_{cm}^2\Cr{Sc} \\
 K_{cm} ({u_c}+{u_m})+K_{mn} {u_n} & K_{cm} ({u_c}+{u_m})+K_{cn}
   {u_n} & {u_n} K_{cm}^2\Cr{Sc}
\end{matrix}
\right),
\end{align}
where
\begin{align}
 \Cl[S]{Sc}:=\frac{1}{K_{cm}^2}(K_{cm}K_{cn}u_c+K_{cm}K_{mn}u_m+K_{cn}K_{mn}u_n).
\end{align}
Using \cref{novoid}, we can rewrite {\eqref{inv-matrix}} as
\begin{align}
 \bar A^{-1}=&\frac{1}{K_{cm}^2\Cr{Sc}}\left(
\begin{matrix}
  {u_c}(K_{mn}-K_{cm})-K_{mn}& {u_c} (K_{cn}-K_{cm}) & {u_c}
   K_{cm}^2\Cr{Sc} \\
 {u_m} (K_{mn}-K_{cm}) & {u_m}(K_{cn}-K_{cm}) -K_{cn}  & {u_m}
   K_{cm}^2\Cr{Sc} \\
 {u_n} (K_{mn}-K_{cm}) +K_{cm}& {u_n}(K_{cn}-K_{cm})+K_{cm}
    & {u_n} K_{cm}^2\Cr{Sc}
\end{matrix}
\right)\nonumber\\
=&\frac{1}{K_{cm}^2\Cr{Sc}}(u(\Cr{bc},\Cr{bm},K_{cm}^2\Cr{Sc})+(\Cl[b]{Bm},\Cl[b]{Bn},(0,0,0)^T)),\label{Am}
\end{align}
where
\begin{align}
 (\Cl[B]{bc},\Cl[B]{bm}):=(K_{mn}-K_{cm},K_{cn}-K_{cm}),\label{BB}
\end{align}
\begin{align}
 (\Cr{Bm},\Cr{Bn}):=\left(
\begin{matrix}
  -K_{mn}& 0 \\
 0 &  -K_{cn} \\
 K_{cm}& K_{cm}
\end{matrix}
\right).\label{bb}
\end{align}
Resolving \cref{systm} with respect to $J_i$'s and using \cref{Am}, we get 
\begin{align}
 (J_c,J_m,J_n)=&\frac{1}{K_{cm}^2\Cr{Sc}}\left(\left(u\left(\Cr{bc},\Cr{bm},K_{cm}^2\Cr{Sc}\right)+\left(\Cr{Bm},\Cr{Bn},(0,0,0)^T\right)\right)(z_c,z_m,J_{sum})^T\right)^T\nonumber\\
 =&\frac{1}{K_{cm}^2\Cr{Sc}}\left(z_c(\Cr{bc}u+\Cr{Bm})^T+z_m(\Cr{bm}u+\Cr{Bn})^T\right)+J_{sum}u^T.\label{Jcmn}
\end{align}
Combining \cref{BB,bb,Jcmn}, we obtain
\begin{align}
 -J_c=&\frac{1}{K_{cm}^2\Cr{Sc}}\left((K_{cm}u_c+K_{mn}(1-u_c))z_c+(K_{cm}-K_{cn})u_cz_m\right)-J_{sum}u_c\nonumber\\
 =&\frac{1}{K_{cm}^2\Cr{Sc}}((K_{cm}u_c+K_{mn}(1-u_c))(-u_c\nabla(u\cdot\tau)+\nabla (u_c\tau_c))\nonumber\\
 &\qquad\ \ +(K_{cm}-K_{cn})u_c(-u_m\nabla(u\cdot\tau)+\nabla (u_m\tau_m)))-J_{sum}u_c\label{Jcc}
\end{align}
and
\begin{align}
 -J_m=&\frac{1}{K_{cm}^2\Cr{Sc}}\left((K_{cm}-K_{mn})u_mz_c+(K_{cm}u_m+K_{cn}(1-u_m))z_m\right)-J_{sum}u_m\nonumber\\
 =&\frac{1}{K_{cm}^2\Cr{Sc}}((K_{cm}-K_{mn})u_m(-u_c\nabla(u\cdot\tau)+\nabla (u_c\tau_c))\nonumber\\
 &\qquad\ \ +(K_{cm}u_m+K_{cn}(1-u_m))(-u_m\nabla(u\cdot\tau)+\nabla (u_m\tau_m)))-J_{sum}u_m.\label{Jmc}
\end{align}
Recalling that $\tau_{n}=0$ due to \cref{Atau}, we conclude  
from \cref{Jcc}-\cref{Jmc} that
\begin{align}
 -J_c+J_{sum}u_c
 =&\frac{1}{K_{cm}^2\Cr{Sc}}(K_{cm}u_c+K_{mn}(1-u_c))((1-u_c)\nabla (u_c\tau_c)-u_c\nabla(u_m\tau_m))\nonumber\\
 &+\frac{1}{K_{cm}^2\Cr{Sc}}(K_{cm}-K_{cn})u_c(-u_m\nabla(u_c\tau_c)+(1-u_m)\nabla(u_m\tau_m))\nonumber\\
 =&\frac{1}{K_{cm}^2\Cr{Sc}}\left((K_{cm}u_c+K_{mn}(1-u_c))(1-u_c)-(K_{cm}-K_{cn})u_cu_m\right)\nabla (u_c\tau_c)\nonumber\\
 &+\frac{1}{K_{cm}^2\Cr{Sc}}\left(-(K_{cm}u_c+K_{mn}(1-u_c))u_c+(K_{cm}-K_{cn})u_c(1-u_m)\right)\nabla(u_m\tau_m).
 \label{Jcm}
\end{align}
and, similarly, 
\begin{align}
 -J_m+J_{sum}u_m
 =&\frac{1}{K_{cm}^2\Cr{Sc}}\left(-(K_{cm}u_m+K_{cn}(1-u_m))u_m+(K_{cm}-K_{mn})u_m(1-u_c)\right)\nabla(u_c\tau_c)\nonumber\\
 &+\frac{1}{K_{cm}^2\Cr{Sc}}\left((K_{cm}u_m+K_{cn}(1-u_m))(1-u_m)-(K_{cm}-K_{mn})u_cu_m\right)\nabla (u_m\tau_m).
 \label{Jmm}
\end{align}

Set
\begin{align}
 &D_{cc}:=(K_{cm}u_c+K_{mn}(1-u_c))(1-u_c)+(K_{cn}-K_{cm})u_cu_m,\\
 &D_{cm}:=-(K_{cm}u_c+K_{mn}(1-u_c))u_c-(K_{cn}-K_{cm})u_c(1-u_m),\\
 &D_{mc}:=-(K_{cm}u_m+K_{cn}(1-u_m))u_m-(K_{mn}-K_{cm})u_m(1-u_c),\\
 &D_{mm}:=(K_{cm}u_m+K_{cn}(1-u_m))(1-u_m)+(K_{mn}-K_{cm})u_cu_m.
\end{align}
Then \cref{Jcm,Jmm} take the form
\begin{align}
 &-J_c=\frac{1}{K_{cm}^2\Cr{Sc}}D_{cc}\nabla(u_c\tau_c)+\frac{1}{K_{cm}^2\Cr{Sc}}D_{cm}\nabla(u_m\tau_m)-J_{sum}u_c,\label{Jc1}\\
 &-J_m=\frac{1}{K_{cm}^2\Cr{Sc}}D_{mc}\nabla(u_c\tau_c)+\frac{1}{K_{cm}^2\Cr{Sc}}D_{mm}\nabla(u_m\tau_m)-J_{sum}u_m.\label{Jm1}
\end{align}
Finally, combining \cref{novoid,CL,Jc1,Jm1,Jsum1,Atau,tau3,tau4} we arrive at the system
\begin{subequations}\label{model2}
\begin{align}
    &\partial_t u_c=\nabla\cdot \left(\frac{D_{cc}}{K_{cm}^2\Cr{Sc}}(u_c,u_m)\nabla(u_c\tau_c(u_c,h))+\frac{D_{cm}}{K_{cm}^2\Cr{Sc}}(u_c,u_m)\nabla(u_m\tau_m(u_c,u_m))-J_{sum}u_c\right)\nonumber\\
    &\quad\quad\   +f_c(u_c,u_m,h),\label{Equc}\\
    &\partial_t u_m=\nabla\cdot \left(\frac{D_{mc}}{K_{cm}^2\Cr{Sc}}(u_c,u_m)\nabla(u_c\tau_c(u_c,h))+\frac{D_{mm}}{K_{cm}^2\Cr{Sc}}(u_c,u_m)\nabla(u_m\tau_m(u_c,u_m))-J_{sum}u_m\right)\nonumber\\
    &\quad\quad \ +f_m(u_c,u_m,h),\label{Equm}\\
    &u_n:=1-u_c-u_m,\label{Equn}\\
    &\nabla\cdot J_{sum}=0.
 \end{align}
\end{subequations}
Our construction guarantees that $u_n$ satisfies
an equation similar to \cref{Equc,Equm}:
\begin{align}
 \partial_t u_n=&\nabla\cdot \left(\frac{D_{wc}}{K_{cm}^2\Cr{Sc}}(u_c,u_m)\nabla(u_c\tau_c(u_c,u_m,h))+\frac{D_{wm}}{K_{cm}^2\Cr{Sc}}(u_c,u_m)\nabla(u_m\tau_m(u_c,u_m,h))-J_{sum}u_n\right)\nonumber\\
 &+f_w(u_c,u_m,h),\label{Equnbis}
\end{align}
where 
\begin{align*}
 &D_{wc}:=-(D_{cc}+D_{mc}),\\
 &D_{wm}:=-(D_{cm}+D_{mm}),
\end{align*}
and 
\begin{align}
 D_{wc}=D_{wm}=0\qquad \text{for all }u\text{ such that }u_c+u_m=1.\nonumber
\end{align}
This a priori ensures that $u_n\geq0$ is satisfied if the solution components are smooth. 
\begin{Remark}
Model \cref{model2} is a generalisation of the model derived in \cite{Jackson2002307}. Unlike that work we require neither the constants $K_{ij}$ to be equal nor the space dimension $N$ to be one. 
\end{Remark}
\begin{Remark}
Our allowing the tissues to be displaced (thus giving up the hypothesis of their immovability) led to nonlinear diffusion and drift in \eqref{Equm} (or, for that matter, in \eqref{Equnbis}). This might seem unusual, as most reaction-difusion-transport systems describing cell motility in tissues assume the latter fixed and use ODEs for their evolution. We emphasize that the obtained terms do not state a self-driven motion of these components, but rather the effect of population- and biochemical pressure exerted thereon.
\end{Remark}
\begin{Remark}
 To close model \cref{model2}, one needs to choose some divergence-free $J_{sum}$. In dimension one and for no-flux boundary conditions $J_{sum}\equiv0$ is the only option. In higher dimensions $J_{sum}$ can be a curl field (in $3D$) or, more generally in $N$ dimensions (cf. \cite{barbarosie}), an exterior product of gradients: $J_{sum}=\nabla \gamma_1\wedge\nabla \gamma _2\wedge \dots \wedge \nabla \gamma _{N-1}$. In the physically relevant $3D$ case this means that $J_{sum}=\nabla \gamma _1\times \nabla \gamma _2$, thus there exist some scalar quantities $\gamma_1,\gamma _2$ (e.g., densities/volume fractions/concentrations) such that the total flux is orthogonal to each of their gradients, or, put in another way, $J_{sum}$ is tangential to a curve which lies in the intersection of the surfaces described by $\gamma_1$ and $\gamma_2$. This would mean that the total flux is following a direction which is equally biased by those two species.  
 The solenoidality of $J_{sum}$ is a reasonable assumption in view of our previous requirement that there were no sources or sinks of material, either.
\end{Remark}

\paragraph{Diffusion matrix} 
Without loss of generality we assume that
\begin{align}
 K_{mn}\geq K_{cn}\geq K_{cm}>0.\label{CompK}
\end{align}
One readily verifies that 
\begin{align}
 &\frac{1}{K_{cm}^2\Cr{Sc}}\left(\begin{matrix}
  D_{cc}&D_{cm}\\
  D_{mc}&D_{mm}
 \end{matrix}\right)=\frac{\Cr{Sm}}{K_{cm}\Cr{Sc}}\left(\begin{matrix}
  1-u_c-\frac{u_n}{\Cr{Sm}}\ve_c&-u_c\\
  -u_m&1-u_m-\frac{u_n}{\Cr{Sm}}\ve_m
 \end{matrix}\right),\label{DasPr}
\end{align}
where
\begin{align}
 &\ve_c:=\frac{K_{cn}}{K_{cm}}-1,\qquad \ve_m:=\frac{K_{mn}}{K_{cm}}-1,\label{defeps}\\
 &\Cr{Sc}:=(1+\ve_c)u_c+(1+\ve_m)u_m+(1+\ve_c)(1+\ve_m)(1-u_c-u_m),\\
 &\Cl[S]{Sm}:=1+\ve_c(1-u_m)+\ve_m(1-u_c).\label{defSm}
\end{align}
Due to assumptions \cref{CompK} we have for $0\leq u_c, u_m,u_c+u_m\leq1$ that
\begin{align}
&D_{cc},D_{mm}\geq0,\qquad D_{cm},D_{mc}\leq0,\nonumber\\
 &0\leq\ve_c\leq \ve_m,\nonumber\\
 &\Cr{Sc}\in[1+\ve_c,(1+\ve_c)(1+\ve_m)],\label{S1bnd}\\
 &\Cr{Sm}\in[1+\ve_c,1+\ve_1+\ve_m].\label{S2bnd}
\end{align}
In particular, for small $\ve_i$  matrix \cref{DasPr} can be regarded as  a  perturbation of  matrix
\begin{align*}
 \frac{1}{K_{cm}}\left(\begin{matrix}
  1-u_c&-u_c\\
  -u_m&1-u_m
 \end{matrix}\right)
\end{align*}
corresponding to  $$\ve_c=\ve_m=0,$$
i.e. to 
$$K_{mn}=K_{cn}=K_{cm}.$$
This case was addressed in \cite{JuengelStelzer}. 
Further, we compute
\begin{align}
 \left(\begin{matrix}
 \partial_{u_c}(u_c\tau_c)&\partial_{u_m}(u_c\tau_c)&\partial_{h}(u_c\tau_c)\\\partial_{u_c}(u_m\tau_m)&\partial_{u_m}(u_m\tau_m)&\partial_{h}(u_m\tau_m)
 \end{matrix}\right)=\left(\begin{matrix}
 2\alpha_cu_c&0&u_cg'\\ \theta\alpha_m u_m^2&2\alpha_m(1+\theta u_c)u_m&0
 \end{matrix}\right).\label{Dtau}
\end{align}
Overall, the diffusion matrix of equations \cref{Equc,Equm} takes the form
\begin{align}
 \frac{\Cr{Sm}}{K_{cm}\Cr{Sc}}\left(\begin{matrix}
  1-u_c-\frac{u_n}{\Cr{Sm}}\ve_c&-u_c\\
  -u_m&1-u_m-\frac{u_n}{\Cr{Sm}}\ve_m
 \end{matrix}\right)\left(\begin{matrix}
 2\alpha_cu_c&0&u_cg'\\ \theta\alpha_m u_m^2&2\alpha_m(1+\theta u_c)u_m&0
 \end{matrix}\right).\label{DiffA}
\end{align}
The complete diffusion matrix of \cref{Equc}, \cref{Equm}, and \cref{EqProt} includes the third line
\begin{align}
 (0,0,D_h).\label{line3}
\end{align}

\section{Existence of weak solutions to \cref{Equc}-\cref{Equn}, \cref{EqProt}}\label{Existence}
In this Section we use the 
method that was presented in \cite{Juengel2015}  in order to establish an existence result for  our cross diffusion
system \cref{Equc}-\cref{Equn}, \cref{EqProt}. The key to applying the method is finding a suitable so-called entropy density.
Motivated by the study in \cite{JuengelStelzer}, where the case of equal $K_{ij}$'s  and no acidity was treated, we consider the  following entropy density:
\begin{subequations}\label{Entrop}
\begin{align}
&E:{\cal D}\rightarrow\R,\ \ {\cal D}:={\cal D}_{cm}\times (0,\infty),\ \  {\cal D}_{cm}:=\{(u_c,u_m)\in(0,1)^2:\quad u_c+u_m<1\},\label{domain}\\
 &{E(u_c,u_m,h):=L(u_c,u_m)+\frac{a}{2}h^2,}\\
 &{L(u_c,u_m):=}u_c(\ln(u_c)-1)+u_m(\ln(u_m)-1)+u_n(\ln(u_n)-1),\\
 &u_n:=1-(u_c+u_m).
\end{align}
\end{subequations}
Here {$L$ is the well-known logarithmic entropy and }$a>0$ is a sufficiently large constant yet to be fixed. For the subsequent computations we need {the matrix of second-order partial  derivatives of $E$:}
\begin{align}
 D^2E{(u_c,u_m,h)}=\left(\begin{matrix}\frac{1}{u_n}\frac{1-u_m}{u_c}&\frac{1}{u_n}&0\\\frac{1}{u_n}&\frac{1}{u_n}\frac{1-u_c}{u_m}&0\\0&0&a\end{matrix}\right).\label{D2E}
\end{align}
{In order to be able to apply} the method from \cite{Juengel2015}{, we need to ensure positive (semi-)definiteness of}  matrix $(D^2E)M$ in ${\cal D}$ where $M$ is the diffusion matrix of  \cref{Equc}, \cref{Equm}, and \cref{EqProt}. In this Subsection we verify this property for the parameter values satisfying the following conditions:
\begin{align}
 0<&\underset{y\in[0,1]}{\min}(4\alpha_m\left(1+\ve_c(1-y)\right)\left(4\alpha_c\left(1+\ve_m\right)-2\theta\alpha_m y^2\ve_m\right)-\alpha_m^2y^2\left(\theta \left(1+\ve_c(1-y)\right)-2\ve_m\right)^2)\label{pos1_}
\end{align}
and
\begin{align}
 0<&\underset{y\in[0,1]}{\min}\left(4\alpha_m\left(1+\theta y\right)\left(1+\ve_cy\right)\left(4\alpha_c\left(1+\ve_m(1-y)\right)-2\theta\alpha_m (1-y)^2\ve_m\right)\right.\nonumber\\
 &\left.\qquad\ -\left(-2\alpha_cy\ve_c+\theta\alpha_m (1-y)\left(1+\ve_cy\right)-2\alpha_m(1-y)\left(1+\theta y\right)\ve_m\right)^2\right),\label{pos2_}
\end{align}
where $\ve_c$ and $\ve_m$ are constants defined in \cref{defeps}.
\begin{Remark}
 Since both functions which need to be minimised in \cref{pos1_,pos2_} are fourth degree polynomials in $y$,  these conditions can be readily checked numerically for a set of given parameters.
\end{Remark}

\begin{Lemma}[Uniform ellipticity]\label{LemmaEll}
 Let \cref{pos1_,pos2_} hold. Let $E$ be as defined in \cref{Entrop}. Then there {exist some  constants $a>0$ and   $0<\Cl[C]{C1}\leq\Cl[C]{C2}<\infty$, such that:}
 \begin{align}
  \Cr{C1}|y|^2\leq y^T((D^2E)M)(u_c,u_m,h)y\leq \Cr{C2}|y|^2\qquad\text{for all  }(u_c,u_m,h)\in{\cal D},\ y\in\R^{{3}}.
 \end{align}
\end{Lemma}
\begin{proof}
To begin with, we compute 
\begin{align}
 \Cr{Sm}\left(\begin{matrix}
         \partial^2_{u_cu_c}E&\partial^2_{u_cu_m}E\\\partial^2_{u_cu_m}E&\partial^2_{u_mu_m}E
        \end{matrix}\right)
\left(\begin{matrix}
  1-u_c-\frac{u_n}{\Cr{Sm}}\ve_c&-u_c\\
  -u_m&1-u_m-\frac{u_n}{\Cr{Sm}}\ve_m
 \end{matrix}\right)
 =&\left(\begin{matrix}\frac{1}{u_c}\beta_c&-\ve_m\\-\ve_c&\frac{1}{u_m}\beta_m\end{matrix}\right),\label{D2EM1}
\end{align}
where
\begin{align}
 &\beta_c:=1+\ve_m(1-u_c),\qquad \beta_m:=1+\ve_c(1-u_m),\nonumber
\end{align}
so that
\begin{align}
 &\beta_c\in[1,1+\ve_m],\qquad \beta_m\in[1,1+\ve_c].\label{betacm}
\end{align}
Combining \cref{D2E,D2EM1,DiffA,line3}, we obtain (arguments of functions are omitted)
\begin{align}
 (D^2E)M
 =&\frac{1}{K_{cm}\Cr{Sc}}\left(\begin{matrix}
 2\alpha_c\beta_c-\theta\alpha_m u_m^2\ve_m&-2\alpha_m(1+\theta u_c)u_m\ve_m&\beta_cg'\\
 -2\alpha_cu_c\ve_c+\theta\alpha_m u_m\beta_m&2\alpha_m(1+\theta u_c)\beta_m&-u_c\ve_c g'\\
 0&0&aD_h\Cr{Sm}
 \end{matrix}
\right).
\end{align}
We introduce the symmetric matrix
\begin{align*}
 P:=(D^2E)M+((D^2E)M)^T= \frac{1}{K_{cm}\Cr{Sc}}\left(\begin{matrix}
    \tilde P & \begin{matrix} \beta_cg' \\ -u_c\ve_c g' \end{matrix} \\
    \begin{matrix} \beta_cg' & -u_c\ve_c g' \end{matrix} & 2aD_h\Cr{Sm}
\end{matrix}\right),
\end{align*}
where
\begin{align}
 \tilde P:=\left(\begin{matrix}
        4\alpha_c\beta_c-2\theta\alpha_m u_m^2\ve_m&-2\alpha_cu_c\ve_c+\theta\alpha_m u_m\beta_m-2\alpha_m(1+\theta u_c)u_m\ve_m\\
        -2\alpha_cu_c\ve_c+\theta\alpha_m u_m\beta_m-2\alpha_m(1+\theta u_c)u_m\ve_m&4\alpha_m(1+\theta u_c)\beta_m
       \end{matrix}
\right).\label{MaP}
\end{align}
Next, we study  $\det(\tilde P)$ in ${\cal D}_{cm}$. We compute 
\begin{align}
d_P
:=&\det(\tilde P)\nonumber\\
=&4\alpha_m(1+\theta u_c)\beta_m(4\alpha_c\beta_c-2\theta\alpha_m u_m^2\ve_m)-(-2\alpha_cu_c\ve_c+\theta\alpha_m u_m\beta_m-2\alpha_m(1+\theta u_c)u_m\ve_m)^2\nonumber\\
 =&4\alpha_m\left(1+\theta u_c\right)\left(1+\ve_c(1-u_m)\right)\left(4\alpha_c\left(1+\ve_m(1-u_c)\right)-2\theta\alpha_m u_m^2\ve_m\right)\nonumber\\
 &-\left(-2\alpha_cu_c\ve_c+\theta\alpha_m u_m\left(1+\ve_c(1-u_m)\right)-2\alpha_m\left(1+\theta u_c\right)u_m\ve_m\right)^2.\nonumber
\end{align}
Observe that $d_P$ is quadratic with respect to $u_c$, and the   coefficient of $u_c^2$ is negative. Consequently, $d_P$ cannot attain its minimum inside  ${\cal D}_{cm}\cup((0,1)\times\{0\})$. It remains to ensure that  $d_P$ is positive on the sets $\{0\}\times[0,1]$ and $\{(u_c,u_m)\in(0,1]\times[0,1):\ u_c+u_m=1\}$. This the case if 
\begin{align}
 &\underset{y\in[0,1]}{\min}d_P(0,y)\nonumber\\
 =&\underset{y\in[0,1]}{\min}(4\alpha_m\left(1+\ve_c(1-y)\right)\left(4\alpha_c\left(1+\ve_m\right)-2\theta\alpha_m y^2\ve_m\right)-\alpha_m^2y^2\left(\theta \left(1+\ve_c(1-y)\right)-2\ve_m\right)^2)\nonumber\\ >&0\label{pos1}
\end{align}
and
\begin{align}
 \underset{y\in[0,1]}{\min}d_P(y,1-y)=&\underset{y\in[0,1]}{\min}\left(4\alpha_m\left(1+\theta y\right)\left(1+\ve_cy\right)\left(4\alpha_c\left(1+\ve_m(1-y)\right)-2\theta\alpha_m (1-y)^2\ve_m\right)\right.\nonumber\\
 &\left.\qquad\ -\left(-2\alpha_cy\ve_c+\theta\alpha_m (1-y)\left(1+\ve_cy\right)-2\alpha_m(1-y)\left(1+\theta y\right)\ve_m\right)^2\right)\nonumber\\>&0.\label{pos2}
\end{align}

By Sylvester's criterion,  $P$ is positive definite if and only if $\tilde P$ is positive definite and $\det(P)>0$.
Due to \cref{S1bnd}-\cref{S2bnd}, \cref{betacm}, and $g$ Lipschitz all functions involved in  \cref{MaP} are bounded and functions $\Cr{Sc},\Cr{Sm},\beta_c,\beta_m$ have positive lower bounds in ${\cal D}$. In particular, $\tilde P_{22}\geq 2\alpha_m>0$, so that $\tilde P$ is positive definite if and only if $\det(\tilde P)>0$. Further, we have
\begin{align*}
 \det(P)=&aD_h\Cr{Sm}\det(\tilde P)+\varphi\nonumber\\
 \geq&aD_h\det(\tilde P)+\varphi,
\end{align*}
where $\varphi:\overline{{\cal D}}\rightarrow\R$ is a bounded function ({recall that $g$ is Lipschitz}).
Thus, for sufficiently large $a$, $\det(P)>0$ holds  provided that $\det(\tilde P)>0$. 
Altogether, we conclude that conditions \cref{pos1,pos2} imply that matrix $P$ is positive for all triples {in} $\overline{{\cal D}}$. Assuming these conditions to be satisfied, let $\lambda(P)$  denote an eigenvalue of $P$. 
Then
\begin{align*}
  \frac{\det(P)}{\tr^2(P)}\leq\lambda(P) \leq\tr(P).
\end{align*}
In $\overline{{\cal D}}$, functions $\tr(P)$ and $\det(P)$ are bounded from above and from below, respectively, by some positive constants. Hence, we have positive lower and upper bounds for the eigenvalues of $P$. 
\end{proof}

Now we are ready to {state} our  existence result.
\begin{Theorem}\label{mainthm}
Let  $0<K_{cm}\leq K_{cn}\leq K_{mn}$ and $\alpha_c,\alpha_m,\theta,h_{max},\chi>0$ be some  constants which satisfy    \cref{pos1_,pos2_,defeps}.   Let coefficients $\tau_c$ and $\tau_m$ be as defined in \cref{tau3,tau4} and let Lipschitz functions $f_c,f_m,f_h: {\overline{{\cal D}}}\rightarrow\R${, with domain ${\cal D}$ as in \cref{domain},} be such that
\begin{subequations}\label{Assumpf}
\begin{alignat}{3}
 &f_i\geq 0&&\qquad \text{for }u_i=0,\quad i\in\{c,m\},\\
 &f_c+f_n\leq 0&&\qquad \text{for }u_c+u_m=1,\\
 &f_h\geq0&&\qquad\text{for }h=0.
\end{alignat} 
\end{subequations}
 Then for every given $J_{sum}\in L^{2}_{loc}([0,\infty);(L^2(\Omega))^n)$ and {$(u_{c0},u_{m0},h)\in
(L^{\infty}(\Omega))^2\times L^2(\Omega)$
such that $(u_{c0},u_{m0},h)(x)\in\overline{{\cal D}}$  for a.a. $x\in\Omega$} there exists a weak solution $(u_c,u_m,h):[0,\infty)\times\Omega\rightarrow{\overline{{\cal D}}}$ to system \cref{Equc}-\cref{Equn}, \cref{EqProt} under no-flux boundary conditions. This means that:
  \begin{align}
(u_c,u_m,h)\in L^2_{loc}([0,\infty);({H^1(\Omega)})^3),\quad \partial_t(u_c,u_m,h)\in L^2_{loc}([0,\infty);((H^1(\Omega))^3)'),  \label{weakreg}                                                                                                                                                     \end{align}
\begin{subequations}
  \begin{align}
   \left<\partial_t u_c,\varphi\right>=&-\int_{\Omega}\left(\frac{D_{cc}}{K_{cm}^2\Cr{Sc}}(u_c,u_m)\nabla(u_c\tau_c(u_c,h))+\frac{D_{cm}}{K_{cm}^2\Cr{Sc}}(u_c,u_m)\nabla(u_m\tau_m(u_c,u_m))-J_{sum}u_c\right)\cdot\nabla\varphi\,dx \nonumber\\
   &+\int_{\Omega}\varphi f_c(u_c,{ u_n},h)\,dx,\label{weakuc}\\ 
   \left<\partial_t u_m,\varphi\right>=&-\int_{\Omega}\left(\frac{D_{mc}}{K_{cm}^2\Cr{Sc}}(u_c,u_m)\nabla(u_c\tau_c(u_c,h))+\frac{D_{mm}}{K_{cm}^2\Cr{Sc}}(u_c,u_m)\nabla(u_m\tau_m(u_c,u_m))-J_{sum}u_m\right)\cdot\nabla\varphi\,dx \nonumber\\
   &+\int_{\Omega}\varphi f_m(u_c,u_m,h)\,dx,\label{weakum}\\
   \left<\partial_t h,\varphi\right>=&-\int_{\Omega}D_h\nabla h\cdot\nabla \varphi\, dx +\int_{\Omega}\varphi f_h(u_c,u_m,h)\,dx\label{weakuh}
  \end{align}
  \end{subequations}
a.e. in $(0,\infty)$ for all $\varphi\in H^1(\Omega)^3$, and the initial conditions for each variable are satisfied in $L^2(\Omega)$-sense.

\end{Theorem}
{
\begin{Remark}
 Using the assumptions on the  coefficients of system \cref{Equc}-\cref{Equn}, \cref{EqProt},  one readily verifies that if $(u_c,u_m,h)$ satisfies \cref{weakreg} and  $(u_c,u_m,h)\in{\cal D}$ a.e., then the gradients $\nabla(u_c\tau_c(u_c,u_m))$ and $\nabla(u_m\tau_m(u_c,u_m))$ as well as the fluxes  in  \cref{weakuc,weakum} belong to $L^2_{loc}([0,\infty);({L^2(\Omega)})^N)$.
\end{Remark}
\noindent{\it Proof of \cref{mainthm} (sketch).} 
 We rely on the theory of weak solvability  for cross diffusion systems which was developed in   \cite{Juengel2015}. 
We recall that the main tool of the method in \cite{Juengel2015} is a suitable  entropy density. Here we use function 
$E:{\cal D}\rightarrow\R$ previously   defined in  \cref{Entrop}. This function is the sum of the standard logarithmic entropy $L$ for variables $u_1$ and $u_2$  and the quadratic $ah^2$. Combining the corresponding properties of the logarithmic entropy and our previous calculations in this Section, we obtain:
\begin{enumerate}
\item[H1:]\label{HH1} $E\in C^2({\cal D})$ and is convex and bounded below by 
\begin{align*}
 \underset{{\cal D}}{\min}E=E\left(\frac{1}{3},\frac{1}{3},\frac{1}{3},0\right).
\end{align*}
The derivative 
\begin{align}
 DE:{\cal D}\rightarrow\R^3,\qquad DE=(\ln(u_c)-\ln(u_n),\ln(u_m)-\ln(u_n),ah)
\end{align}
is invertible, the inverse being
\begin{align}
 (DE)^{(-1)}(z_1,z_2,z_3)=\left(\frac{e^{z_1}}{1+e^{z_1}+e^{z_2}},\frac{e^{z_2}}{1+e^{z_1}+e^{z_2}},\frac{1}{a}z_3\right).\nonumber
\end{align}
\item[H2':] For sufficiently large $a$, left multiplying the diffusion matrix   $M$ by $D^2E$ yields a uniformly positive definite matrix (see \cref{LemmaEll}).
\item[H2'':] The diffusion matrix   $M$ is uniformly bounded in ${\cal D}$ (compare \cref{DiffA,line3}, recall that $g$ is Lipschitz).
\item[H3:]\label{HH3} $M$ and $f_i$, $i\in\{c,m,h\}$, are continuous mappings and there exists a constant $\Cr{CfDE}$ such that
\begin{align}
 DE\cdot (f_c,f_m,f_h)\leq\Cl[C]{CfDE}\left(1+E-\underset{{\cal D}}{\min}E\right)\qquad\text{in }{\cal D}.\label{H3}
\end{align}
\end{enumerate}
Estimate \cref{H3} is a standard consequence of the assumptions \cref{Assumpf}. Indeed, since $z\mapsto z\ln z$ is negative and bounded below in $(0,1)$, the part originating from the logarithmic entropy satisfies 
\begin{align}
  (\partial_{u_c}L,\partial_{u_m}L)\cdot (f_c,f_m)=&f_c\ln(u_c)+f_m\ln(u_m)+(-f_c-f_m)\ln(u_n)\nonumber\\
  \leq&-\|\nabla f_c\|_{L^{\infty}({\cal D})}u_c\ln(u_c)-\|\nabla f_m\|_{L^{\infty}({\cal D})}u_m\ln(u_m)\nonumber\\
  &-\left(\|\nabla f_c\|_{L^{\infty}({\cal D})}+\|\nabla f_m\|_{L^{\infty}({\cal D})}\right)(1-(u_c+u_m))\ln(1-(u_c+u_m))\nonumber\\
  \leq &\Cl[C]{Cfi1}\qquad\text{in }{\cal D}\label{estDEf1}
\end{align}
for some constant $\Cr{Cfi1}>0$.
Further, since $f_h$ is Lipschitz, we also have
\begin{align}
 \partial_h E f_h\leq\Cl[C]{Ch}(1+ah^2)\qquad\text{in }{\cal D}\label{estDEf2}
\end{align}
for some constant $\Cr{Ch}>0$. Adding \cref{estDEf1,estDEf2} together, we obtain
\begin{align}
 DE\cdot (f_c,f_m,f_h)\leq&\Cl[C]{Cf3}(1+a h^2)\nonumber\\
 \leq&\Cr{CfDE}\left(1+L-\underset{{\cal D}}{\min}L+a h^2\right)\nonumber\\
 =&\Cr{CfDE}\left(1+E-\underset{{\cal D}}{\min}E\right)\qquad\text{in }{\cal D}\nonumber
\end{align}
since obviously $\min E=\min L$.
\\\indent 
 The main result of \cite{Juengel2015}, Theorem 2 on existence of bounded weak solutions, cannot be directly applied in our case for two reasons: firstly, apart from diffusion and reaction,  equations for $u_c$ and $u_m$ also involve transport in the direction of a given vector-valued function $J_{sum}$;  secondly, the  domain ${\cal D}$ is not bounded. Were it not for these differences, the above  properties H1-H3 would correspond to the hypotheses H1-H3 in \cite{Juengel2015}.   
 Still, the proof of existence of weak solutions can be carried out very similar to the proofs presented in \cite{Juengel2015}. The latter go through the following steps:  approximation of the time derivative using the implicit Euler scheme, regularisation of the diffusion operator  by adding a higher order differential  operator such as  $\epsilon(I+(-\Delta)^m)$ for $m>N/2$ and small $\epsilon$, solving the linearised approximation problem using the Lax-Milgram lemma, solving the nonlinear approximation problem using the Leray-Schauder theorem, establishing uniform estimates, and, finally, using the compactness method in order to pass to the limit and solve the original problem. 
 Our case can be handled in the very same way. Indeed, on the one hand  the linear  transport  along $J_{sum}$ is subordinate to diffusion, so that it, e.g. does not hinder the derivation of estimates. On the other hand even though ${\cal D}$ is not bounded, we actually know that for bounded $h_0$ the $h$-component of a solution is a priori bounded on all finite time cylinders. This is a consequence of the standard theory of semilinear parabolic PDEs. General $h_0\in L^2(\Omega)$ can be regularised and a corresponding solution obtained by means of yet one more limit procedure. \\
 We omit further details and refer the interested reader to \cite{Juengel2015} where the complete proofs for very similar cases can be found.
\qed
}

\section{Numerical study}\label{SecNum}

We perform numerical simulations of the system \eqref{model2} supplemented with the PDE \eqref{EqProt} for the evolution of proton concentration $h$, endowed with no-flux boundary conditions. For simplicity we choose $J_{sum}=0$. We also perform a nondimensionalisation of the model and use it in our simulations (for its concrete form and for the employed parameters refer to the Appendix). The initial conditions are as follows:

\begin{subequations}\label{eq:ICs_mpm}
	\begin{align}
	u_c(0,\mathbf{ x}) &= 0.05   \left( e^{\frac{-(x-500)^2- (y-500)^2}{2(25)^2}}    + e^{\frac{-(x-600)^2- (y-500)^2}{2 (20)^2}} + e^{\frac{-(x-300)^2- (y-400)^2}{2(10)^2}}\right),\label{eq:IC-c}\\
	h(0,\mathbf{ x}) & = 10^{-7}e^{\frac{-(x-500)^2- (y-500)^2}{2 (15)^2}} + 10^{-7}e^{\frac{-(x-600)^2- (y-500)^2}{2(10)^2}} + 10^{-6.4} e^{\frac{-(x-300)^2- (y-400)^2}{2(7.5)^2}},\label{eq:IC-h}\\
	u_n(0,\mathbf{ x}) & = 0.9 \left(e^{\frac{-(x-500)^2- (y-500)^2}{2 (5)^2}} + 
	e^{\frac{-(x-600)^2- (y-500)^2}{2(2)^2}} + e^{\frac{-(x-300)^2- (y-400)^2}{2(1)^2}} \right),\label{eq:IC-n}\\
	u_m(0,\mathbf{ x}) &= 1 - u_c(0,\mathbf{ x}) - u_n(0,\mathbf{ x}) \quad \text{with} \quad u_m(0,500,500) = u_m(0,600,500) = u_m(0,300,400)=0.\label{eq:IC-m}
	\end{align}
\end{subequations}

\noindent
These are illustrated in Figure \ref{fig:IC_set1_mpm}. The problem is set in a square $[0,1000]\times [0,1000]$ (in $\mu m$), corresponding to the size of large pseudopalisades \cite{Brat2004}. For the discretisation we use a method of lines approach. Thereby, the diffusion terms in all involved PDEs are computed by using a standard central difference scheme, while a first order upwind scheme is employed for the advection terms occurring in the equations for glioma and normal tissue. For the acidity equation we discretise time upon using an implicit-explicit (IMEX) method, with forward and backward Euler schemes for the diffusion and reaction terms, respectively. The glioma and normal tissue PDEs are discretised in time by an explicit Euler method.

\begin{figure}[!htbp]
	\centering
		\begin{subfigure}[b]{0.25\textwidth}
			{\includegraphics[width=1\linewidth]{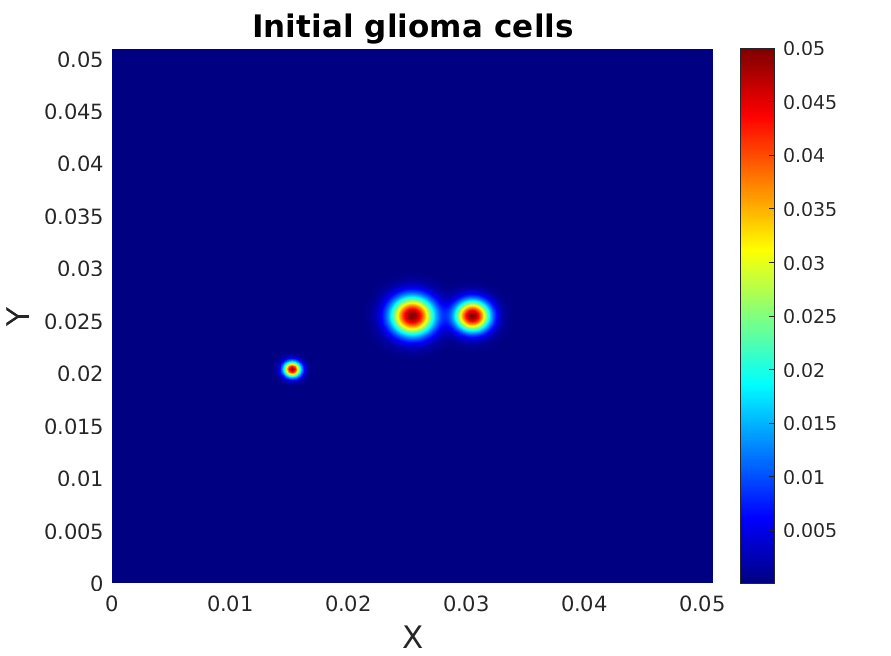}\subcaption{Glioma cells}}
		\end{subfigure}\hfill
		\begin{subfigure}[b]{0.25\textwidth}
			{\includegraphics[width=1\linewidth]{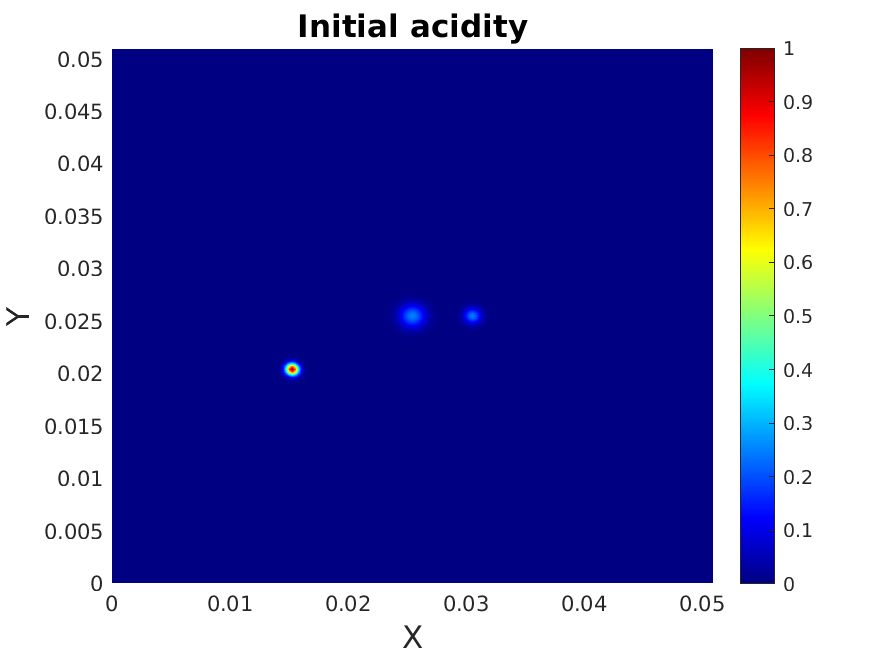}\subcaption{Acidity}}
		\end{subfigure}\hfill
		\begin{subfigure}[b]{0.25\textwidth}
			{\includegraphics[width=1\linewidth]{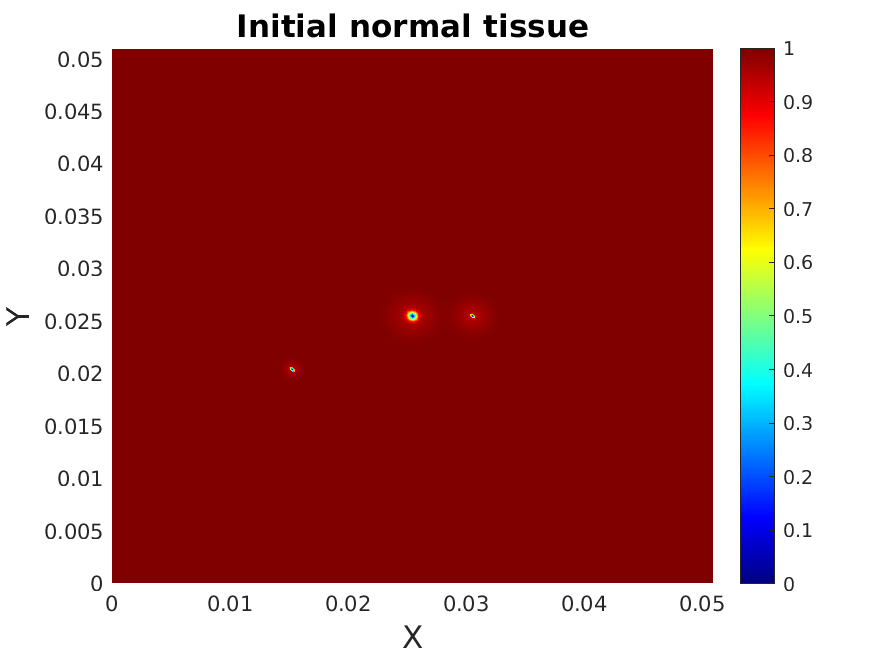}\subcaption{Normal tissue}}
		\end{subfigure}\hfill
		\begin{subfigure}[b]{0.25\textwidth}
			{\includegraphics[width=1\linewidth]{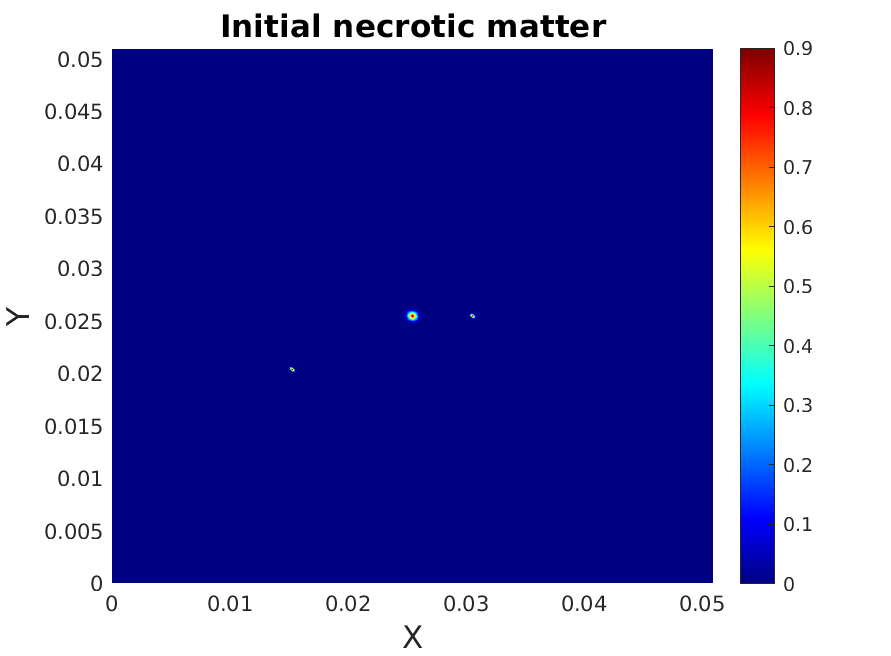}\subcaption{Necrotic matter}}%
		\end{subfigure}
	\caption{Initial conditions \eqref{eq:ICs_mpm}.}
	\label{fig:IC_set1_mpm}
\end{figure}

Figure \ref{fig:simulations_mpm} shows the computed volume fractions of glioma (first column), normal tissue (3rd column), necrotic matter (last column), and acidity concentration (second column) at several times, in a total time span which is relevant for pseudopalisade formation. The typical garland-like structure of glioma pseudopalisades is clearly visible. The tumor cells encircle a highly acidic and necrotic region, while outwards, beyond the glioma ring, the normal tissue remains non-depleted and the acidity decreases. 

\begin{figure}[!htbp]
	\centering
		\begin{minipage}[hstb]{.24\linewidth}
		\raisebox{1.2cm}{\rotatebox[origin=t]{90}{30 days}}{\includegraphics[width=1\linewidth, height = 3cm]{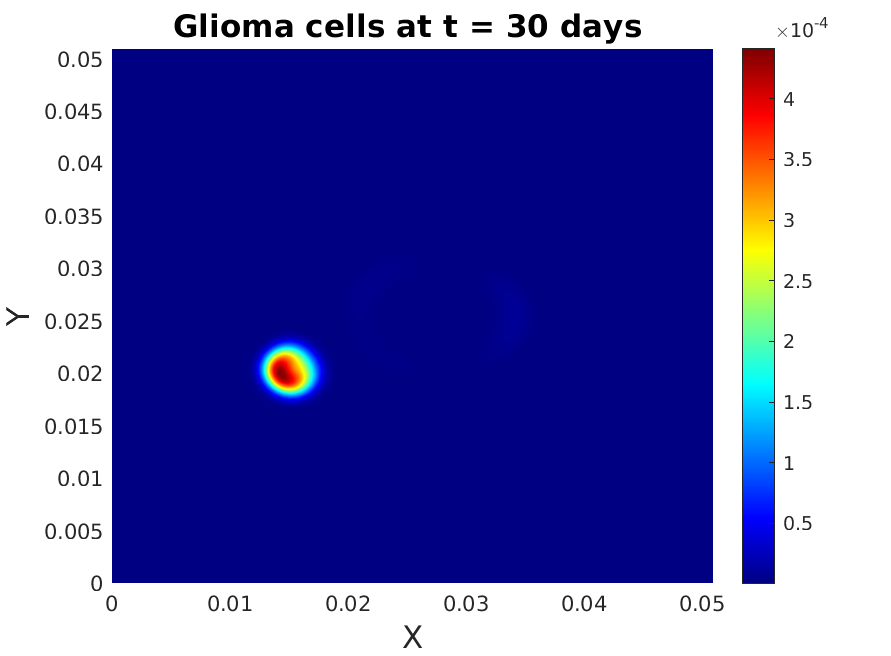}}\\
		\raisebox{1.2cm}{\rotatebox[origin=t]{90}{90 days}}{\includegraphics[width=1\linewidth, height = 3cm]{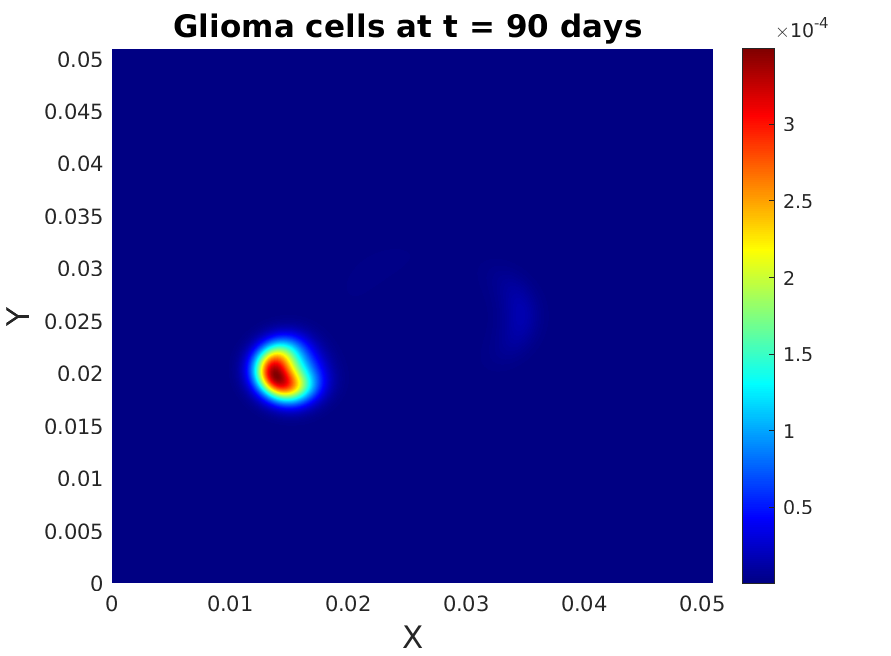}}\\
		\raisebox{1.2cm}{\rotatebox[origin=t]{90}{150 days}}{\includegraphics[width=1\linewidth, height = 3cm]{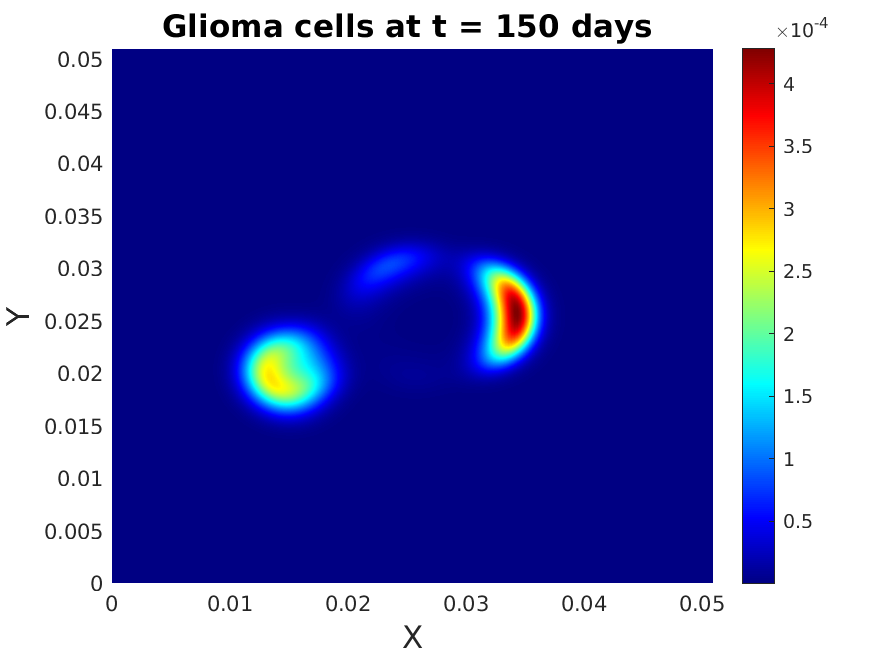}}\\
		\raisebox{1.2cm}{\rotatebox[origin=t]{90}{210 days}}{\includegraphics[width=1\linewidth, height = 3cm]{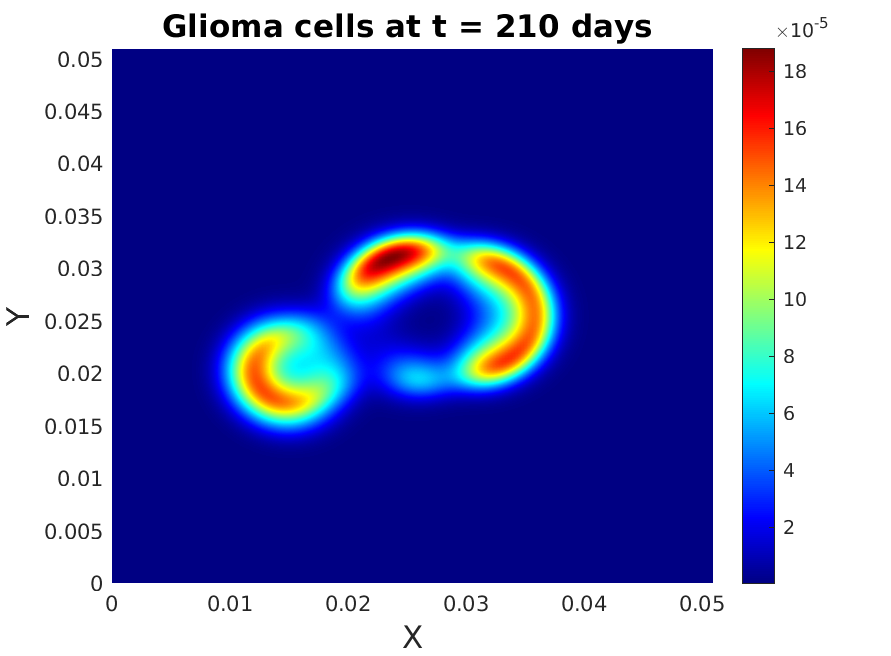}}\\
		\raisebox{1.2cm}{\rotatebox[origin=t]{90}{300 days}}{\includegraphics[width=1\linewidth, height = 3cm]{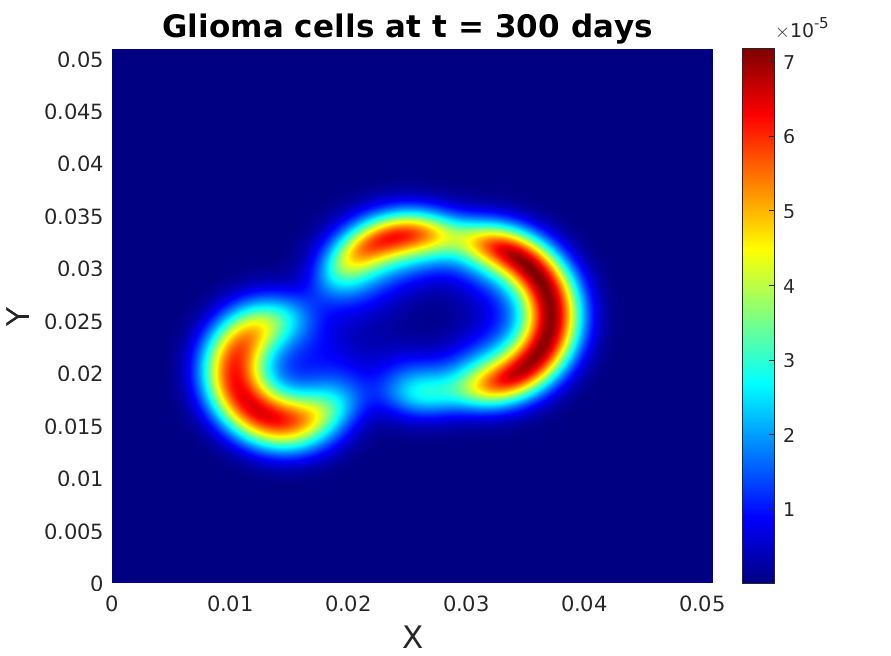}}\\
		\raisebox{1.2cm}{\rotatebox[origin=t]{90}{360days}}{\includegraphics[width=1\linewidth, height = 3cm]{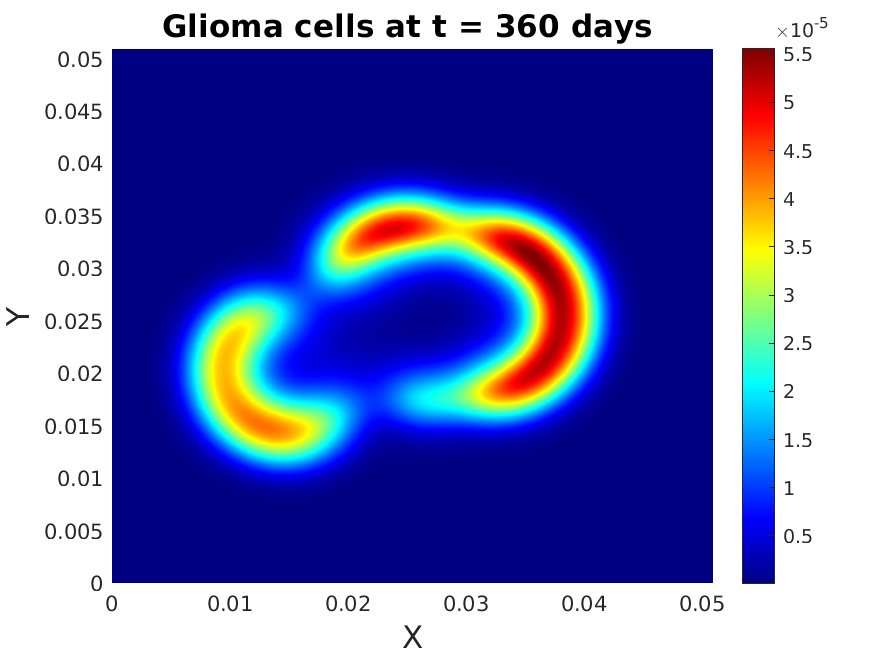}}\\
		\raisebox{1.2cm}{\rotatebox[origin=t]{90}{450 days}}{\includegraphics[width=1\linewidth, height = 3cm]{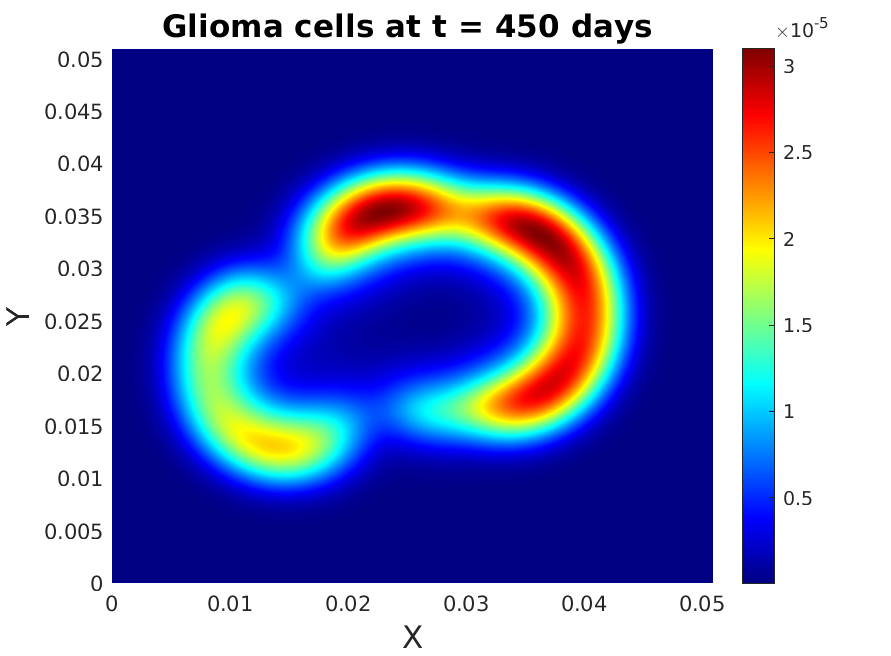}}%
		\subcaption{Glioma cells}
	\end{minipage}%
	\hspace{0.2cm}
	\begin{minipage}[hstb]{.24\linewidth}
		{\includegraphics[width=1\linewidth, height = 3cm]{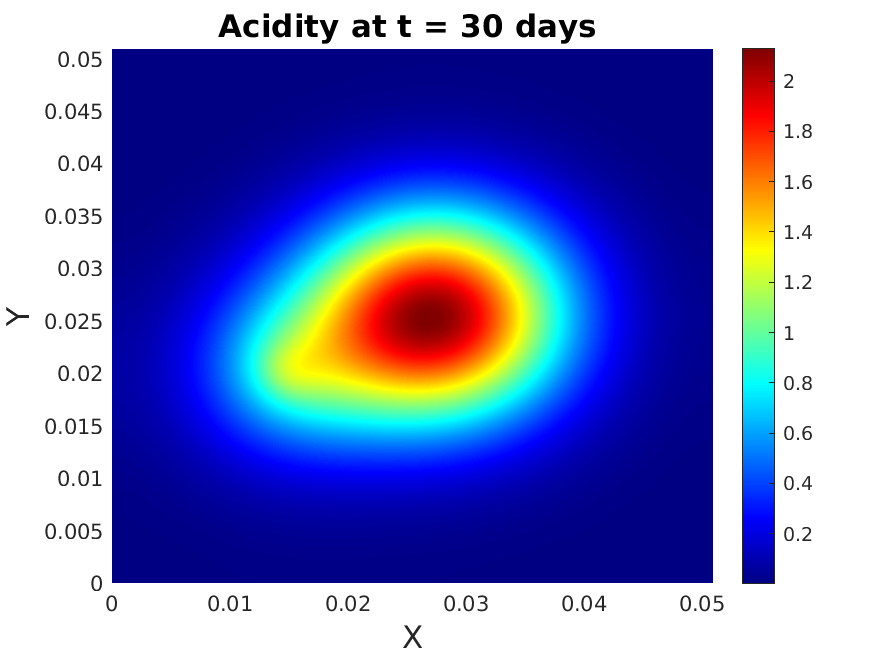}}\\
		{\includegraphics[width=1\linewidth, height = 3cm]{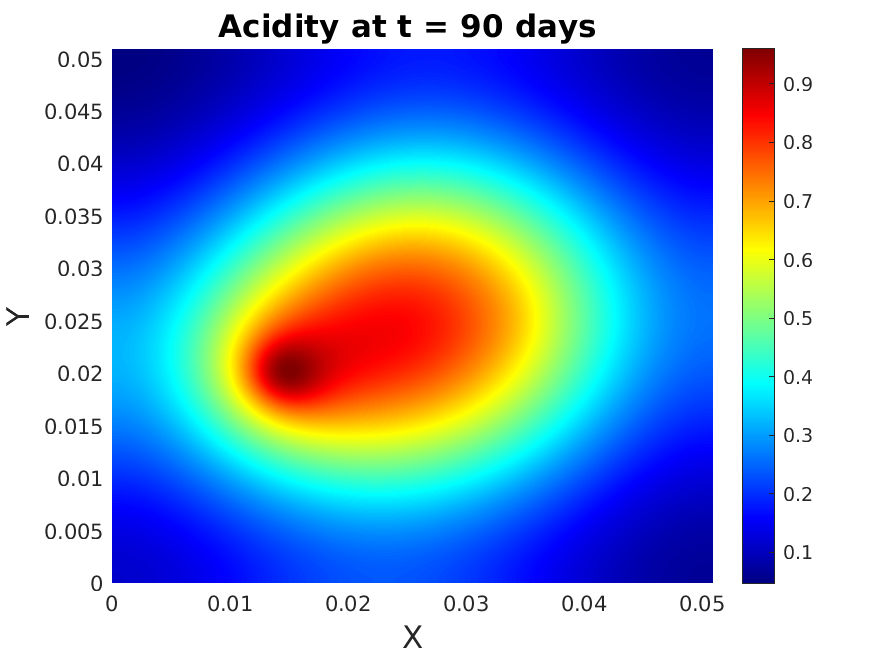}}\\
		{\includegraphics[width=1\linewidth, height = 3cm]{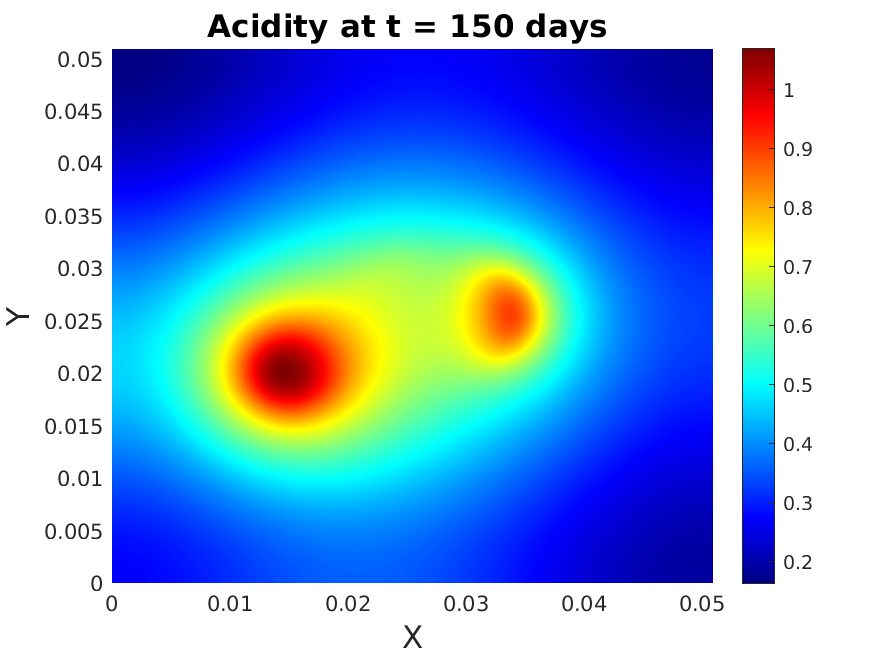}}\\
		{\includegraphics[width=1\linewidth, height = 3cm]{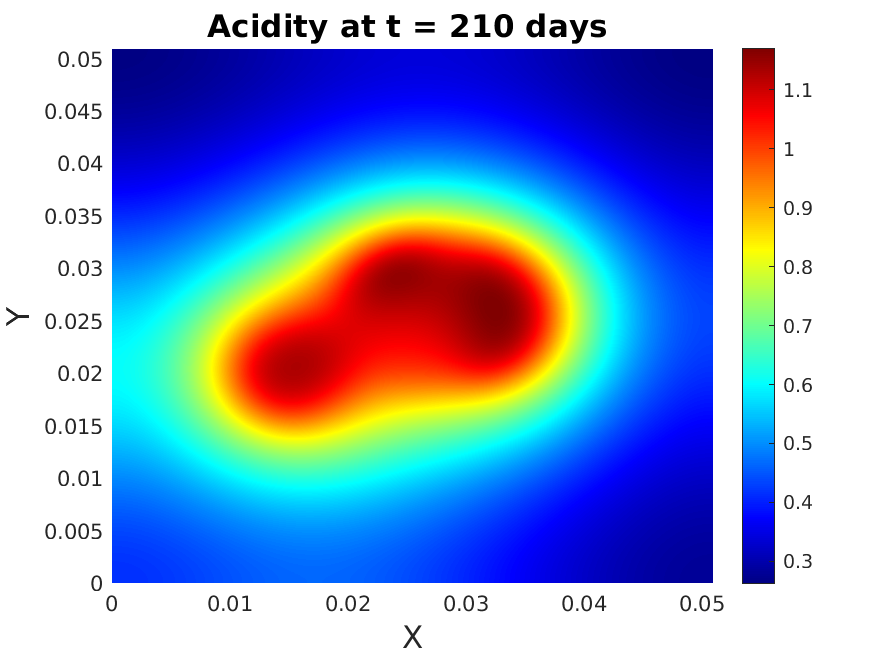}}\\
		{\includegraphics[width=1\linewidth, height = 3cm]{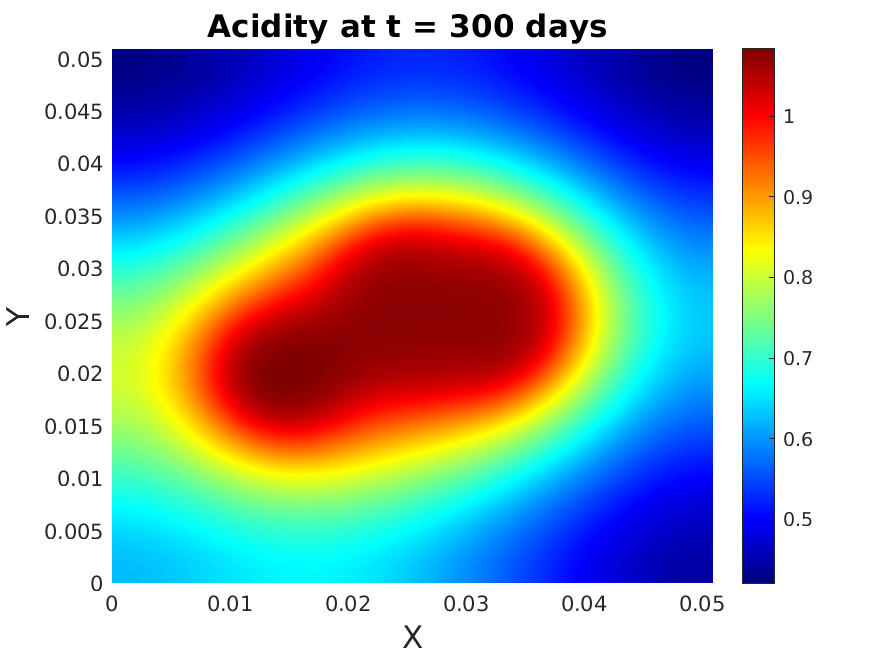}}\\
		{\includegraphics[width=1\linewidth, height = 3cm]{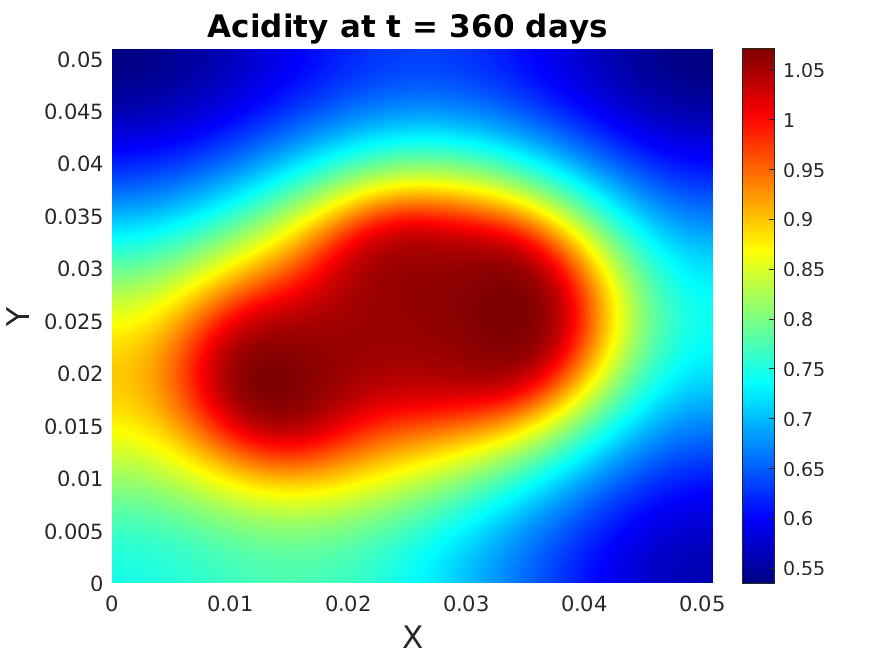}}\\
		{\includegraphics[width=1\linewidth, height = 3cm]{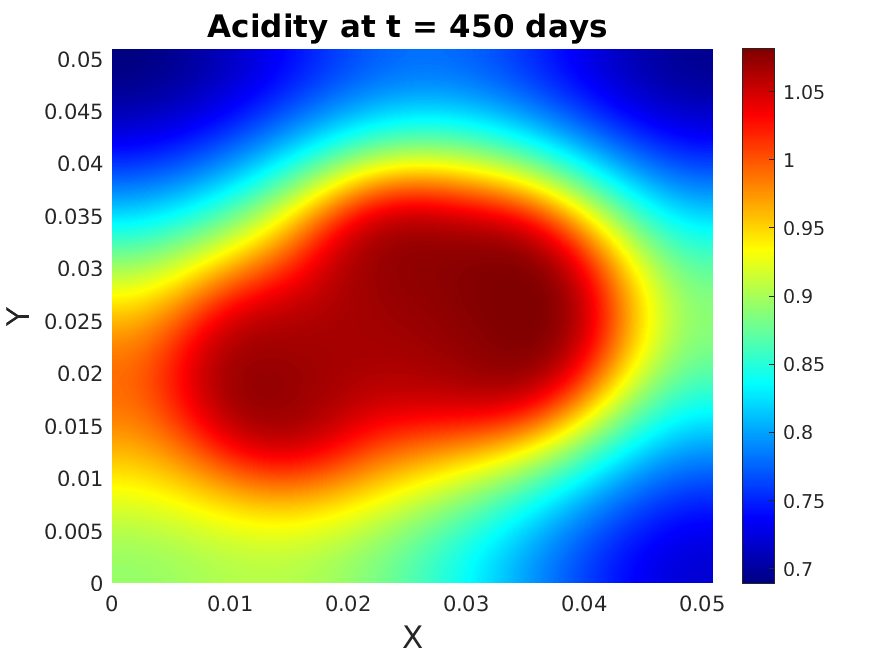}}%
		\subcaption{Acidity}
	\end{minipage}%
	\hspace{0.01cm}
	\begin{minipage}[hstb]{.24\linewidth}
		{\includegraphics[width=1\linewidth, height = 3cm]{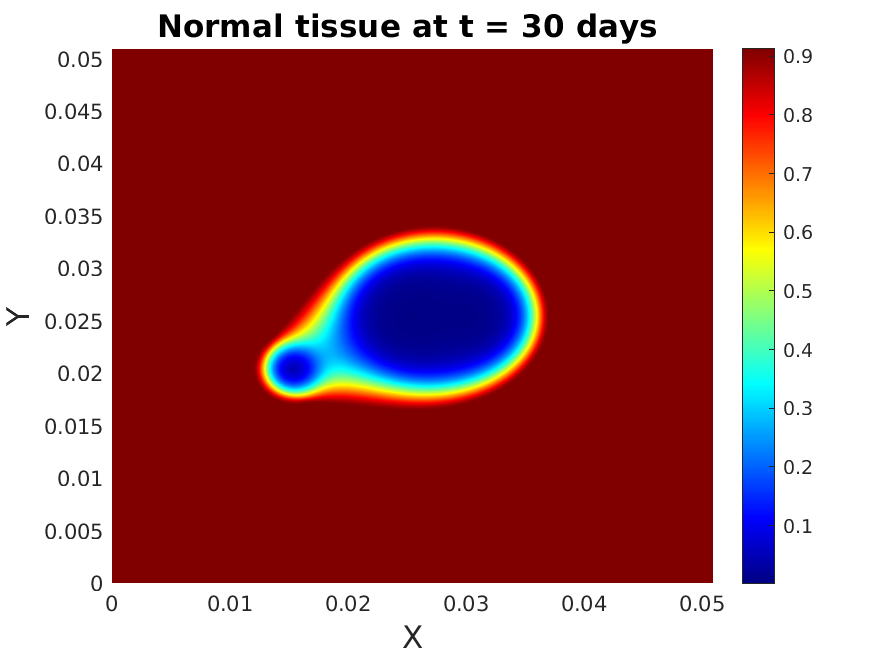}}\\
		{\includegraphics[width=1\linewidth, height = 3cm]{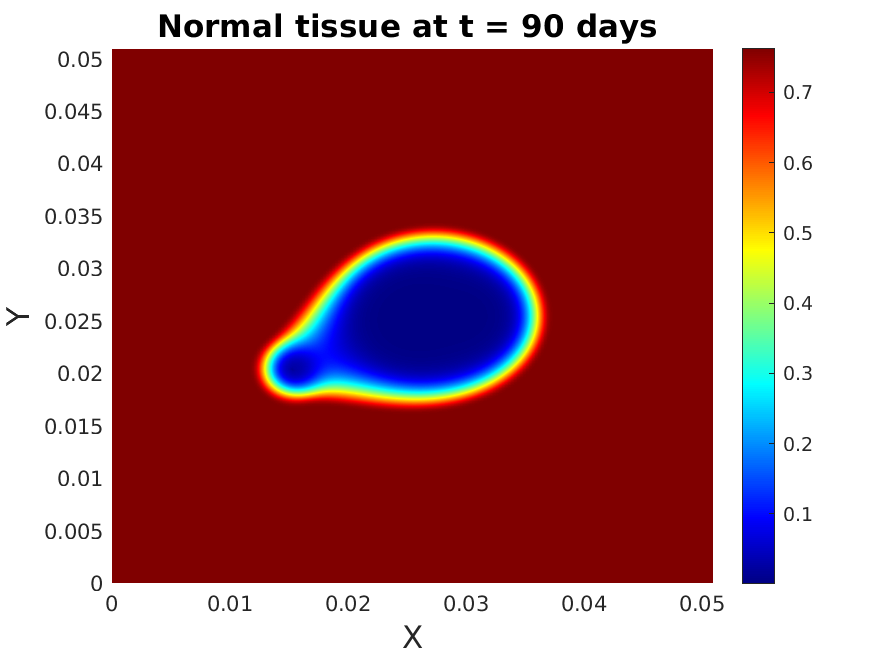}}\\
		{\includegraphics[width=1\linewidth, height = 3cm]{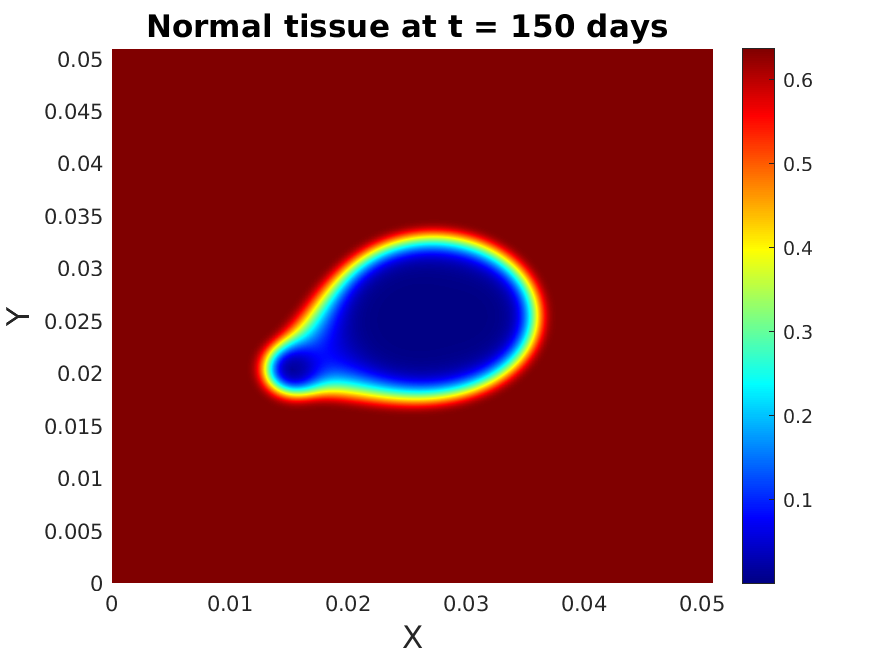}}\\
		{\includegraphics[width=1\linewidth, height = 3cm]{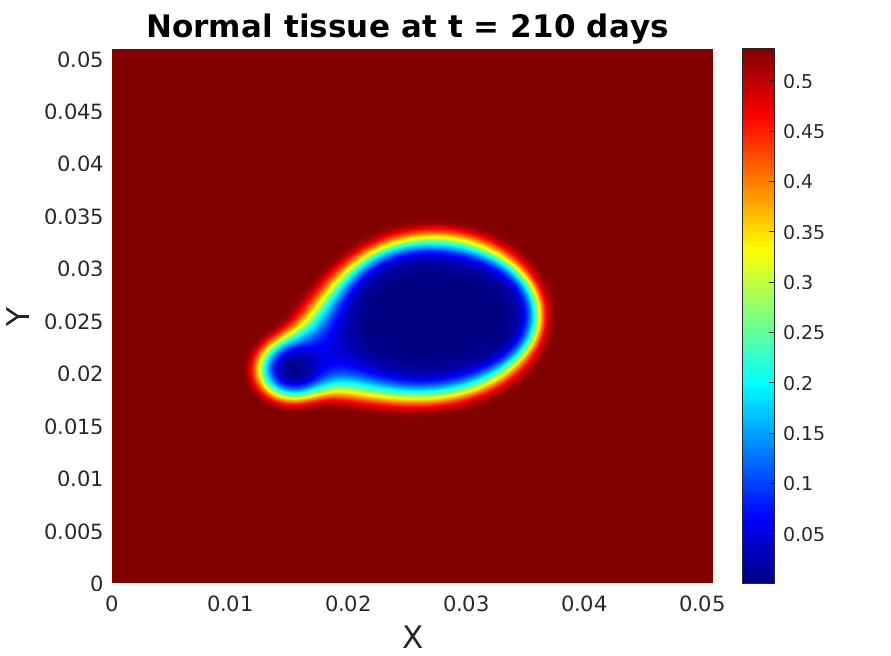}}\\
		{\includegraphics[width=1\linewidth, height = 3cm]{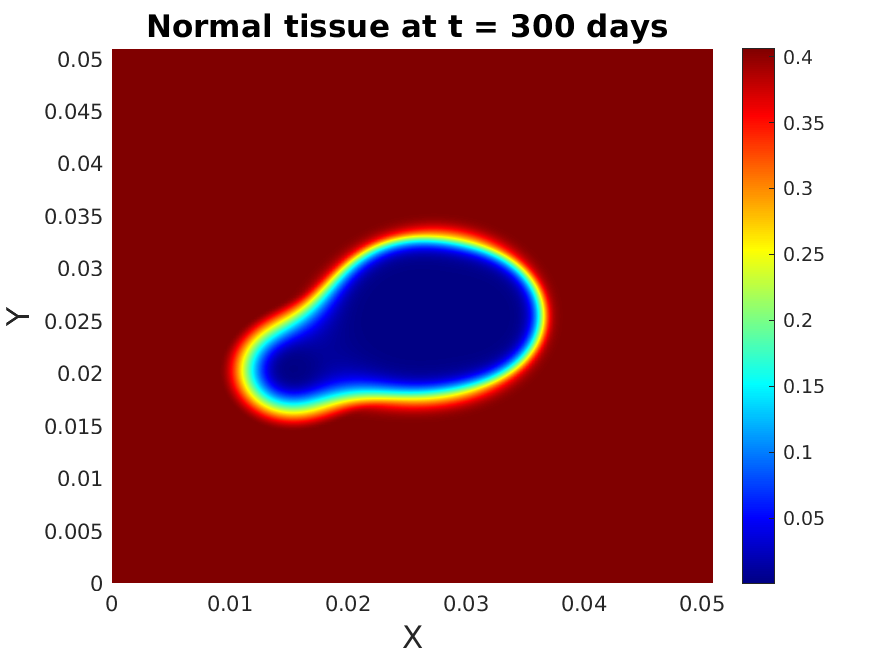}}\\
		{\includegraphics[width=1\linewidth, height = 3cm]{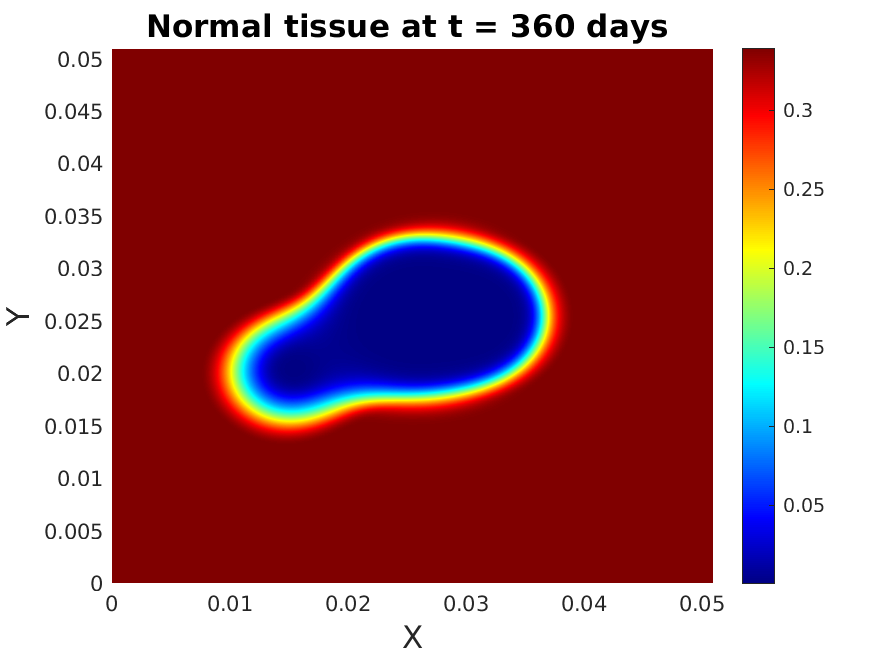}}\\
		{\includegraphics[width=1\linewidth, height = 3cm]{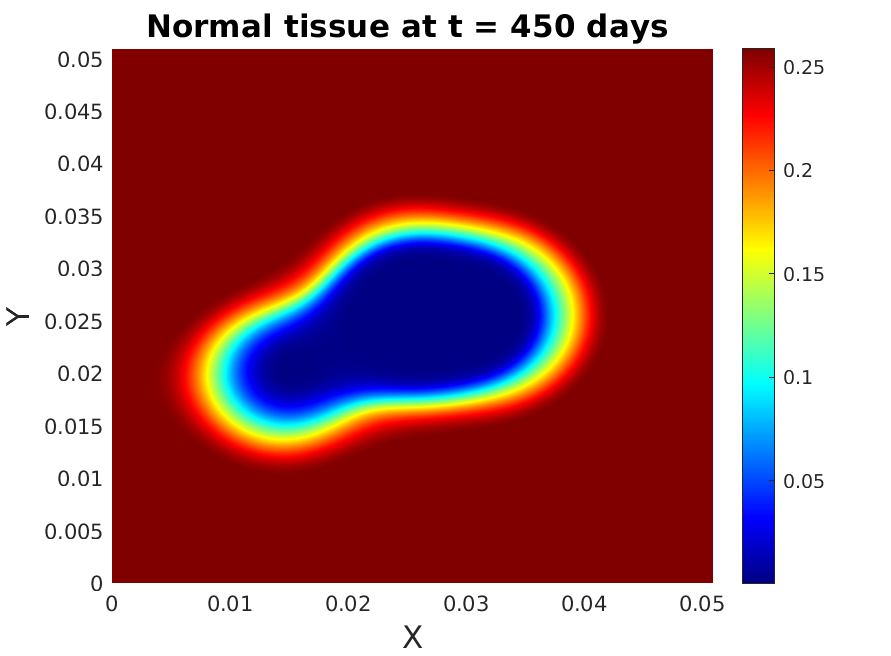}}%
		\subcaption{Normal tissue}
	\end{minipage}%
	\hspace{0.01cm}
	\begin{minipage}[hstb]{.24\linewidth}
		{\includegraphics[width=1\linewidth, height = 3cm]{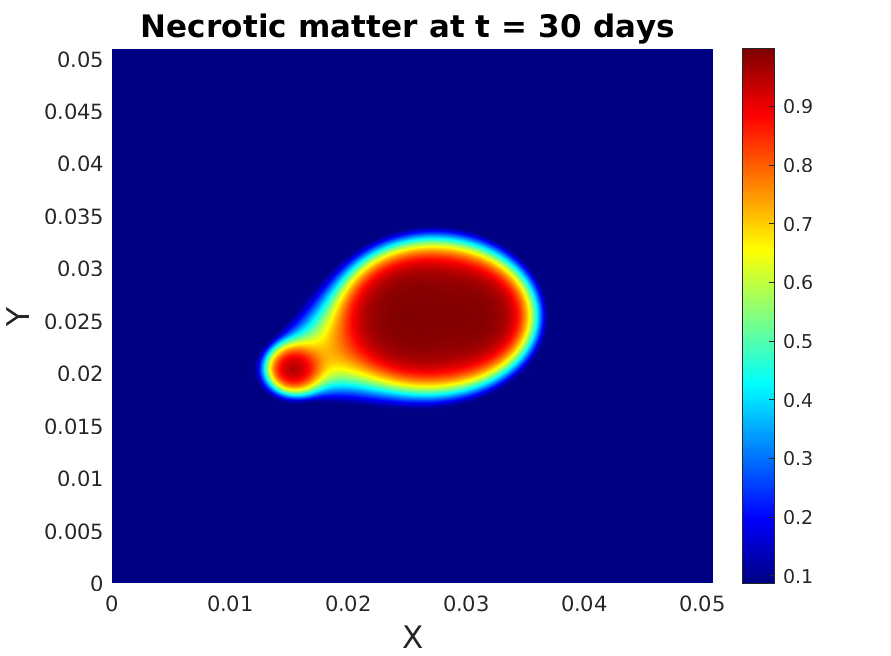}}\\
		{\includegraphics[width=1\linewidth, height = 3cm]{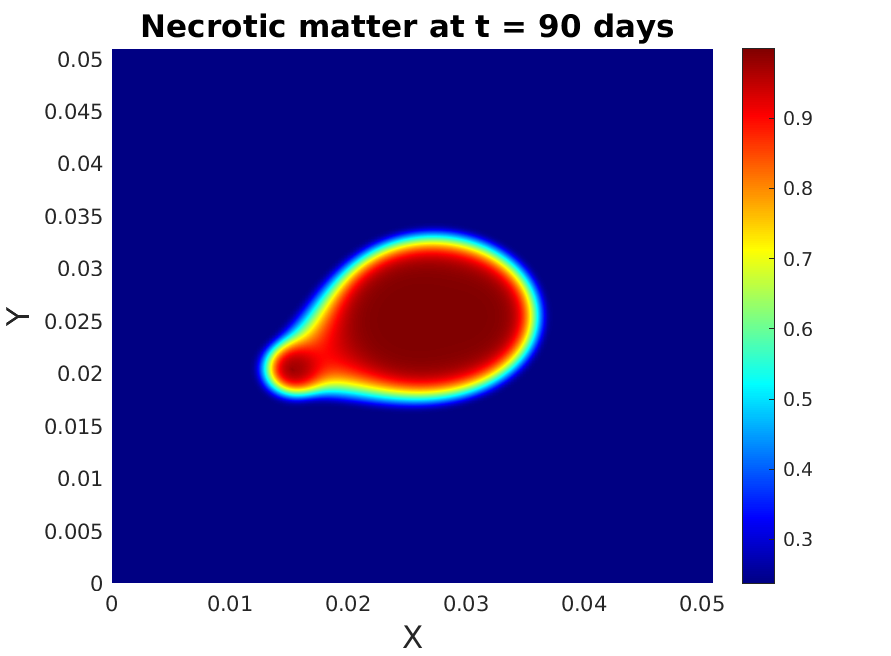}}\\
		{\includegraphics[width=1\linewidth, height = 3cm]{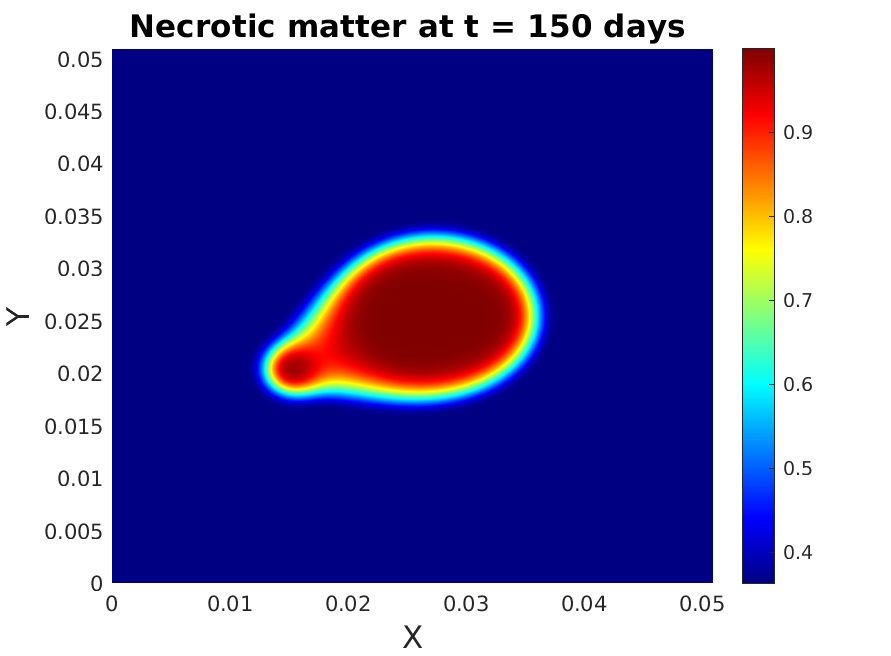}}\\
		{\includegraphics[width=1\linewidth, height = 3cm]{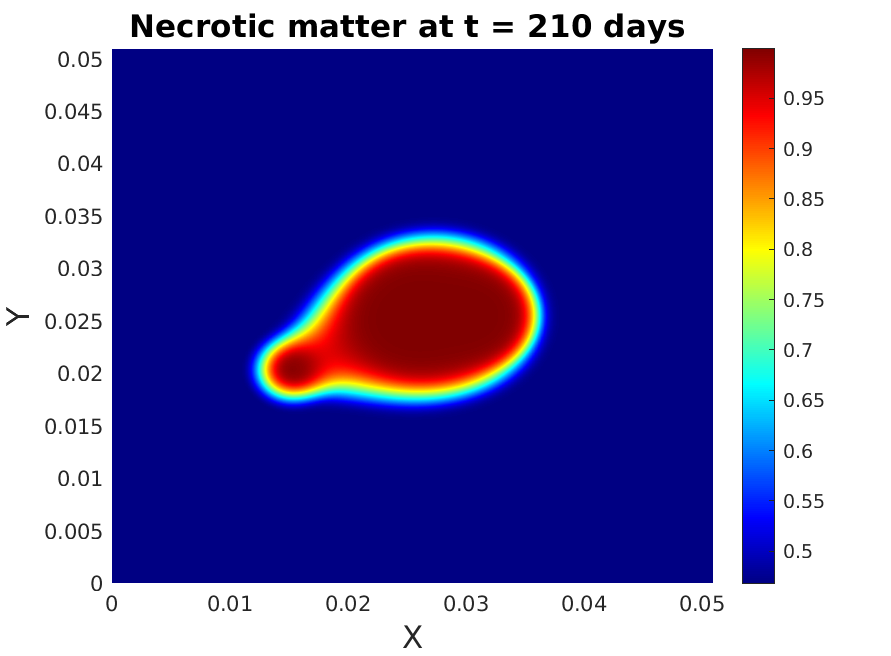}}\\
		{\includegraphics[width=1\linewidth, height = 3cm]{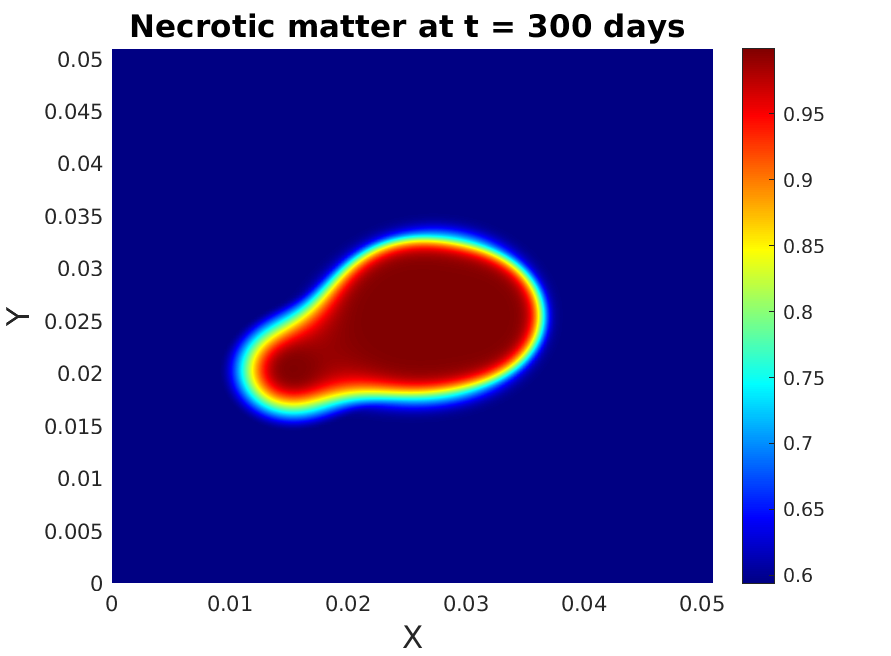}}\\
		{\includegraphics[width=1\linewidth, height = 3cm]{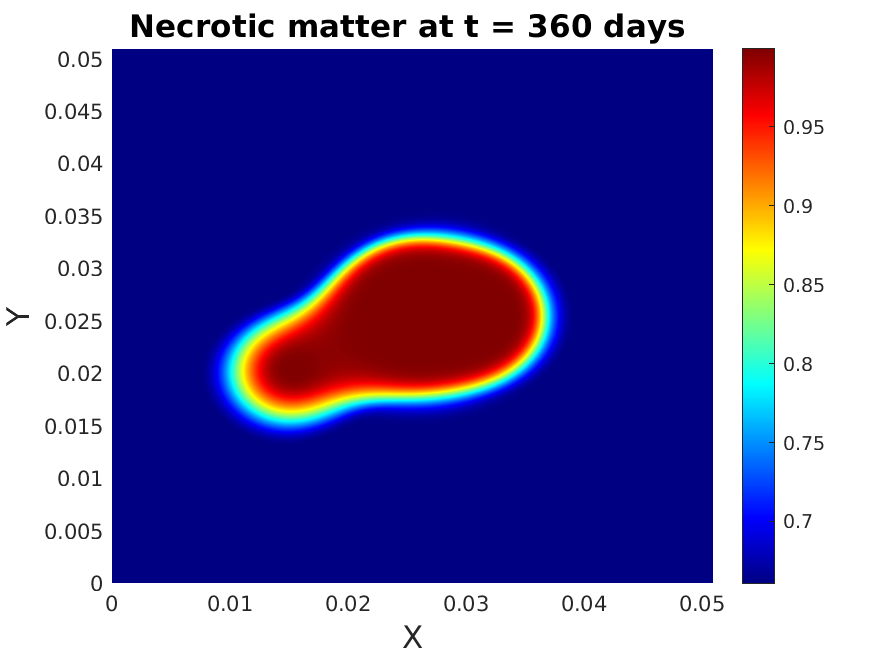}}\\
		{\includegraphics[width=1\linewidth, height = 3cm]{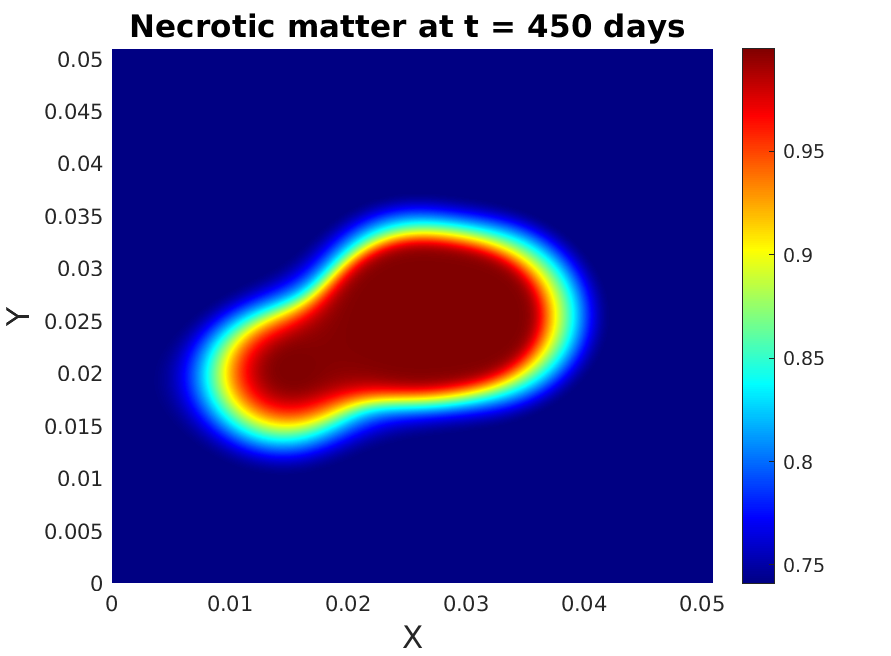}}%
		\subcaption{Necrotic matter}
	\end{minipage}%
	\caption{Computed $u_c$, $h$, $u_m$, and $u_n$ from \eqref{model2} and \eqref{EqProt} with initial conditions \eqref{eq:ICs_mpm} and with no-flux boundary conditions.}
	\label{fig:simulations_mpm}
\end{figure}

To investigate the effect of acidity we compare the previous results with those obtained for $\chi=0$, i.e. in the absence of additional pressure on glioma cells due to acidosis, recall \eqref{tau3}. The corresponding plots are shown in Figure \ref{fig:comp-acid}. The effect of repellent pH-taxis can be seen in the difference between the solution components in two situations (with $\chi>0$, with pH-taxis and acid-influenced motility, respectively with $\chi=0$). The pseudopalisades develop faster and get thicker and more extensive in the former case, accompanied by enhanced acidity in the proximity of higher cancer cell densities and enhanced normal tissue depletion inside the ring-shaped structures and around the glioma aggregates, which also triggers a higher volume fraction for the necrotic matter. These observations are in line with those obtained in \cite{kumar2020multiscale} by another modeling approach: the repellent pH-taxis is not the driving factor of pseudopalisade formation, but it leads to wider such structures.

\begin{figure}[!htbp]
	\centering
	\begin{minipage}[hstb]{.24\linewidth}
		\raisebox{1.2cm}{\rotatebox[origin=t]{90}{30 days}}{\includegraphics[width=1\linewidth, height = 3cm]{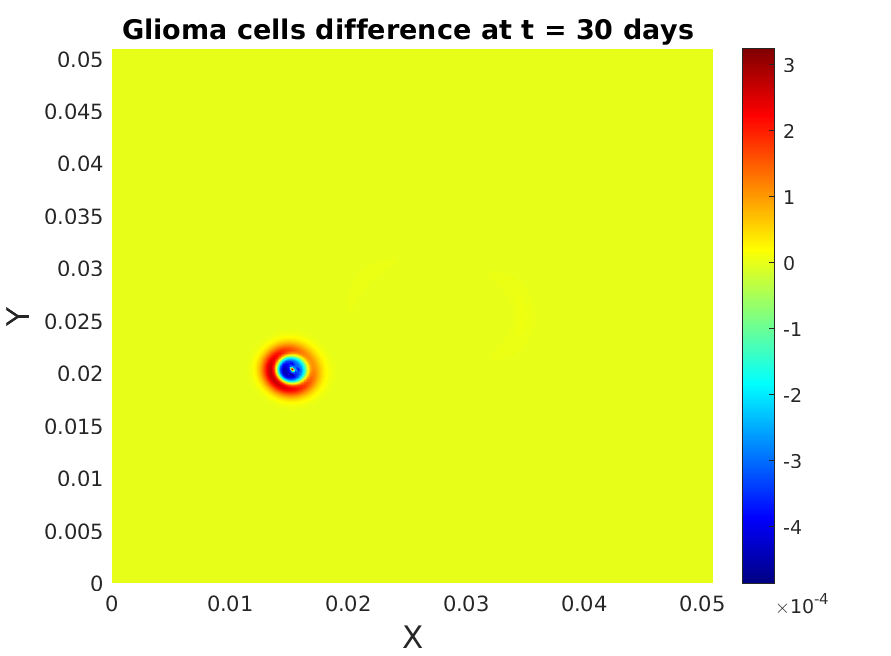}}\\
		\raisebox{1.2cm}{\rotatebox[origin=t]{90}{90 days}}{\includegraphics[width=1\linewidth, height = 3cm]{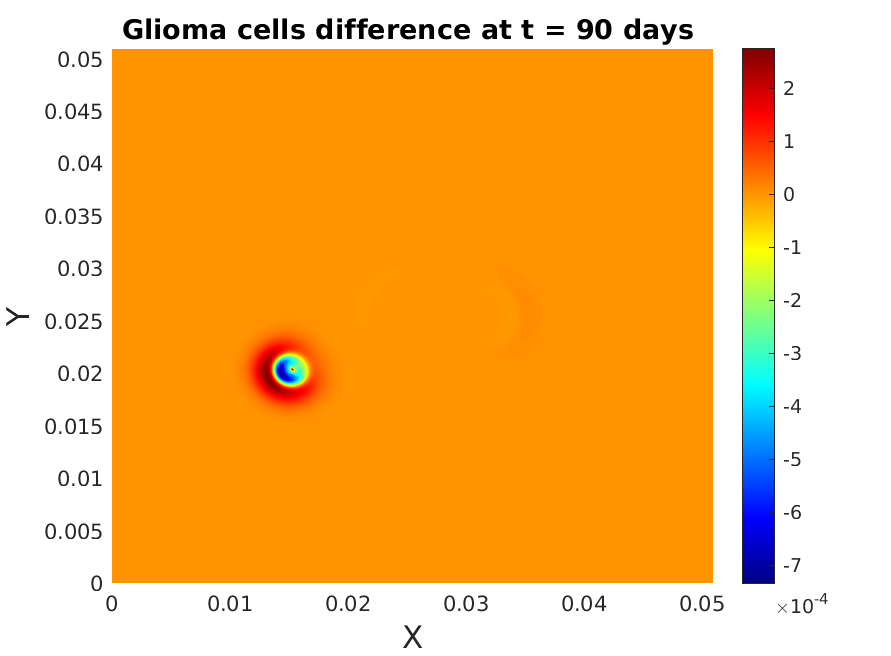}}\\
		\raisebox{1.2cm}{\rotatebox[origin=t]{90}{150 days}}{\includegraphics[width=1\linewidth, height = 3cm]{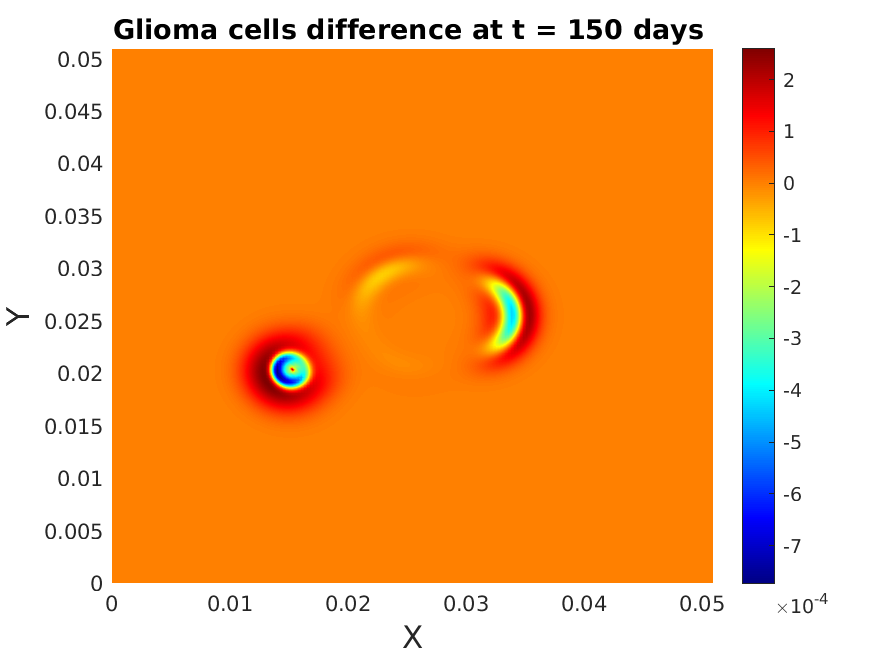}}\\
		\raisebox{1.2cm}{\rotatebox[origin=t]{90}{210 days}}{\includegraphics[width=1\linewidth, height = 3cm]{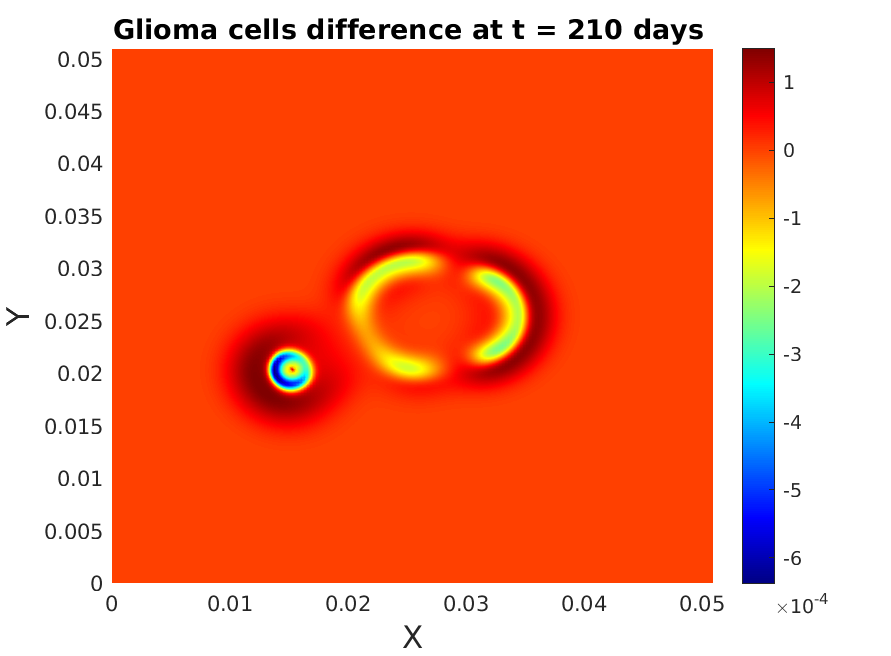}}\\
		\raisebox{1.2cm}{\rotatebox[origin=t]{90}{300 days}}{\includegraphics[width=1\linewidth, height = 3cm]{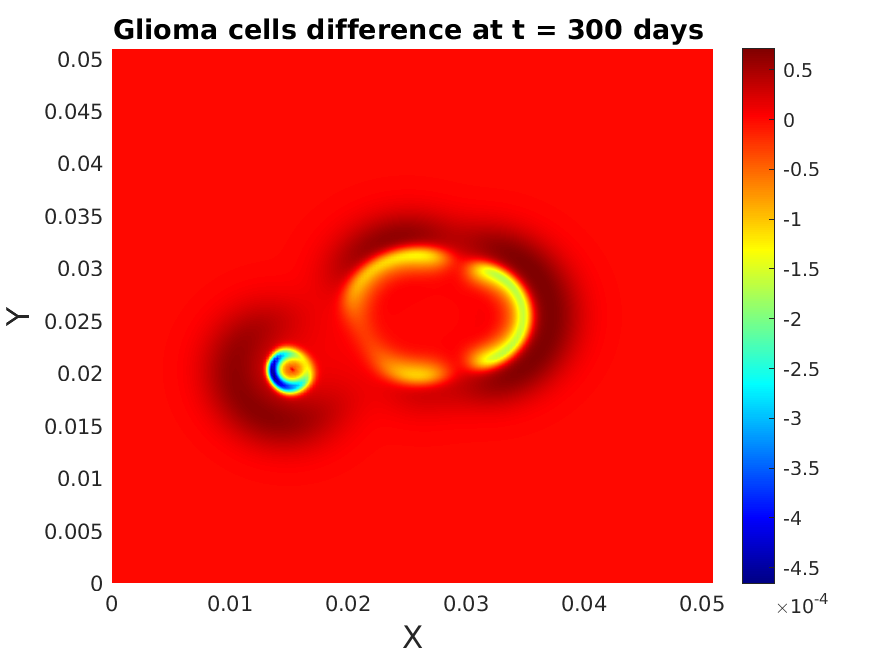}}\\
		\raisebox{1.2cm}{\rotatebox[origin=t]{90}{360 days}}{\includegraphics[width=1\linewidth, height = 3cm]{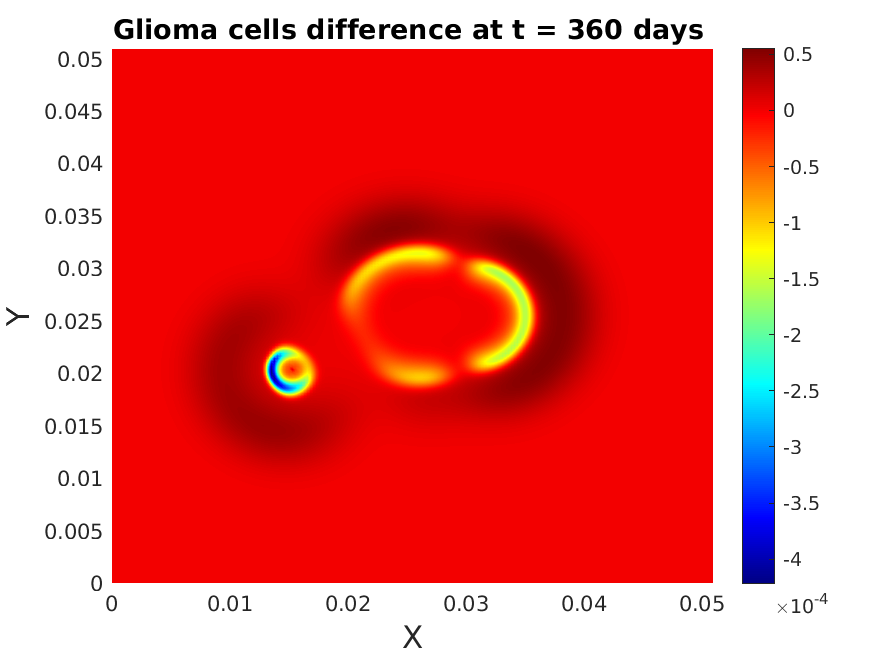}}\\
		\raisebox{1.2cm}{\rotatebox[origin=t]{90}{450 days}}{\includegraphics[width=1\linewidth, height = 3cm]{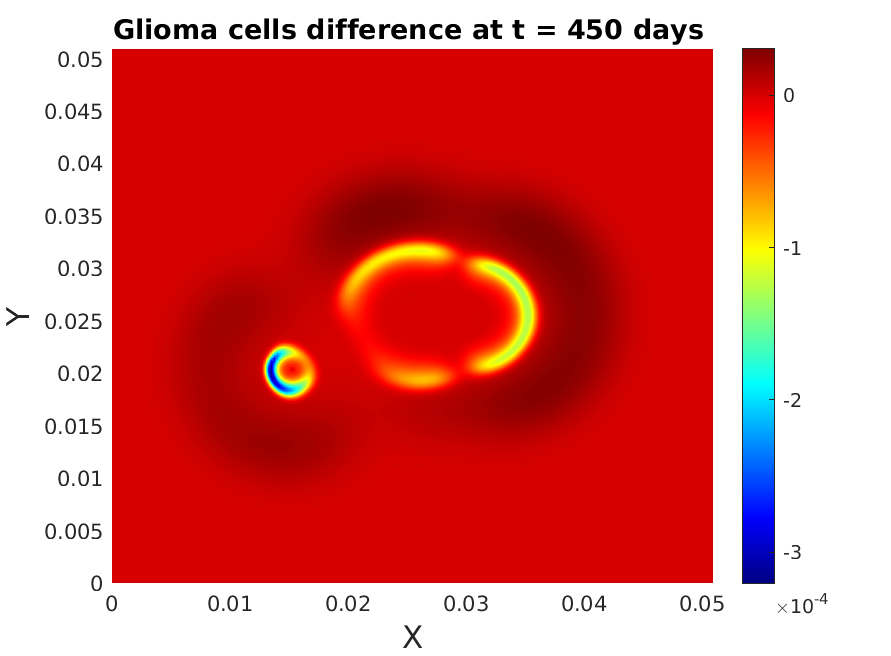}}%
		\subcaption{\scriptsize Glioma cells difference}
	\end{minipage}%
	\hspace{0.2cm}
	\begin{minipage}[hstb]{.24\linewidth}
		{\includegraphics[width=1\linewidth, height = 3cm]{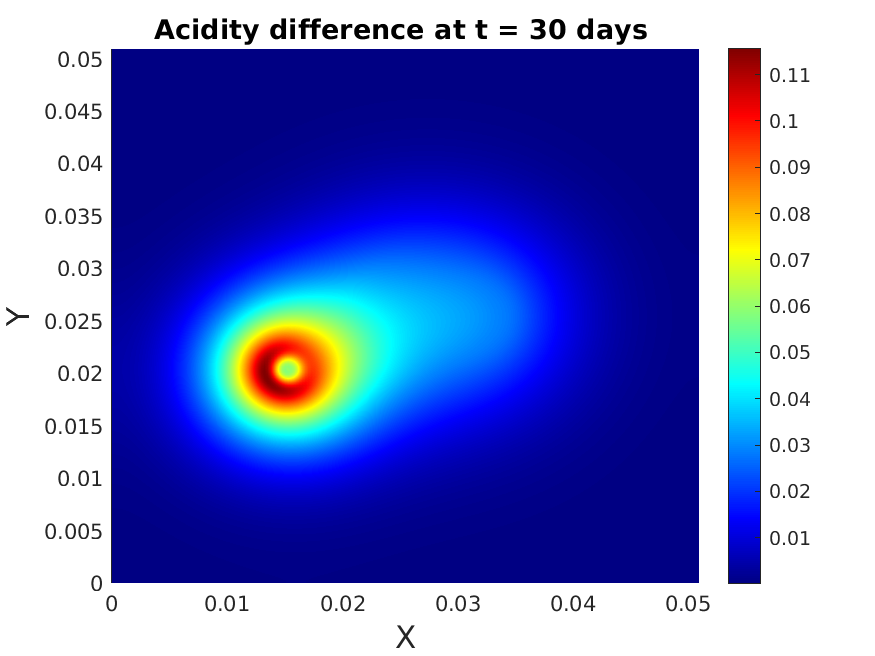}}\\
		{\includegraphics[width=1\linewidth, height = 3cm]{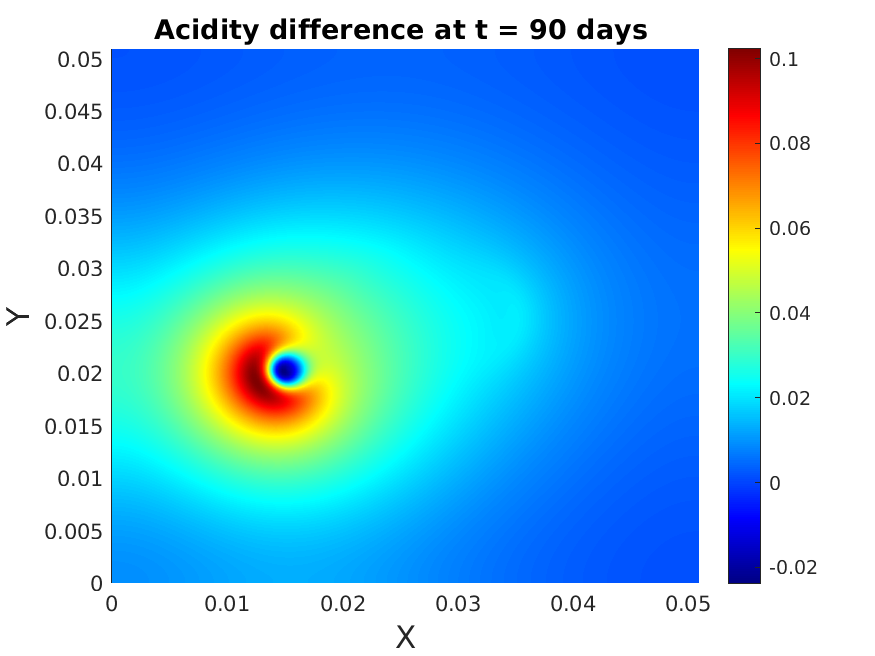}}\\
		{\includegraphics[width=1\linewidth, height = 3cm]{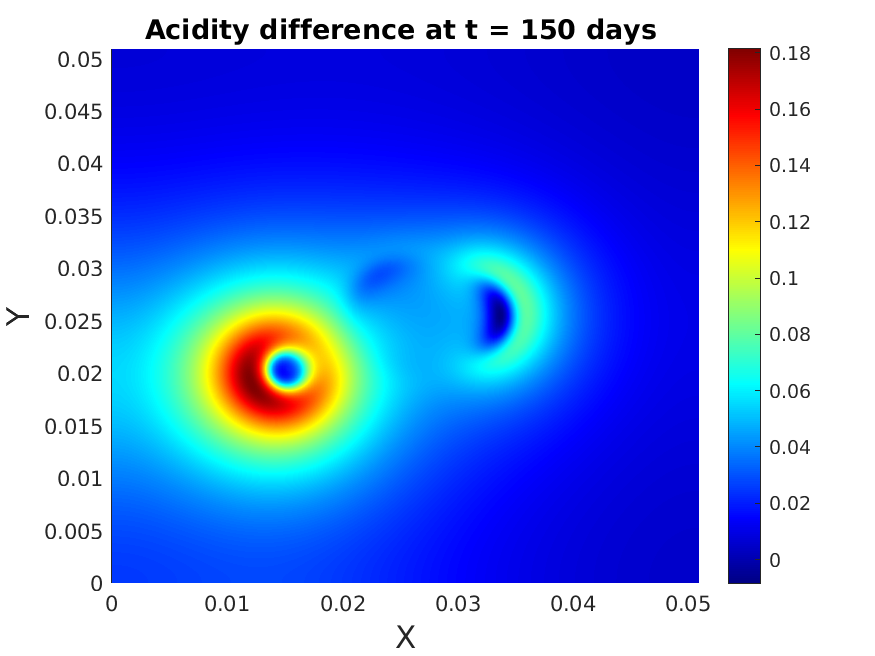}}\\
		{\includegraphics[width=1\linewidth, height = 3cm]{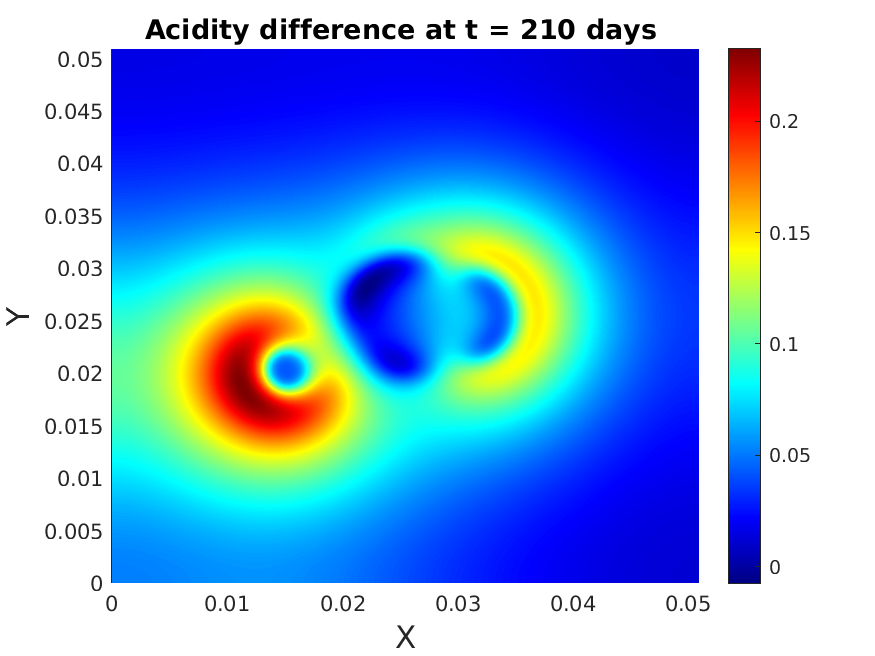}}\\
		{\includegraphics[width=1\linewidth, height = 3cm]{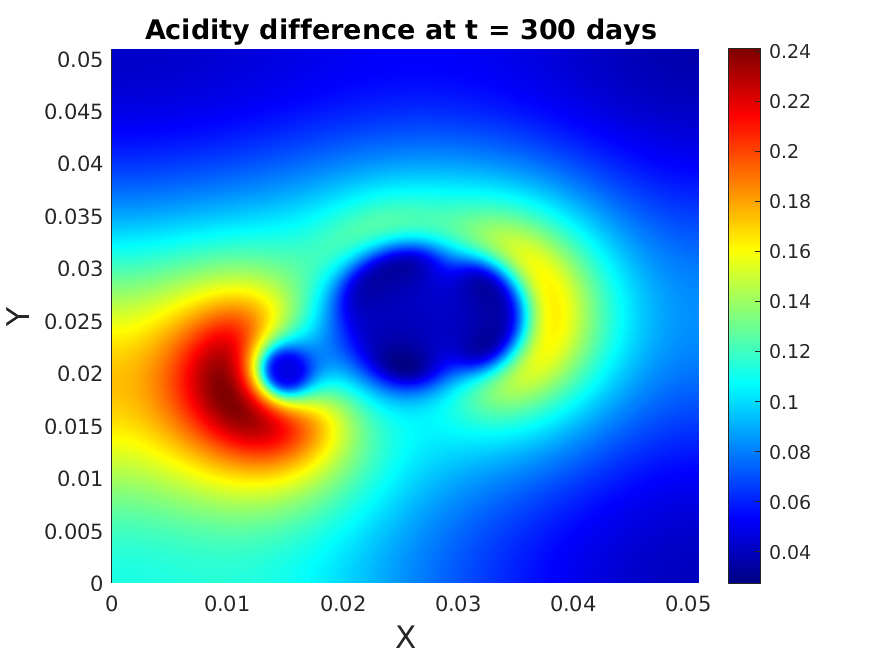}}\\
		{\includegraphics[width=1\linewidth, height = 3cm]{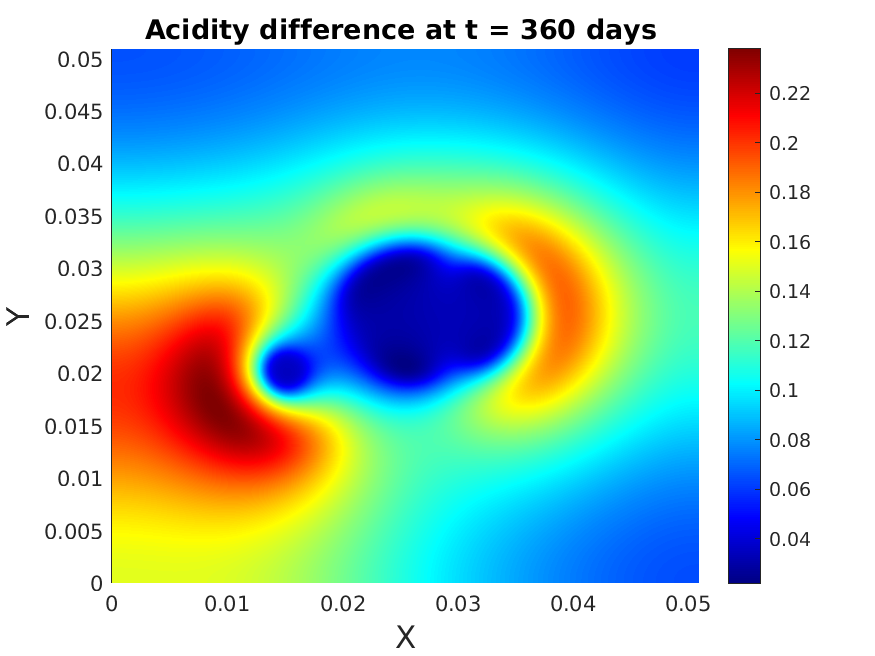}}\\
		{\includegraphics[width=1\linewidth, height = 3cm]{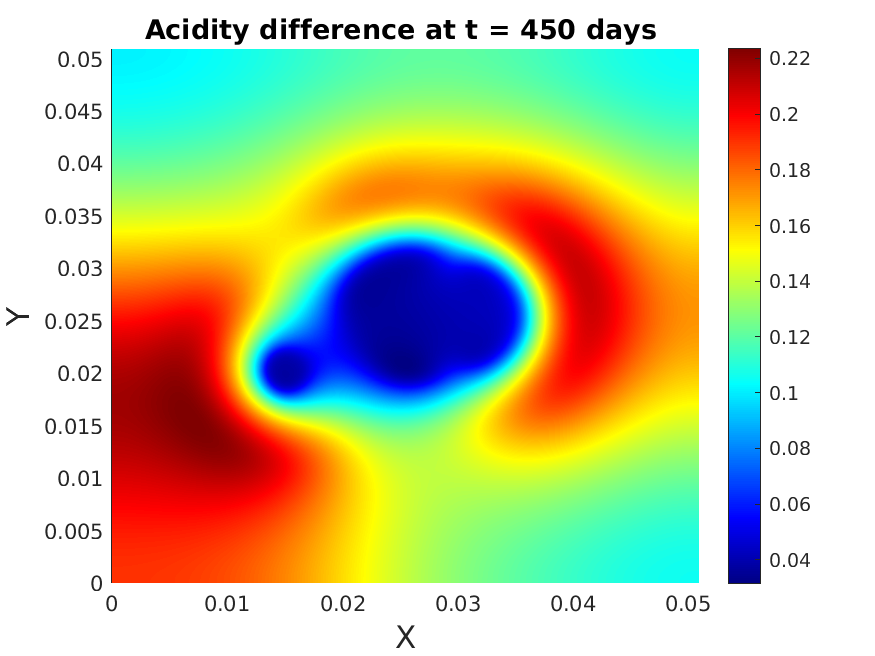}}%
		\subcaption{\scriptsize Acidity difference}
	\end{minipage}%
	\hspace{0.01cm}
	\begin{minipage}[hstb]{.24\linewidth}
		{\includegraphics[width=1\linewidth, height = 3cm]{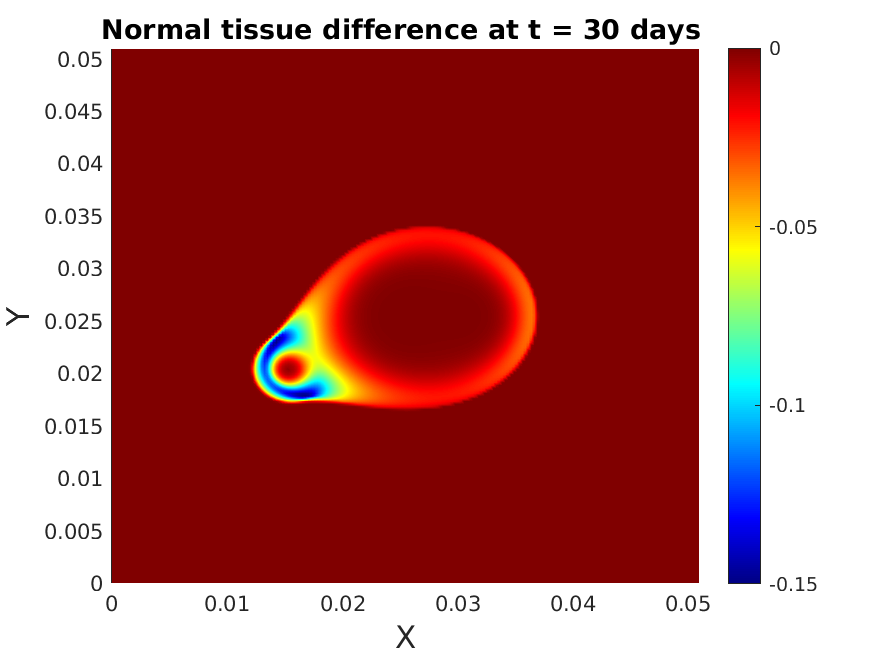}}\\
		{\includegraphics[width=1\linewidth, height = 3cm]{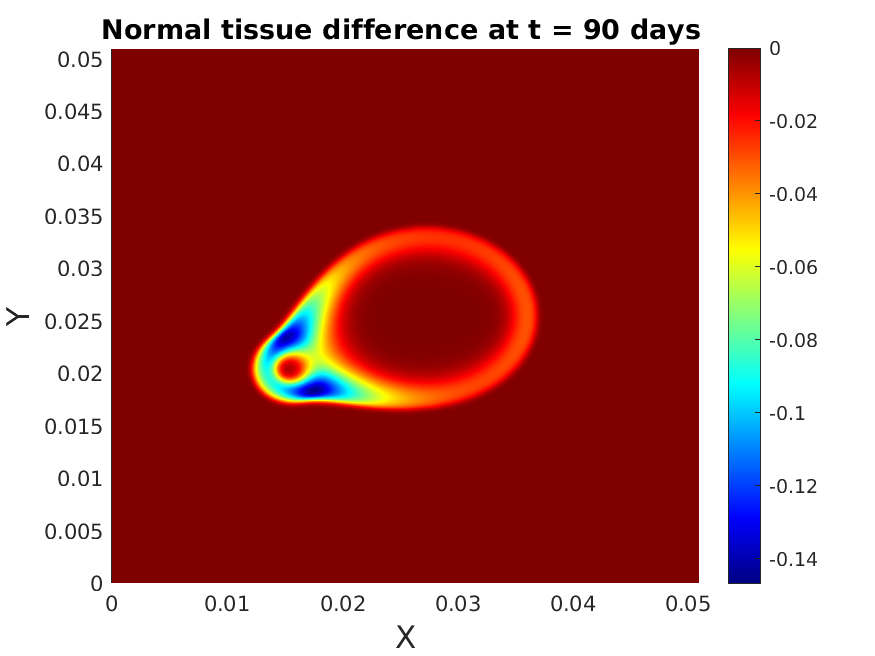}}\\
		{\includegraphics[width=1\linewidth, height = 3cm]{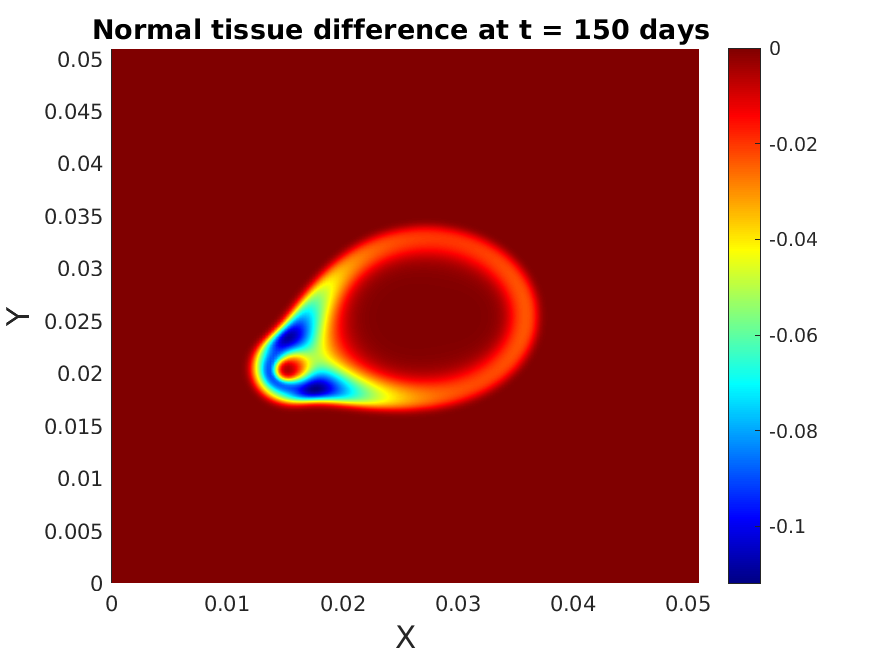}}\\
		{\includegraphics[width=1\linewidth, height = 3cm]{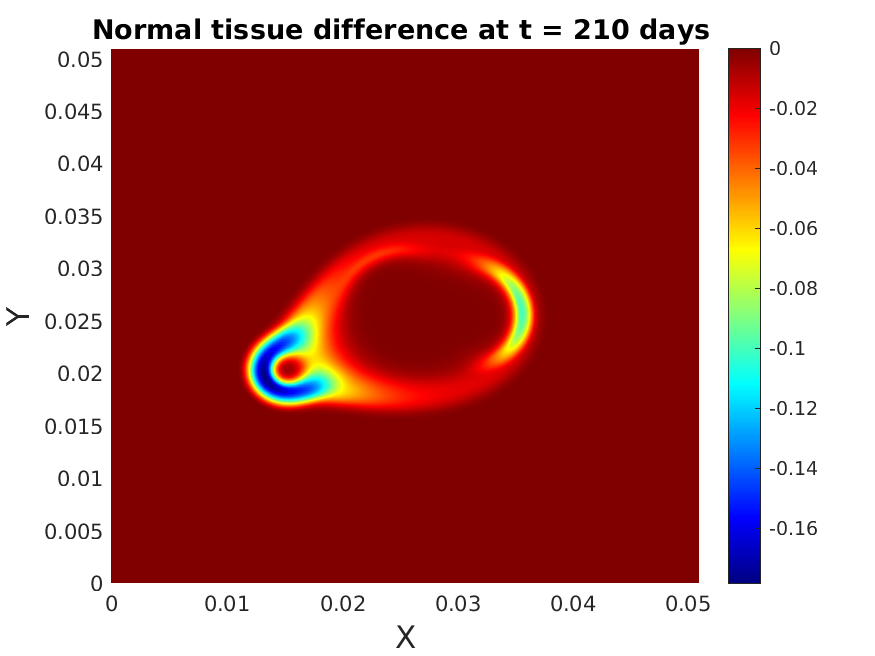}}\\
		{\includegraphics[width=1\linewidth, height = 3cm]{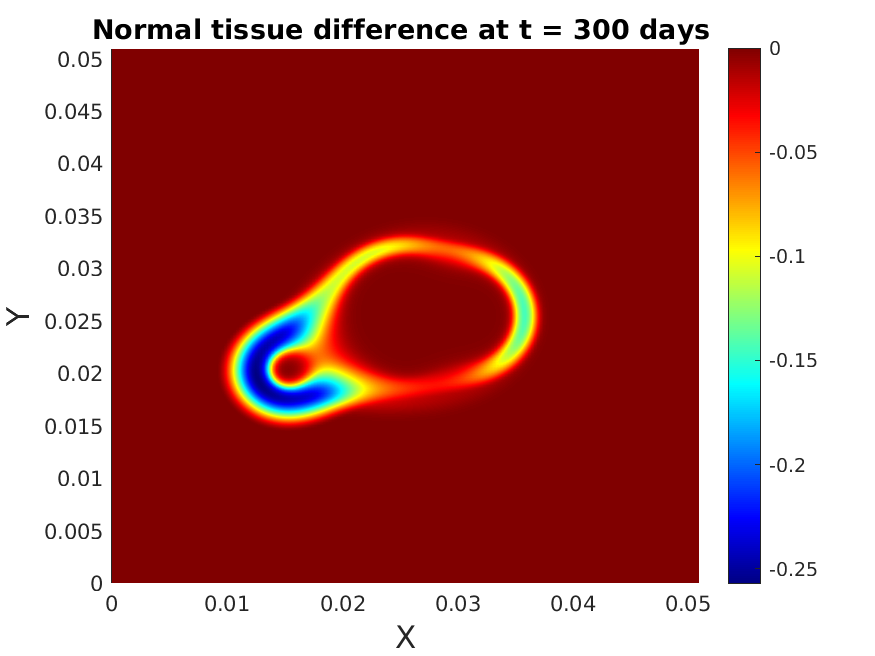}}\\
		{\includegraphics[width=1\linewidth, height = 3cm]{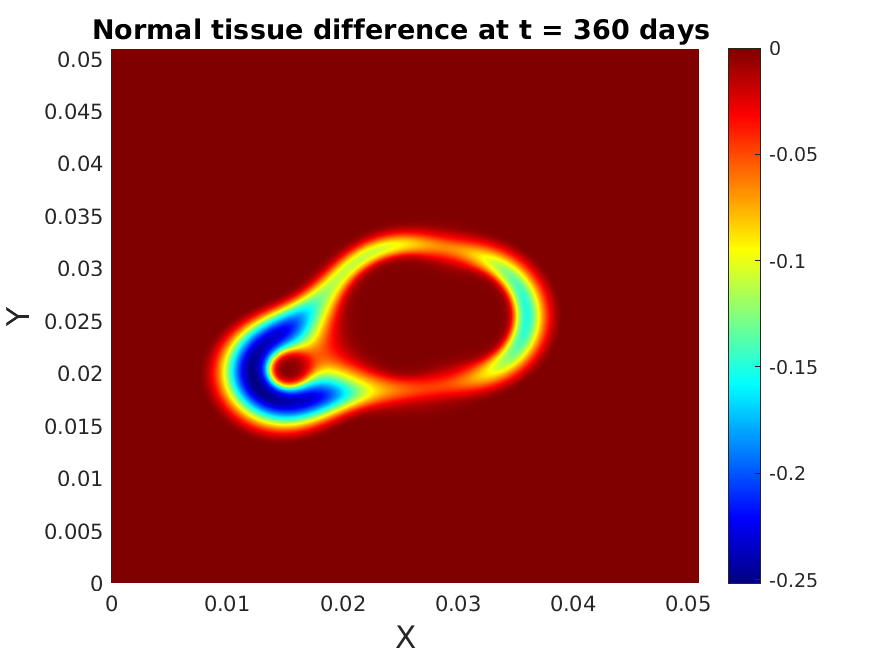}}\\
		{\includegraphics[width=1\linewidth, height = 3cm]{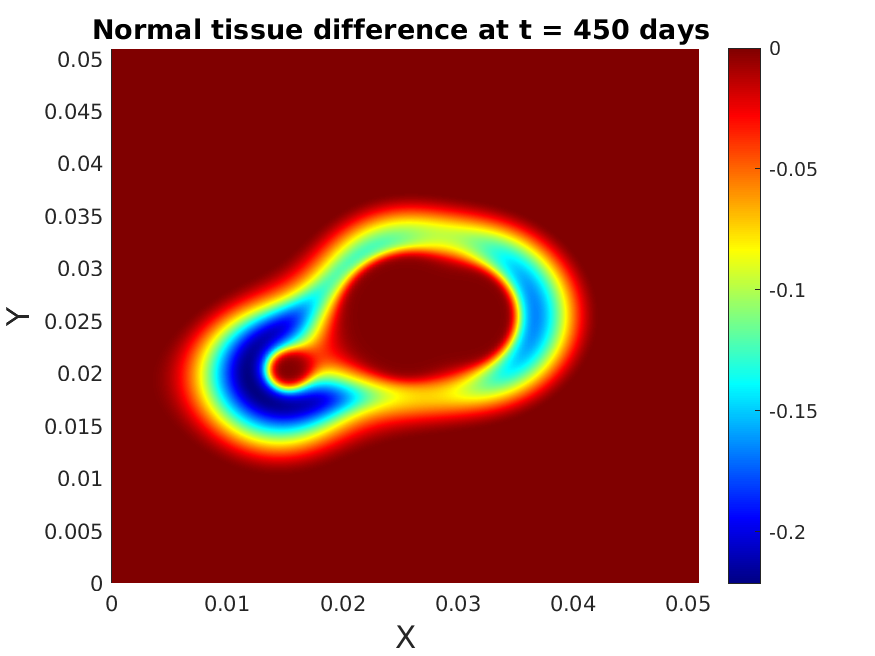}}%
		\subcaption{\scriptsize Normal tissue difference}
	\end{minipage}%
	\hspace{0.01cm}
	\begin{minipage}[hstb]{.24\linewidth}
		{\includegraphics[width=1\linewidth, height = 3cm]{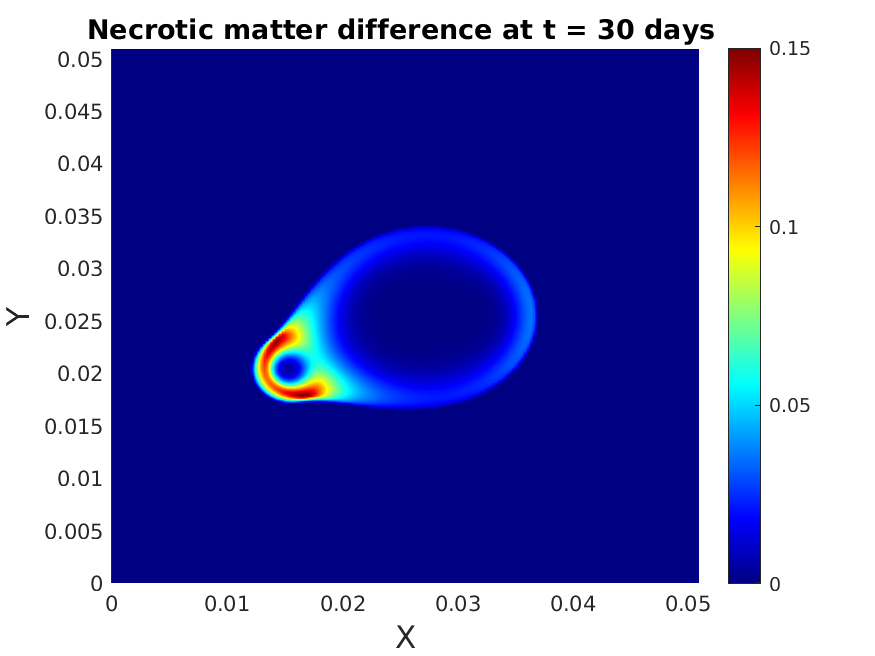}}\\
		{\includegraphics[width=1\linewidth, height = 3cm]{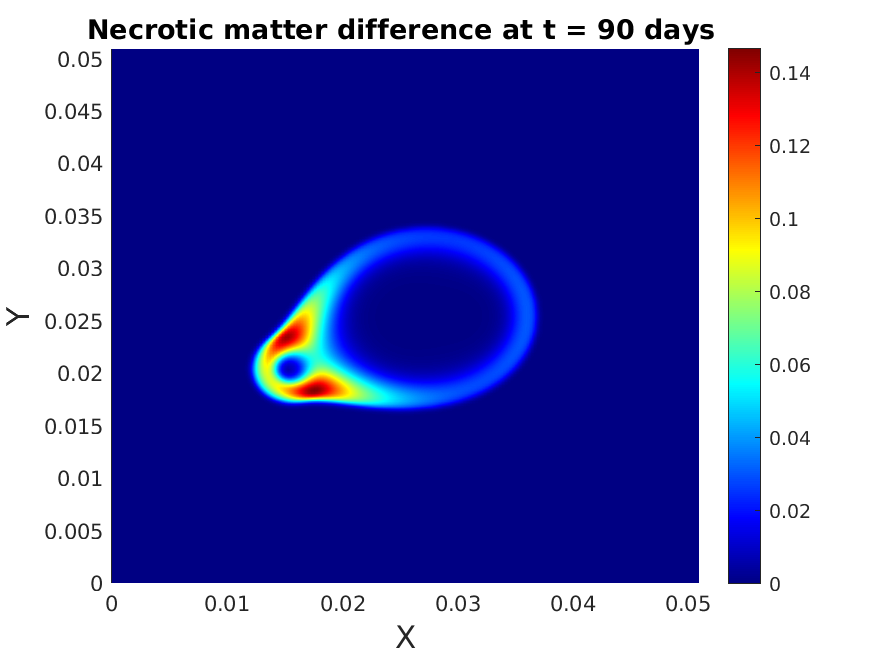}}\\
		{\includegraphics[width=1\linewidth, height = 3cm]{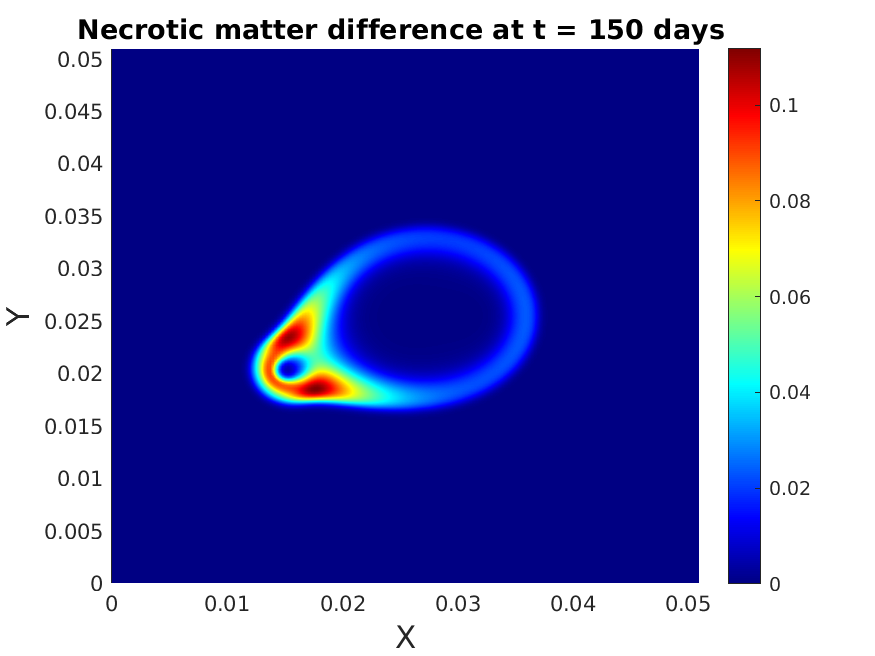}}\\
		{\includegraphics[width=1\linewidth, height = 3cm]{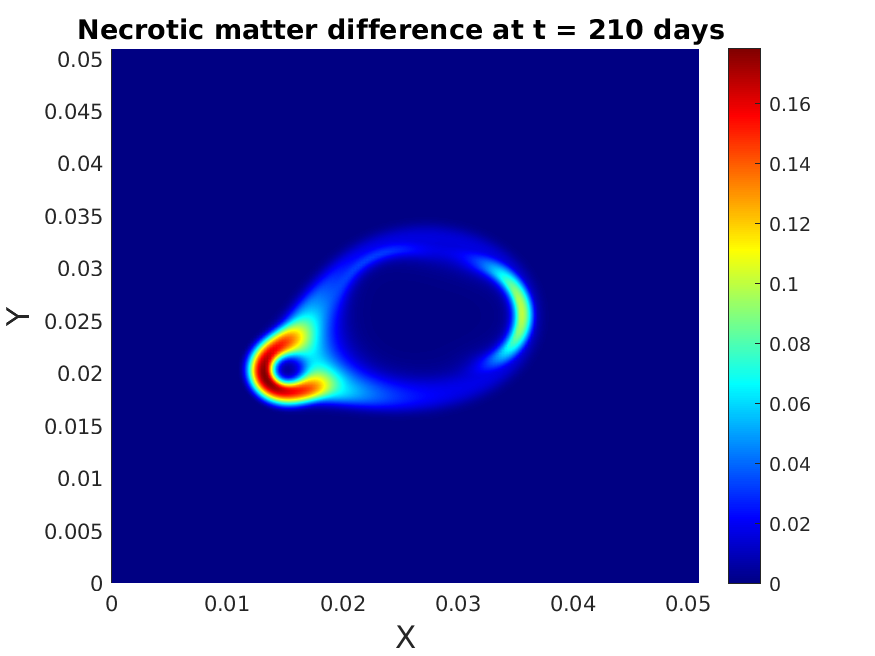}}\\
		{\includegraphics[width=1\linewidth, height = 3cm]{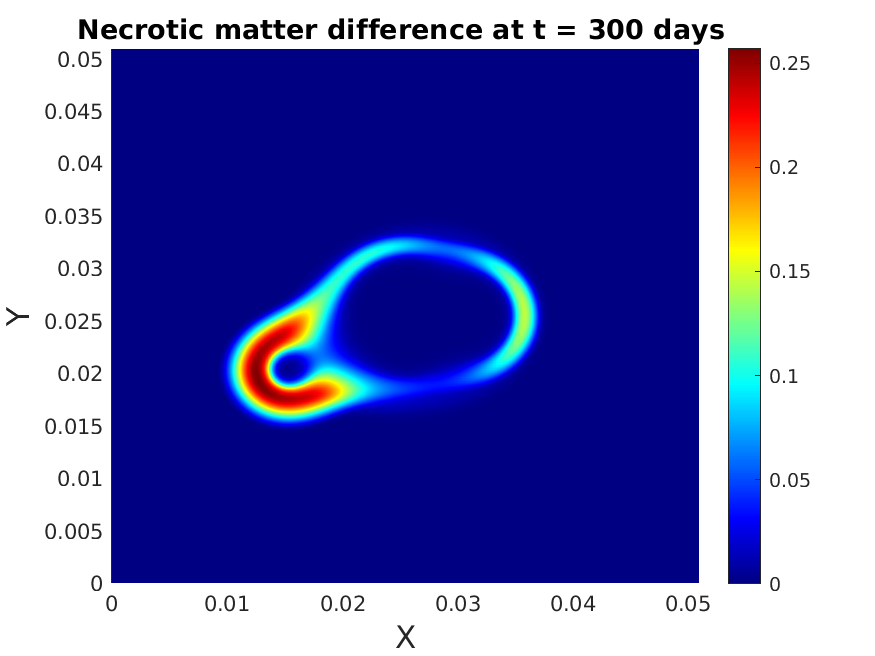}}\\
		{\includegraphics[width=1\linewidth, height = 3cm]{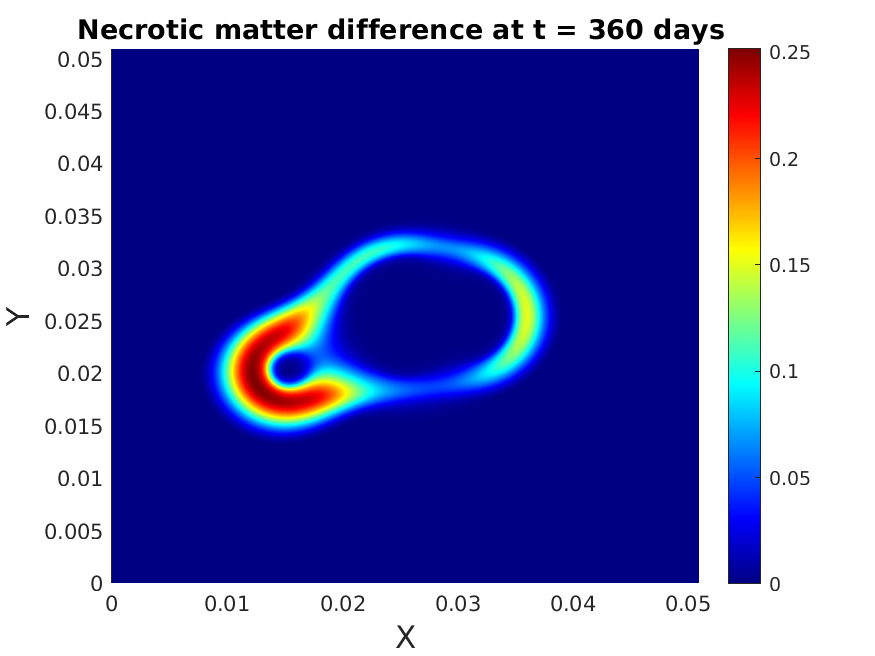}}\\
		{\includegraphics[width=1\linewidth, height = 3cm]{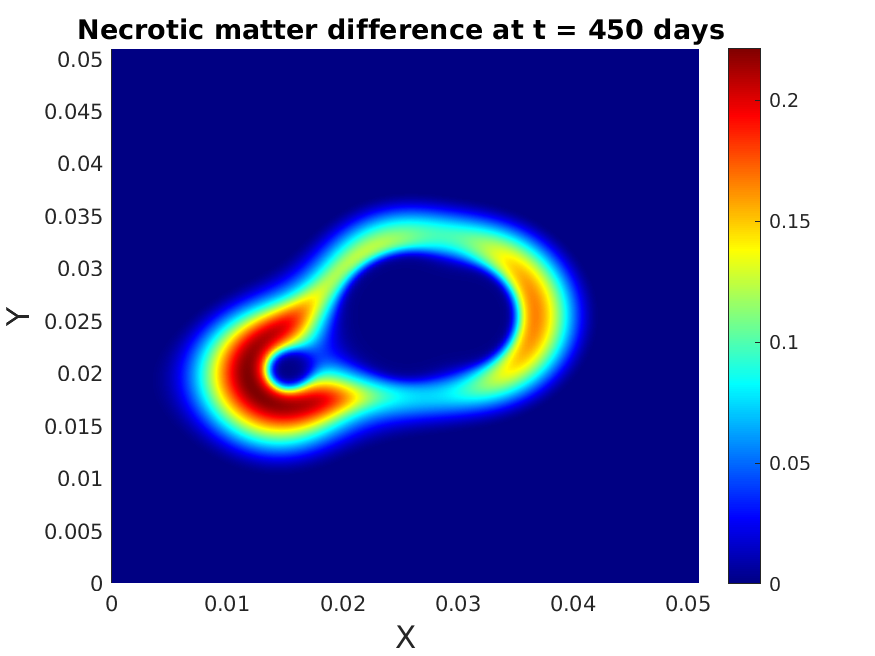}}%
		\subcaption{\scriptsize Necrotic matter difference}
	\end{minipage}%
	\caption{Difference between solution components obtained with $\chi>0$ and those computed with $\chi=0$.}
	\label{fig:comp-acid}
\end{figure}

Finally, we also study the effect of drag coefficients $K_{cm}, K_{mn}, K_{cn}$ in \eqref{Force}. Recall that they represent the resistance inferred by the phases in contact when passing over each other. They are contained in the diffusion and drift coefficients of system \eqref{model2}. In the computations for Figure \ref{fig:simulations_mpm} we had considered them to be different, more precisely $K_{cm}<K_{mn}<K_{cn}$. This ensured a relatively easy movement of glioma through normal tissue when compared to the shift over necrotic matter of both normal tissue and cancer cells. Now let all these drag coefficients be equal (as was assumed in \cite{Jackson2002307}) and take the previously lowermost value $K_{cm}$, thus reducing the drag between necrotic matter and the other two phases. The differences between the two cases are plotted in Figure \ref{fig:comp-Kcm} and show in the second case an enhanced outward migration of glioma cells, away from the highly acidic area at the core of the pseudopalisade and with a pronounced suppression of normal cells, and emergence of larger necrotic matter. An opposite effect is noticed when letting $K_{cm}=K_{mn}=K_{cn}$, all at the previously highest value $K_{cn}$; we do not show here those results (we refer for them to \cite{pawans-diss}), but rather illustrate in Figure \ref{fig:comp-Kmn} the situation when all drag coefficients equal the previously intermediate value $K_{mn}$. This means that the drag between cancer cells and normal tissue is increased, while they can easier shift over the necrotic matter. As a consequence, the extent of the pseudopalisades is reduced, but the glioma aggregates in the garland-like structures are larger. This is due not only to the higher drag between cancer cells and normal tissue, but also to the lower $K_{cn}$ value, which ensures that the glioma cells can leave faster than previously the acidic area inside the 'ring'. Around the site where the cancer cells were initially more concentrated the necrosis is reduced and the normal tissue is better preserved. 

\section*{Acknowledgement} P.K. acknowledges funding by the German Academic Exchange Service (DAAD) in form of a PhD scholar\-ship.

\begin{figure}[!htbp]
	\centering
		\begin{minipage}[hstb]{.24\linewidth}
		\raisebox{1.2cm}{\rotatebox[origin=t]{90}{30 days}}{\includegraphics[width=1\linewidth, height = 3cm]{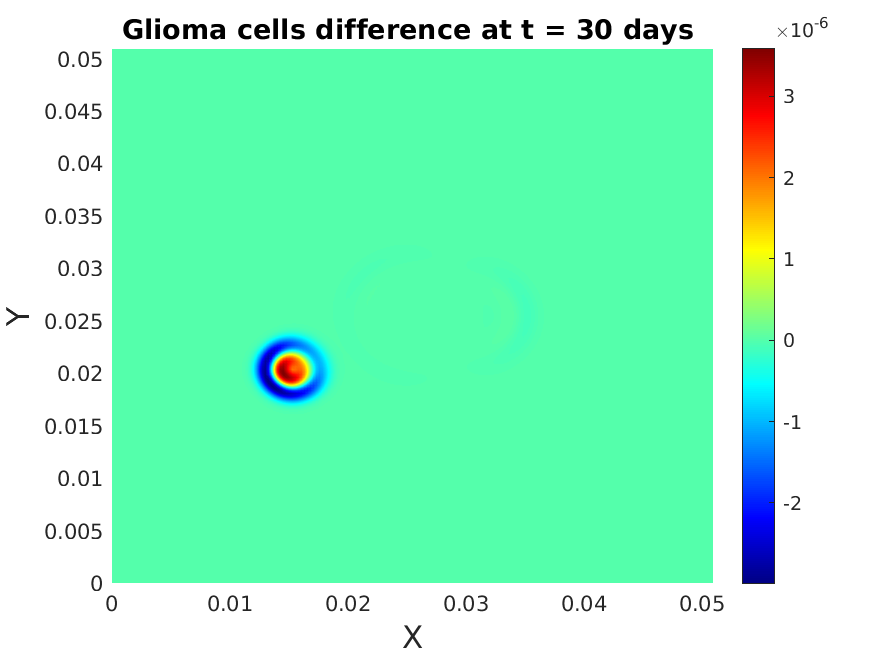}}\\
		\raisebox{1.2cm}{\rotatebox[origin=t]{90}{90 days}}{\includegraphics[width=1\linewidth, height = 3cm]{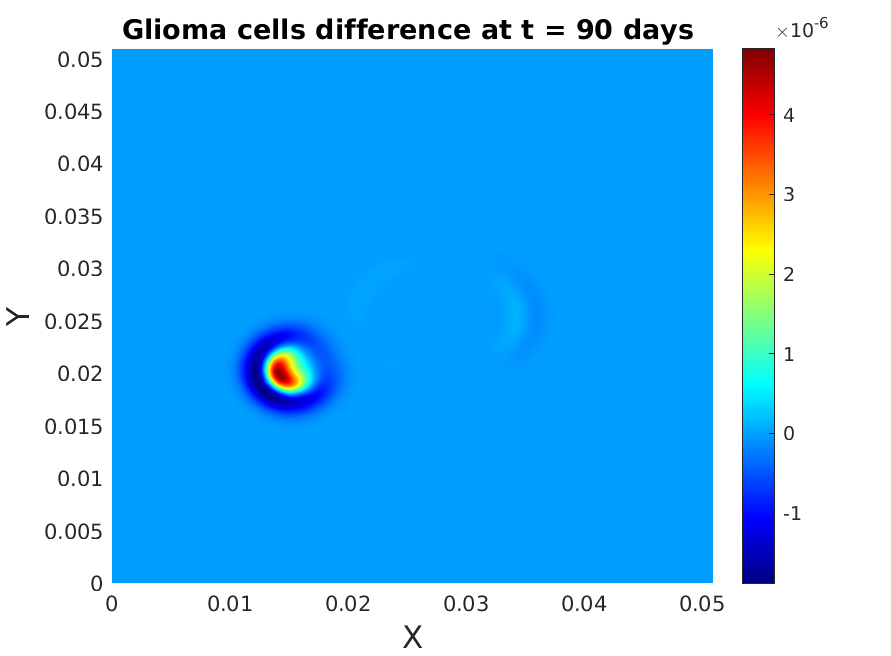}}\\
		\raisebox{1.2cm}{\rotatebox[origin=t]{90}{150 days}}{\includegraphics[width=1\linewidth, height = 3cm]{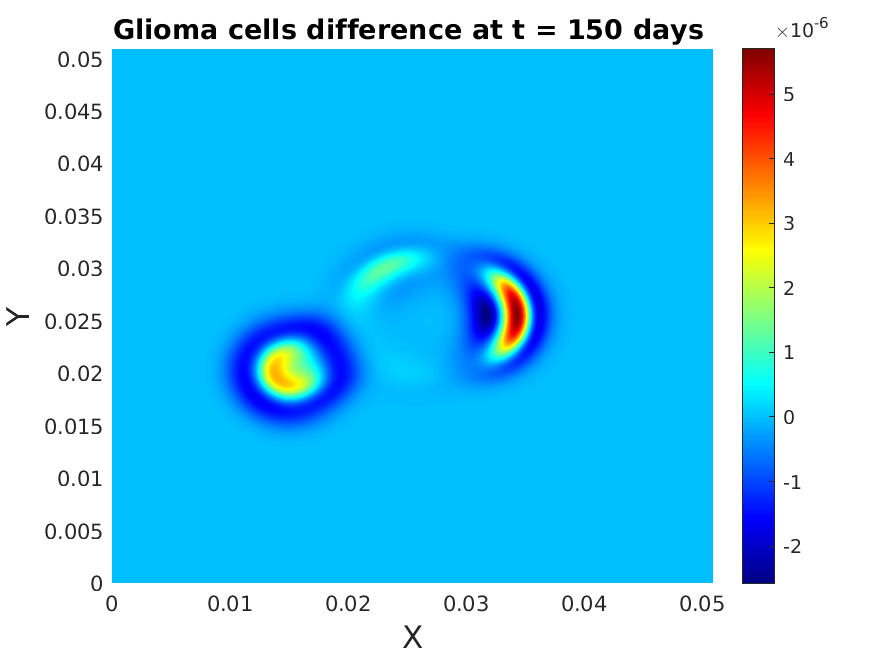}}\\
		\raisebox{1.2cm}{\rotatebox[origin=t]{90}{210 days}}{\includegraphics[width=1\linewidth, height = 3cm]{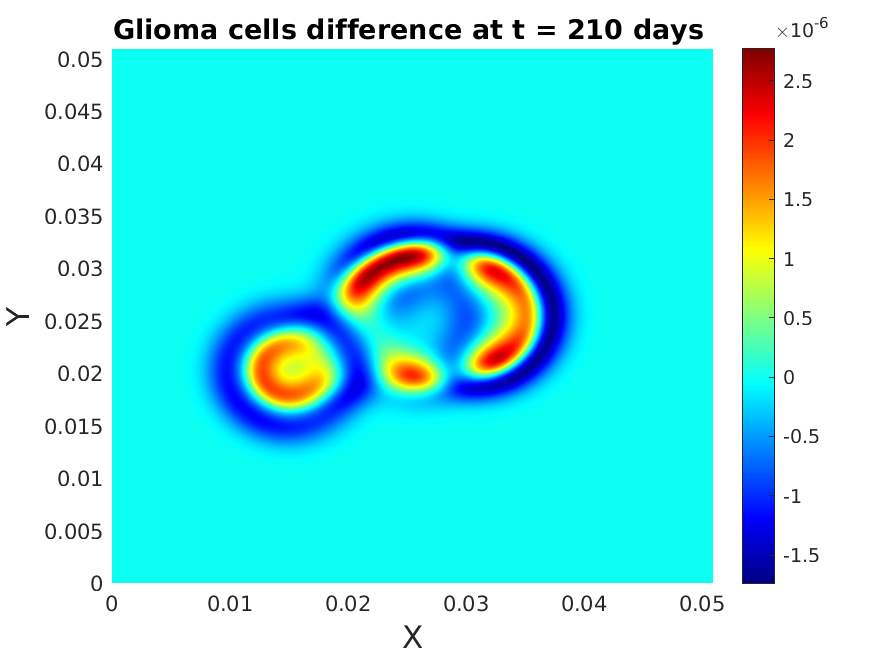}}\\
		\raisebox{1.2cm}{\rotatebox[origin=t]{90}{300 days}}{\includegraphics[width=1\linewidth, height = 3cm]{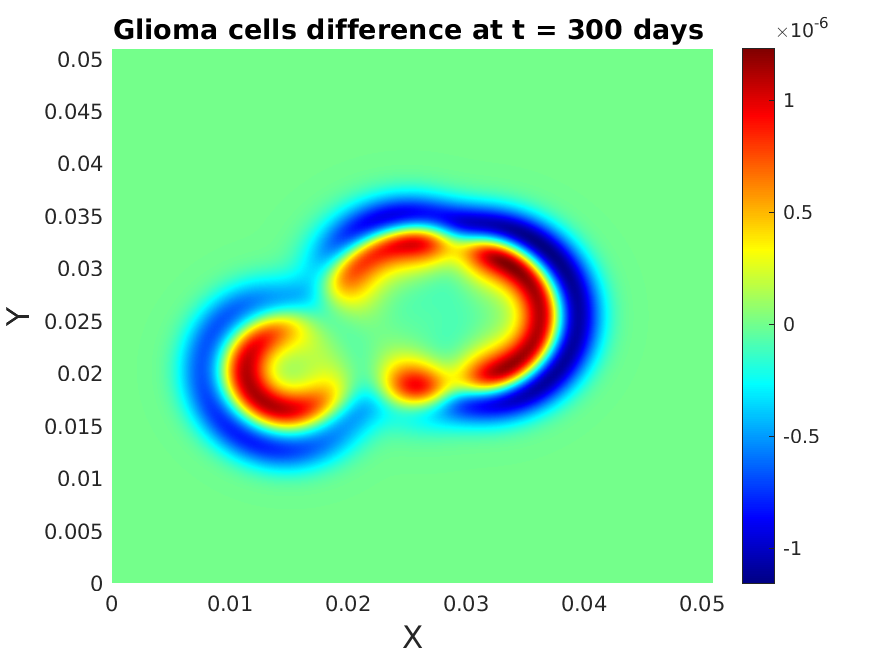}}\\
		\raisebox{1.2cm}{\rotatebox[origin=t]{90}{360 days}}{\includegraphics[width=1\linewidth, height = 3cm]{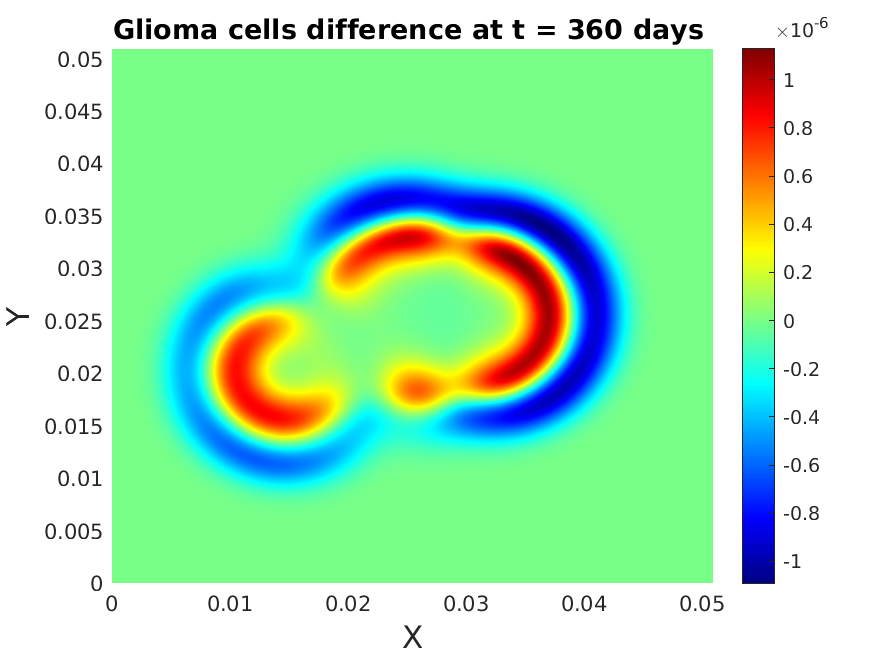}}\\
		\raisebox{1.2cm}{\rotatebox[origin=t]{90}{450 days}}{\includegraphics[width=1\linewidth, height = 3cm]{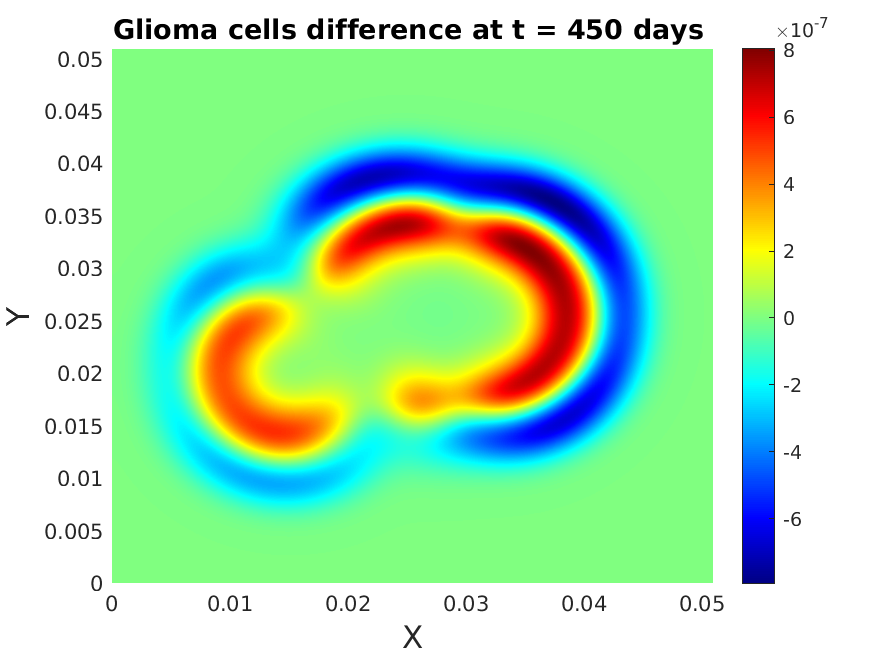}}%
		\subcaption{\scriptsize Glioma cells difference}
	\end{minipage}%
	\hspace{0.2cm}
	\begin{minipage}[hstb]{.24\linewidth}
		{\includegraphics[width=1\linewidth, height = 3cm]{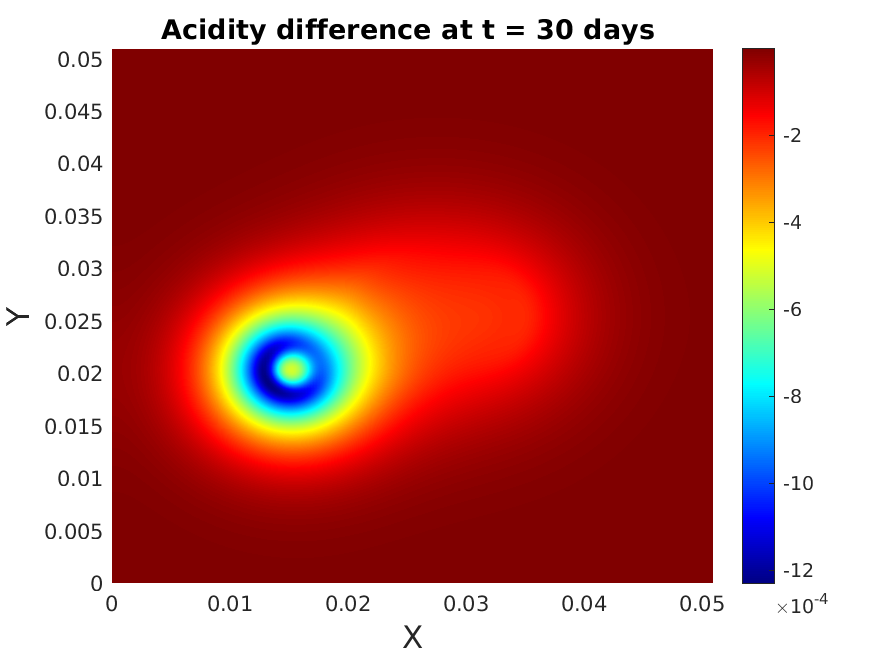}}\\
		{\includegraphics[width=1\linewidth, height = 3cm]{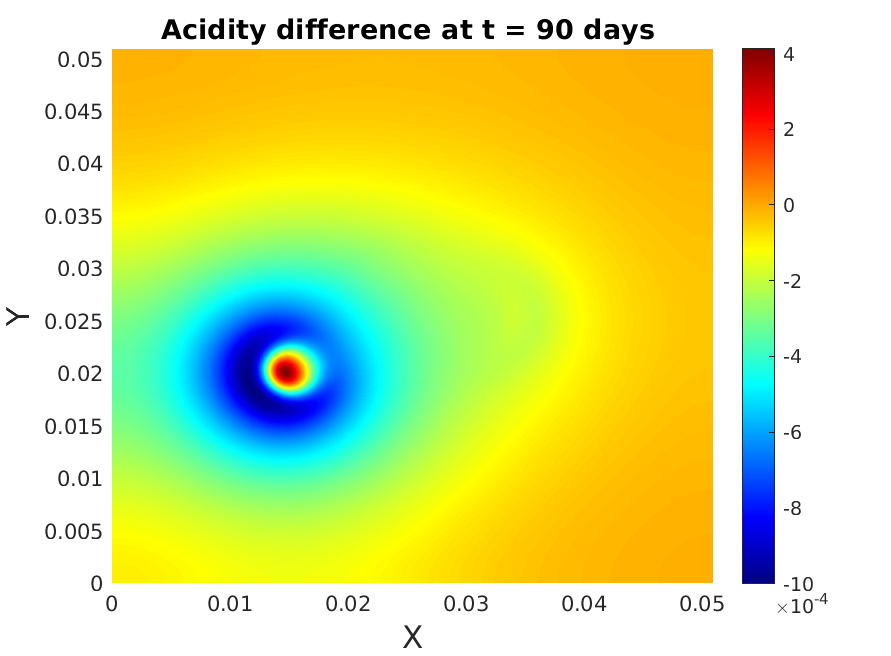}}\\
		{\includegraphics[width=1\linewidth, height = 3cm]{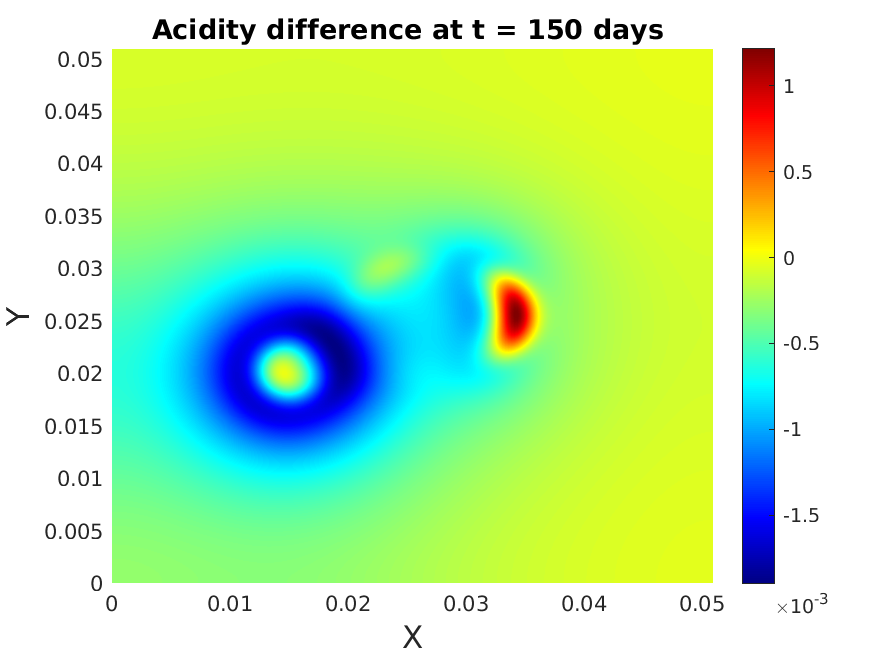}}\\
		{\includegraphics[width=1\linewidth, height = 3cm]{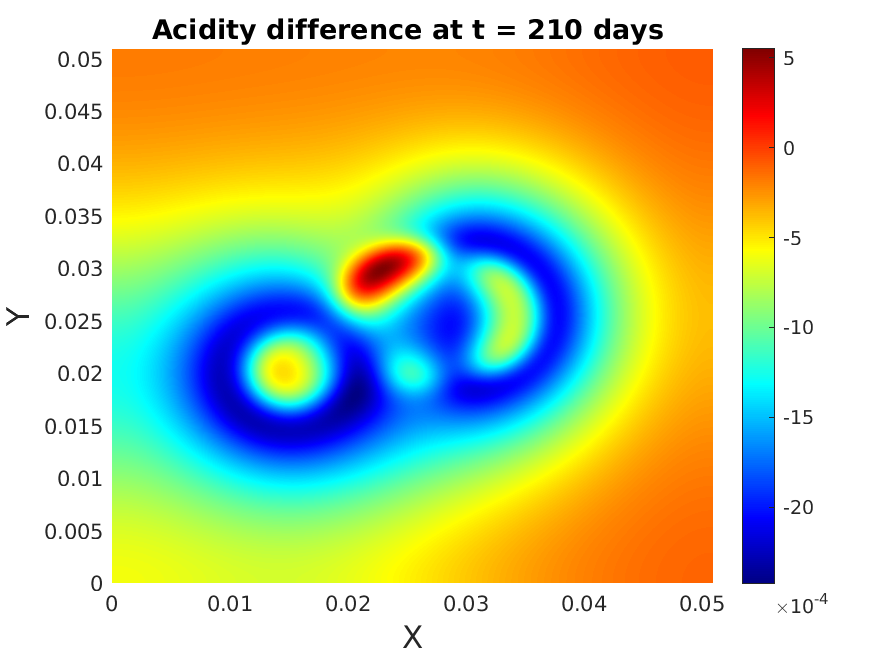}}\\
		{\includegraphics[width=1\linewidth, height = 3cm]{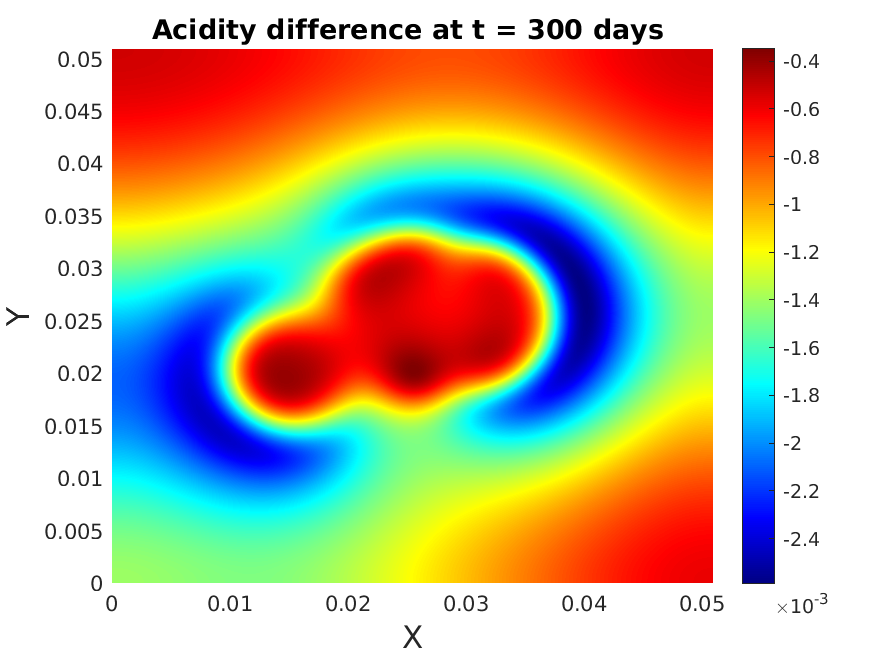}}\\
		{\includegraphics[width=1\linewidth, height = 3cm]{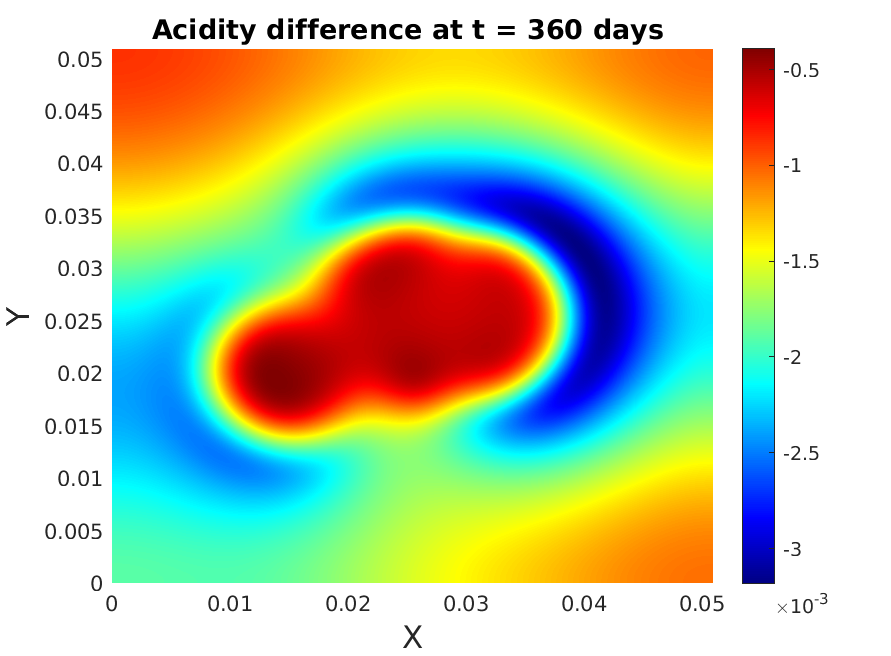}}\\
		{\includegraphics[width=1\linewidth, height = 3cm]{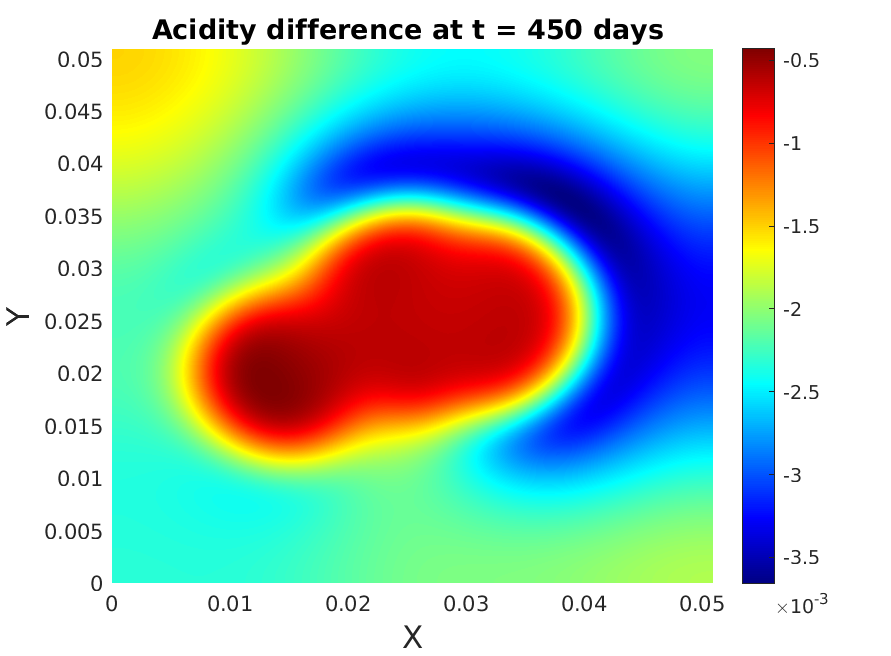}}%
		\subcaption{\scriptsize Acidity difference}
	\end{minipage}%
	\hspace{0.01cm}
	\begin{minipage}[hstb]{.24\linewidth}
		{\includegraphics[width=1\linewidth, height = 3cm]{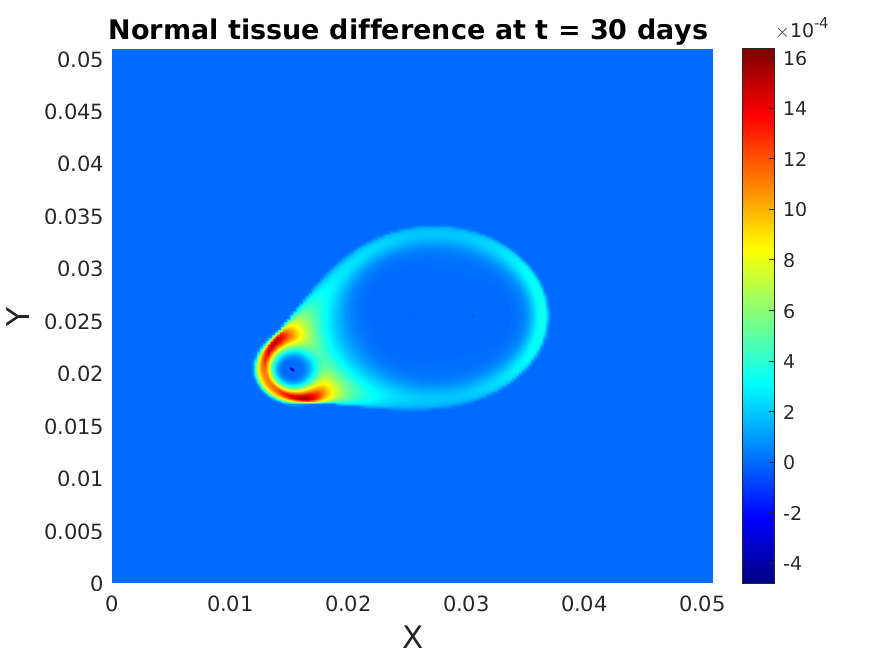}}\\
		{\includegraphics[width=1\linewidth, height = 3cm]{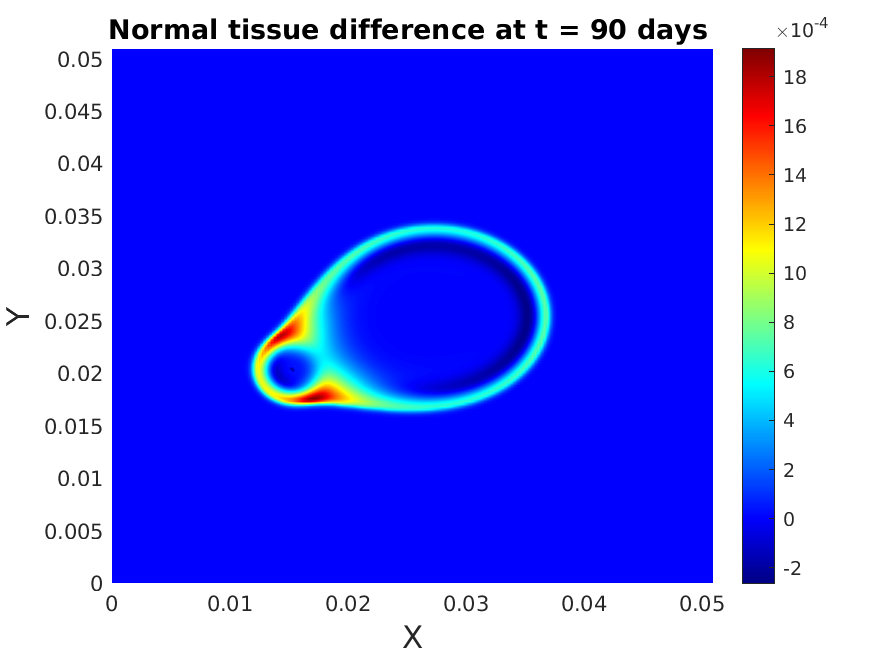}}\\
		{\includegraphics[width=1\linewidth, height = 3cm]{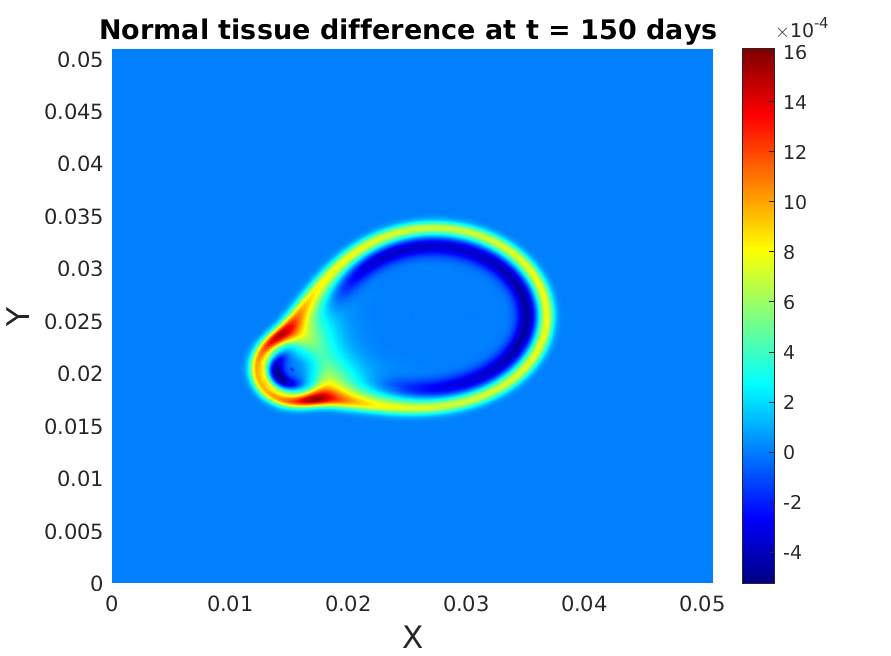}}\\
		{\includegraphics[width=1\linewidth, height = 3cm]{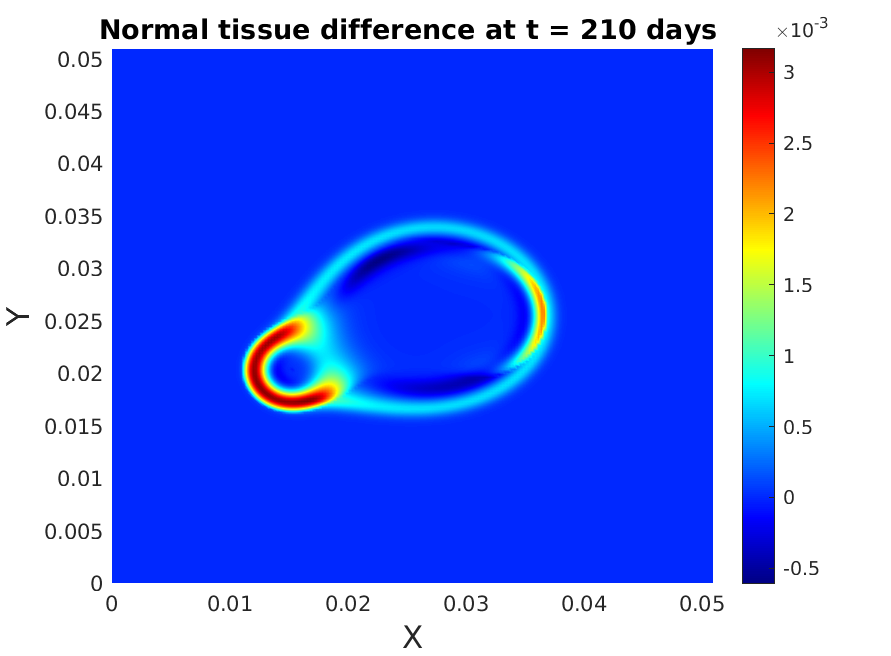}}\\
		{\includegraphics[width=1\linewidth, height = 3cm]{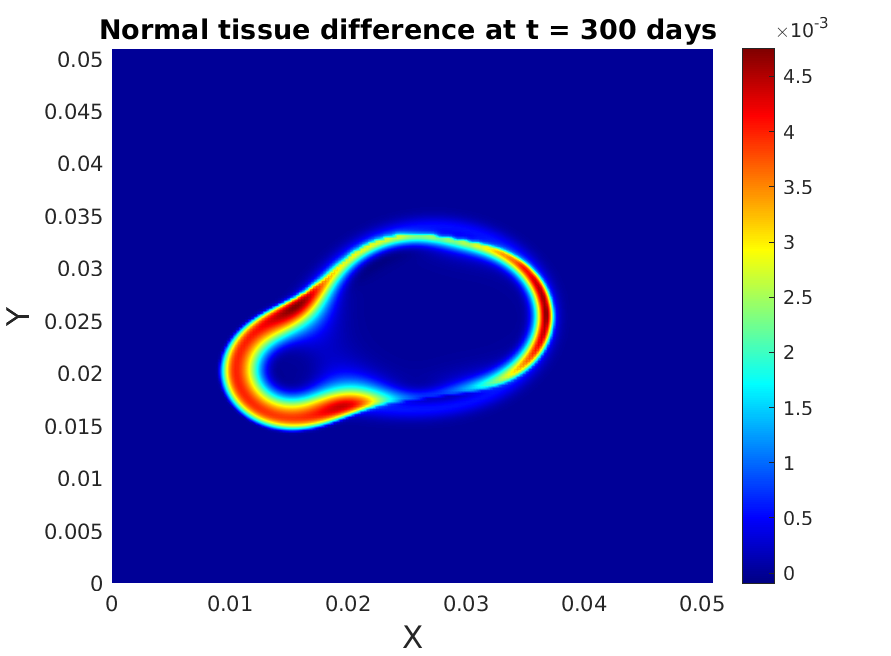}}\\
		{\includegraphics[width=1\linewidth, height = 3cm]{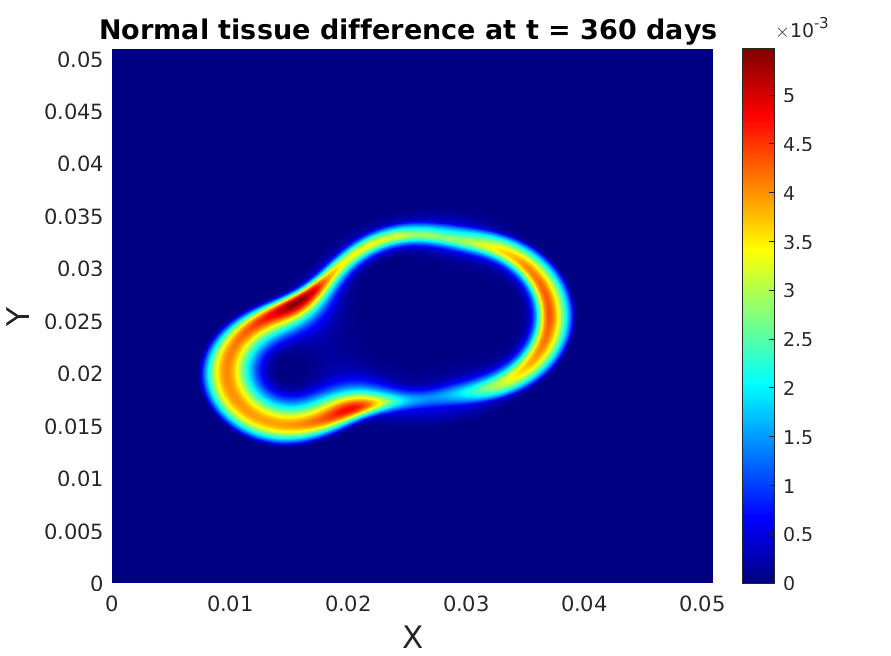}}\\
		{\includegraphics[width=1\linewidth, height = 3cm]{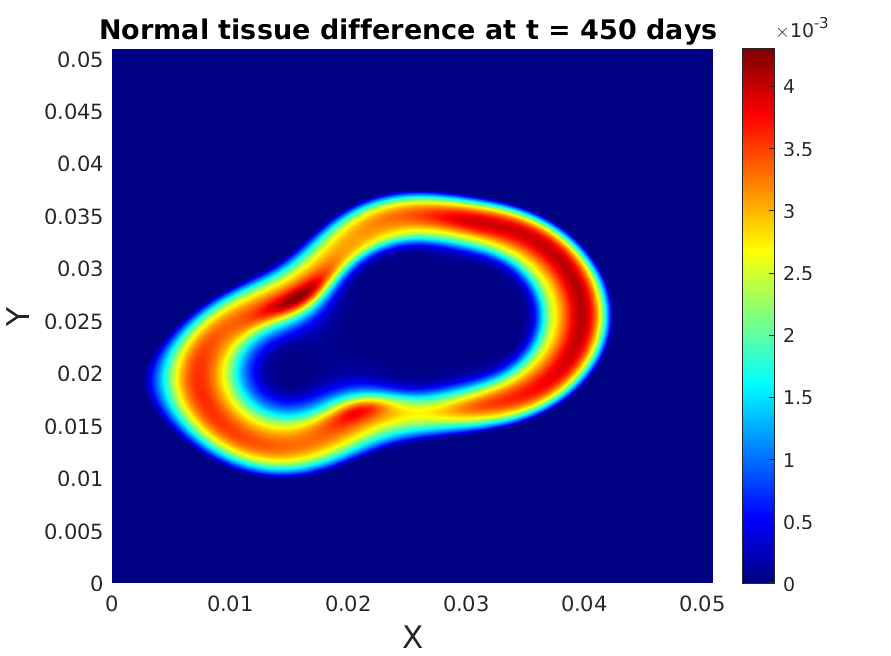}}%
		\subcaption{\scriptsize Normal tissue difference}
	\end{minipage}%
	\hspace{0.01cm}
	\begin{minipage}[hstb]{.24\linewidth}
		{\includegraphics[width=1\linewidth, height = 3cm]{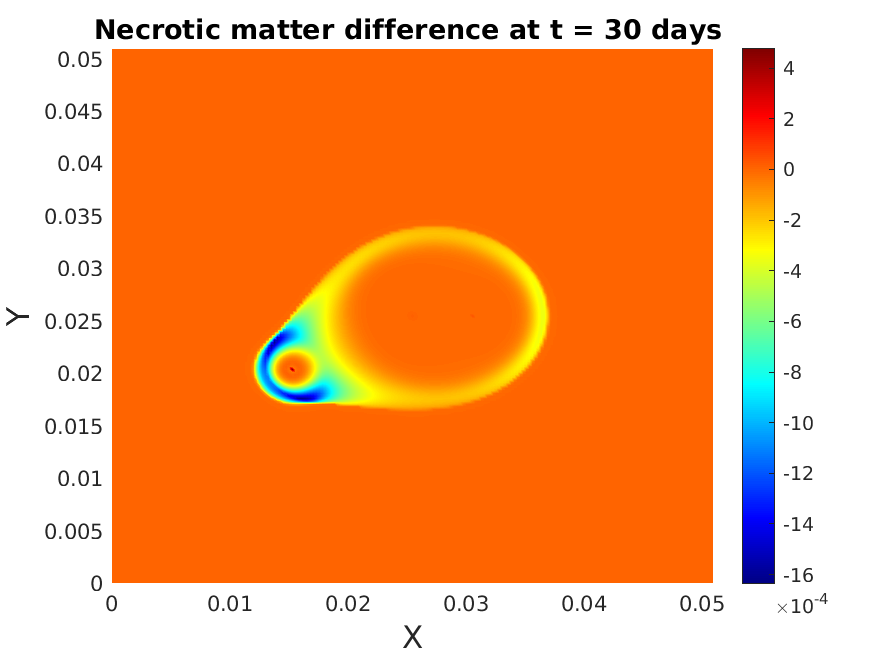}}\\
		{\includegraphics[width=1\linewidth, height = 3cm]{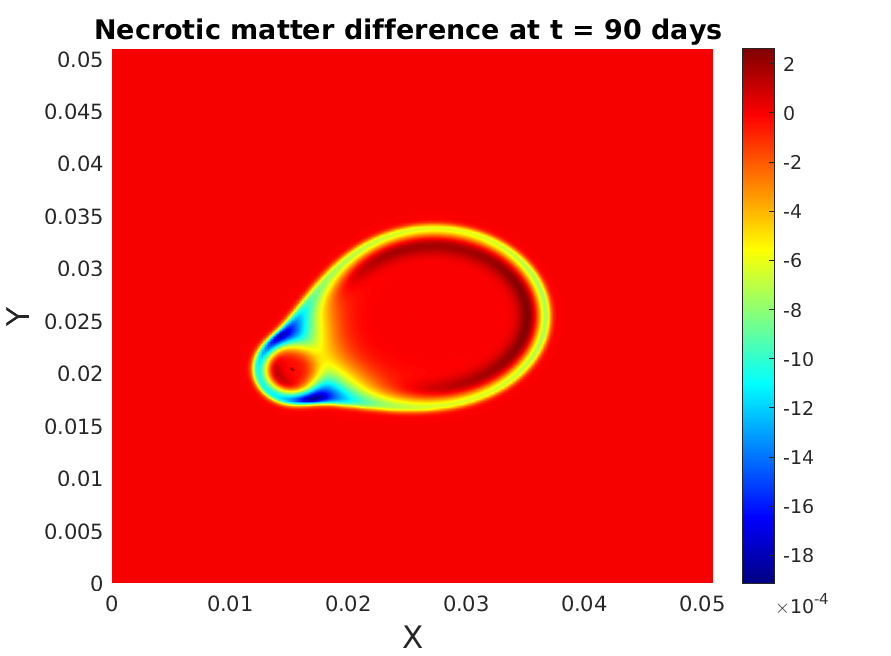}}\\
		{\includegraphics[width=1\linewidth, height = 3cm]{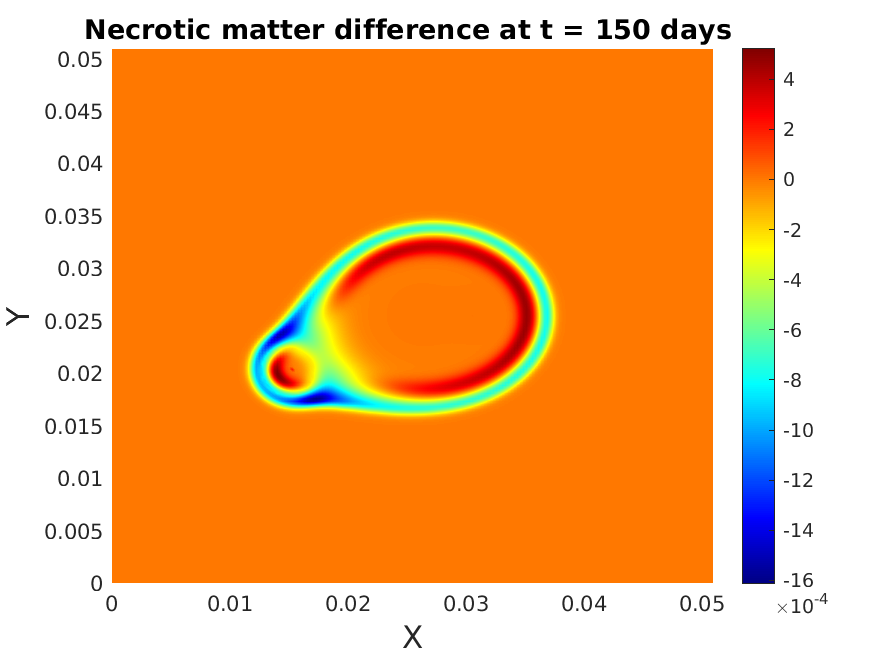}}\\
		{\includegraphics[width=1\linewidth, height = 3cm]{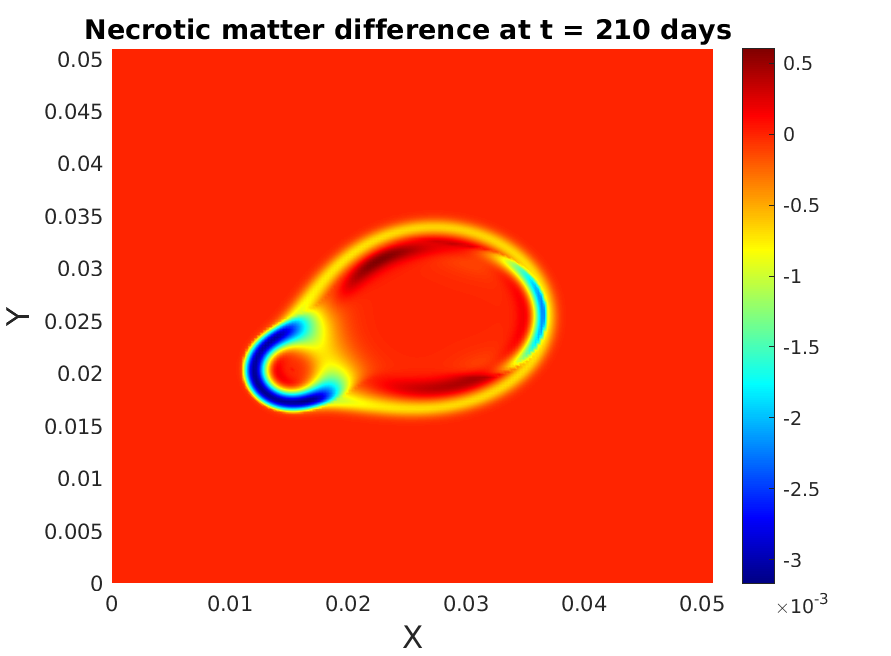}}\\
		{\includegraphics[width=1\linewidth, height = 3cm]{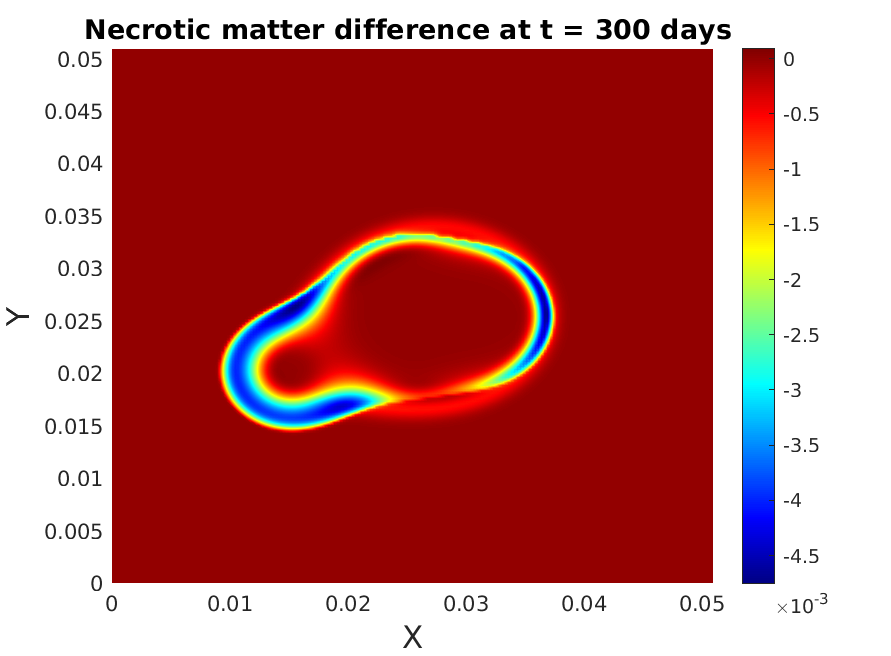}}\\
		{\includegraphics[width=1\linewidth, height = 3cm]{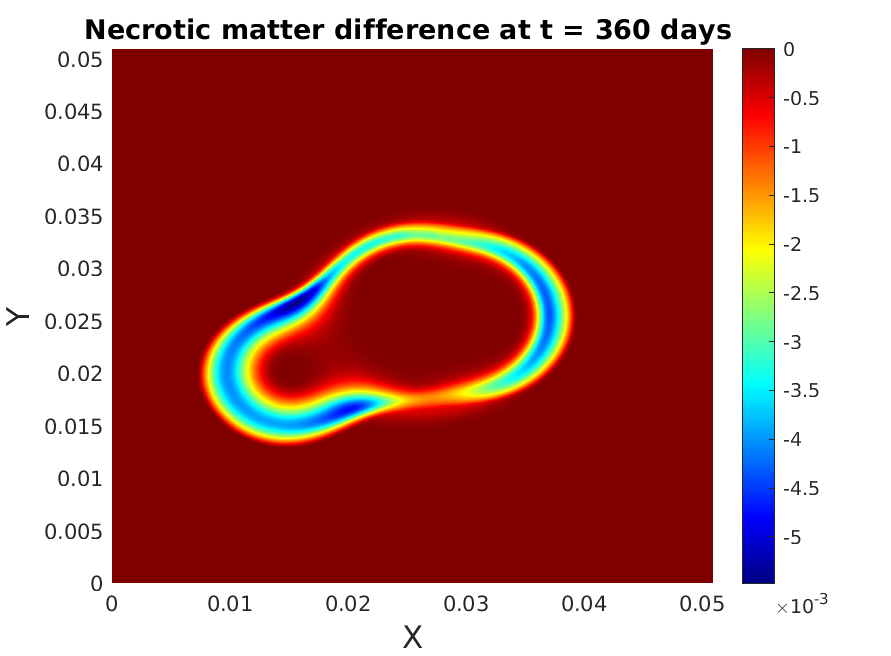}}\\
		{\includegraphics[width=1\linewidth, height = 3cm]{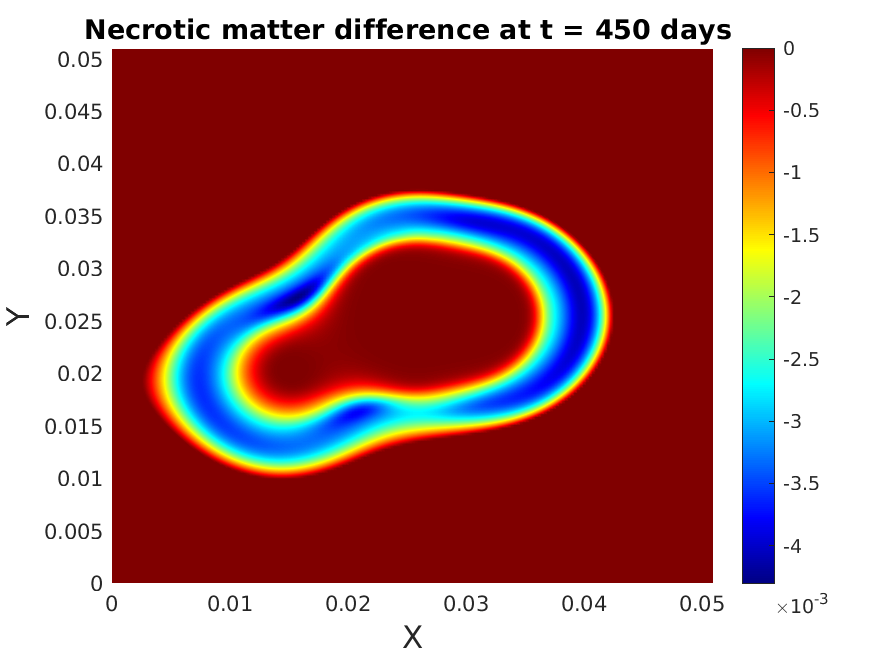}}%
		\subcaption{\scriptsize Necrotic matter difference}
	\end{minipage}%
	\caption{Difference between solution components obtained with $K_{cm}<K_{mn}<K_{cn}$ and those computed with $K_{cm}=K_{mn}=K_{cn}$, all at the previously lower value $K_{cm}$.}
	\label{fig:comp-Kcm}
\end{figure}

\begin{figure}[!htbp]
	\centering
	\begin{minipage}[hstb]{.24\linewidth}
		\raisebox{1.2cm}{\rotatebox[origin=t]{90}{30 days}}{\includegraphics[width=1\linewidth, height = 3cm]{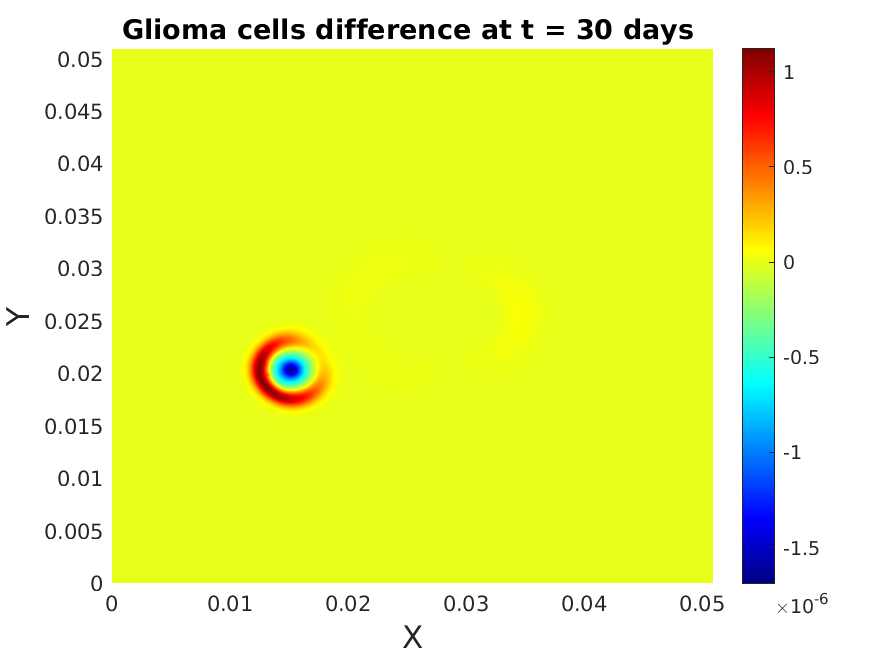}}\\
		\raisebox{1.2cm}{\rotatebox[origin=t]{90}{90 days}}{\includegraphics[width=1\linewidth, height = 3cm]{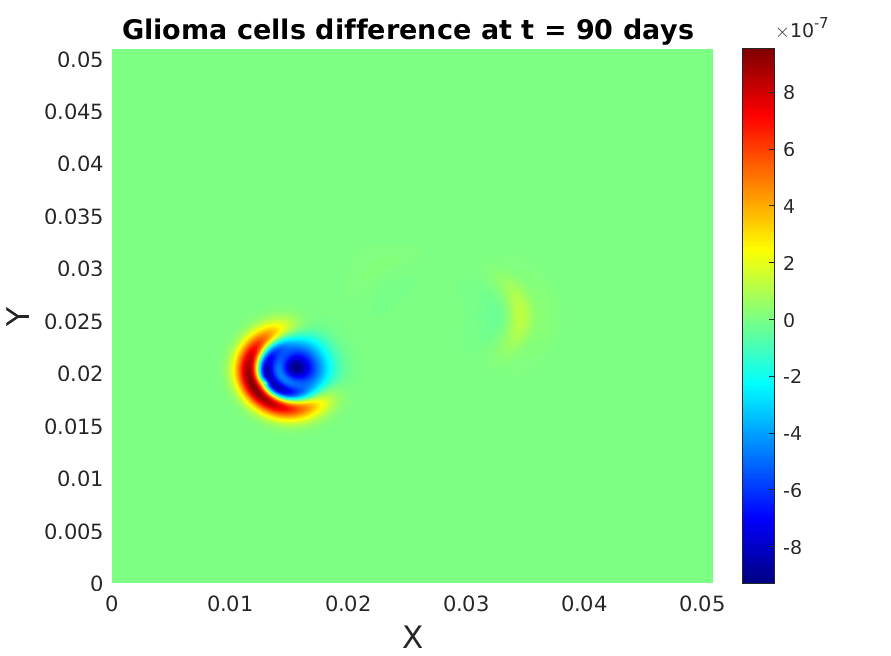}}\\
		\raisebox{1.2cm}{\rotatebox[origin=t]{90}{150 days}}{\includegraphics[width=1\linewidth, height = 3cm]{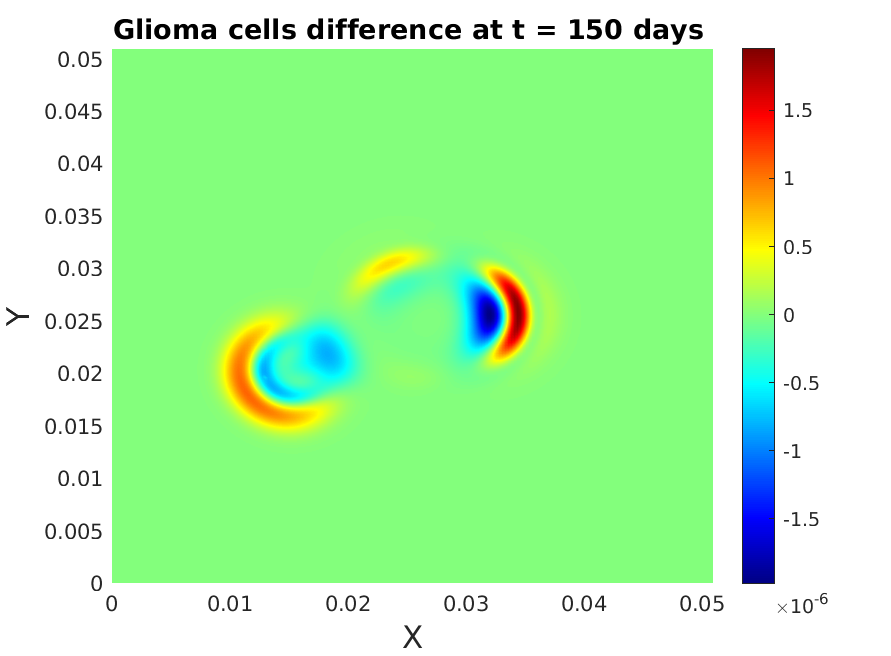}}\\
		\raisebox{1.2cm}{\rotatebox[origin=t]{90}{210 days}}{\includegraphics[width=1\linewidth, height = 3cm]{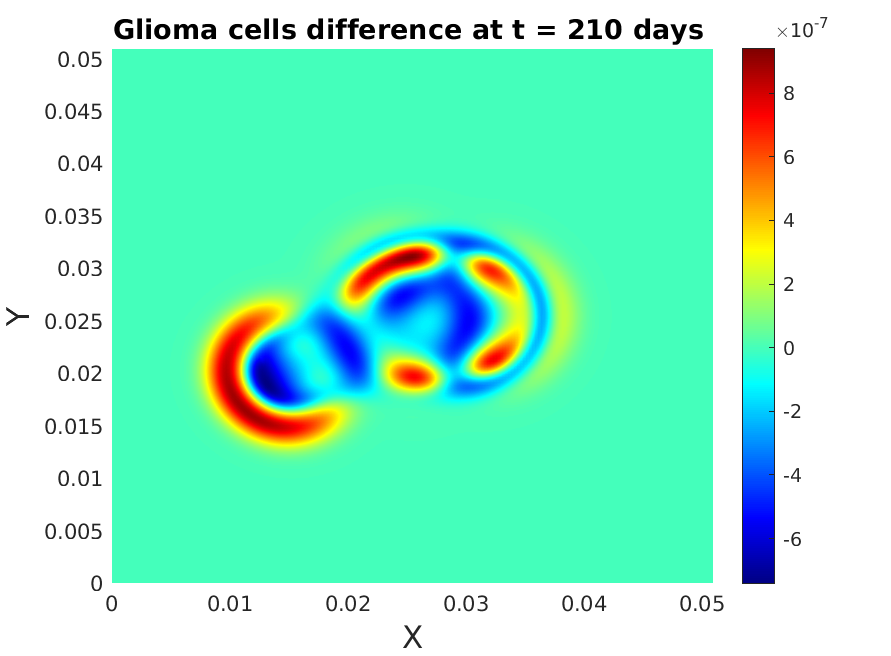}}\\
		\raisebox{1.2cm}{\rotatebox[origin=t]{90}{300 days}}{\includegraphics[width=1\linewidth, height = 3cm]{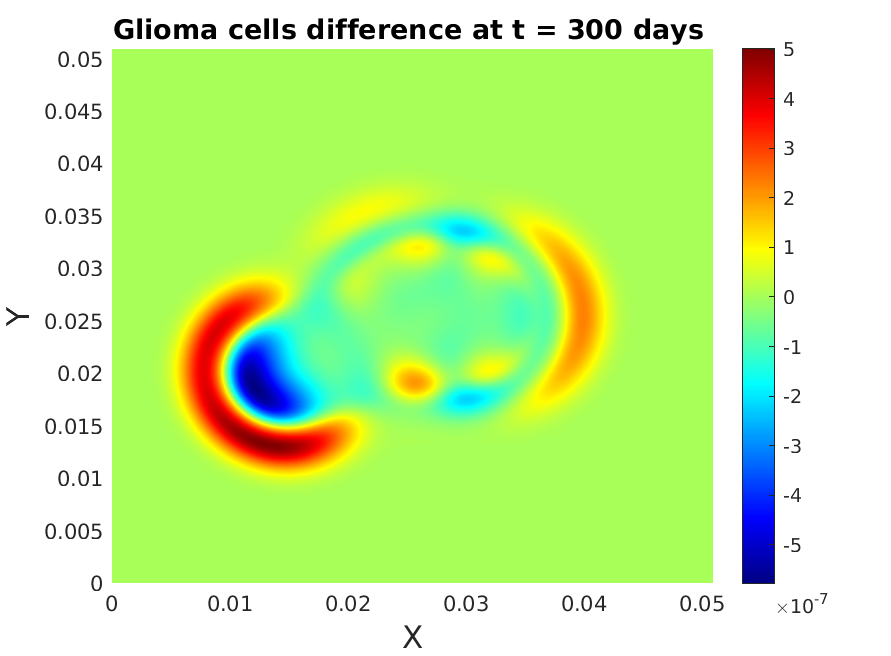}}\\
		\raisebox{1.2cm}{\rotatebox[origin=t]{90}{360 days}}{\includegraphics[width=1\linewidth, height = 3cm]{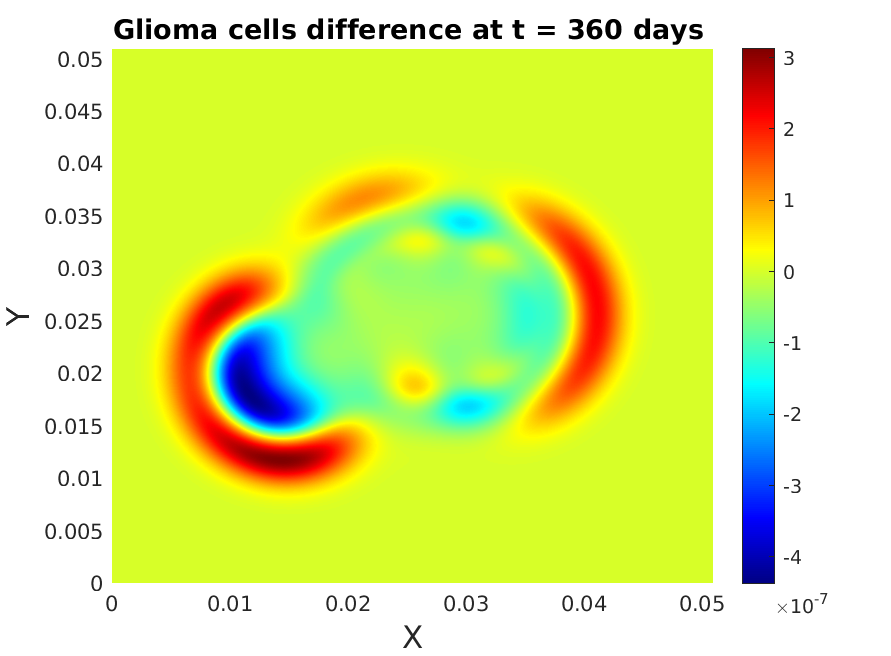}}\\
		\raisebox{1.2cm}{\rotatebox[origin=t]{90}{450 days}}{\includegraphics[width=1\linewidth, height = 3cm]{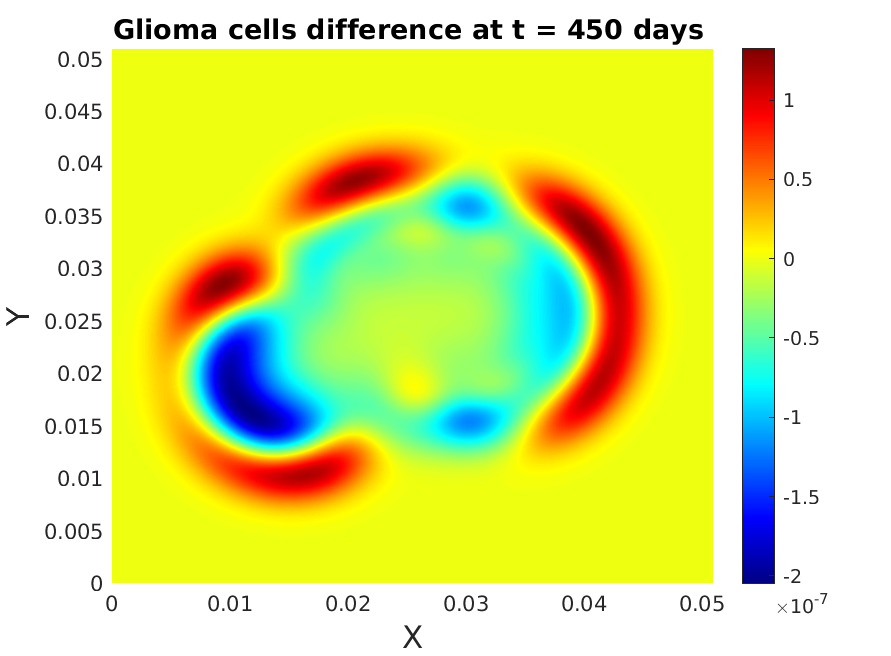}}%
		\subcaption{\scriptsize Glioma cells difference}
	\end{minipage}%
	\hspace{0.2cm}
	\begin{minipage}[hstb]{.24\linewidth}
		{\includegraphics[width=1\linewidth, height = 3cm]{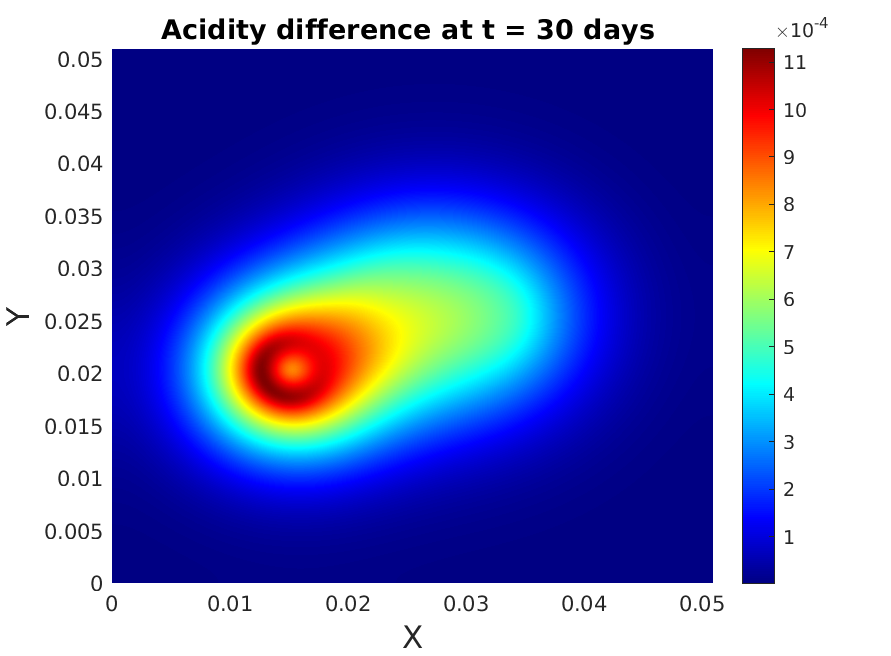}}\\
		{\includegraphics[width=1\linewidth, height = 3cm]{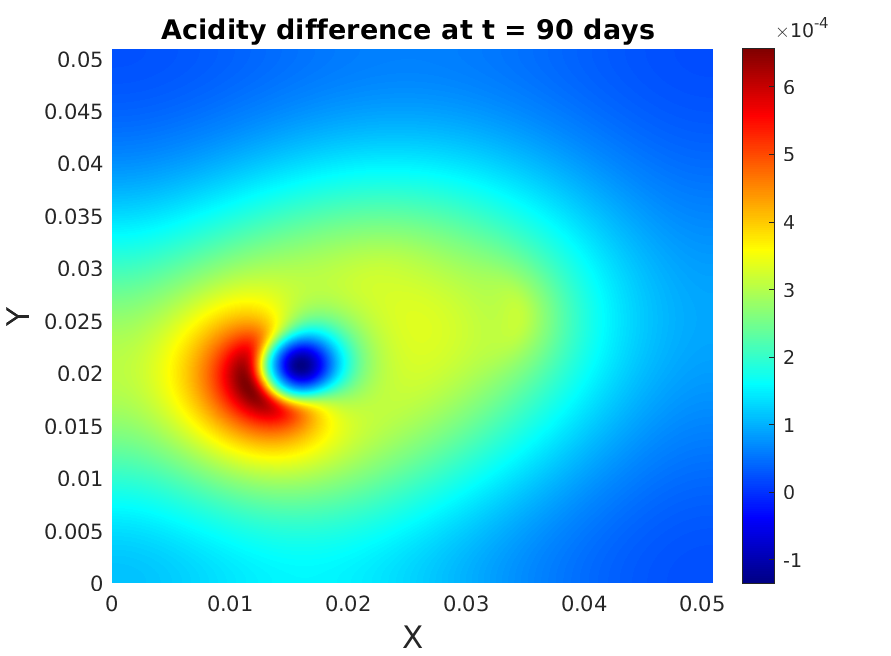}}\\
		{\includegraphics[width=1\linewidth, height = 3cm]{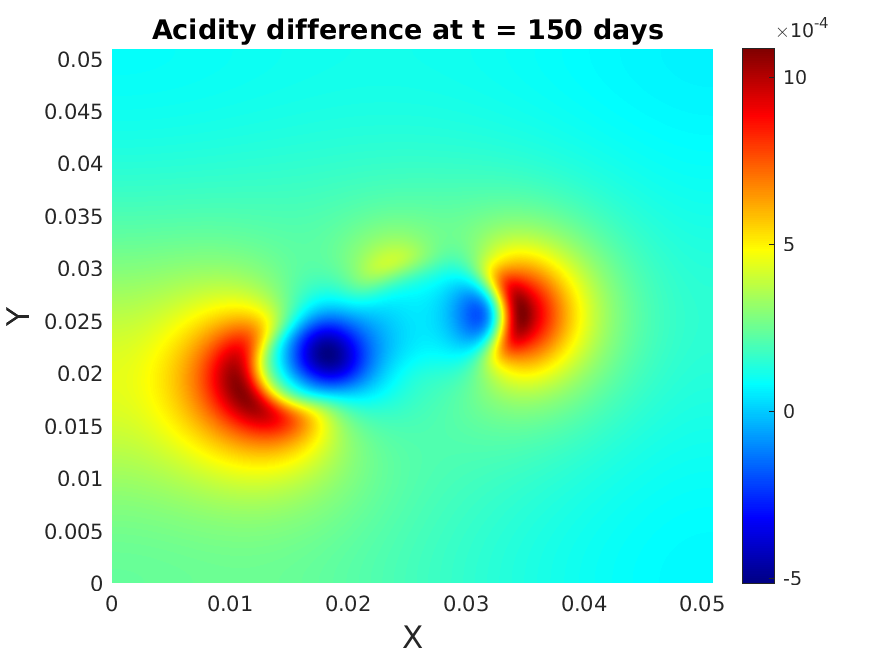}}\\
		{\includegraphics[width=1\linewidth, height = 3cm]{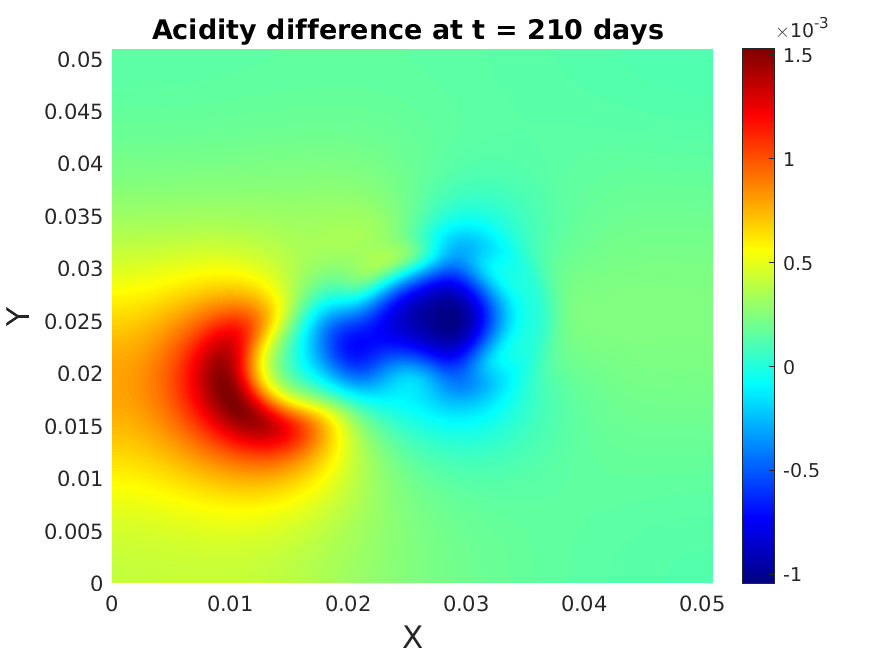}}\\
		{\includegraphics[width=1\linewidth, height = 3cm]{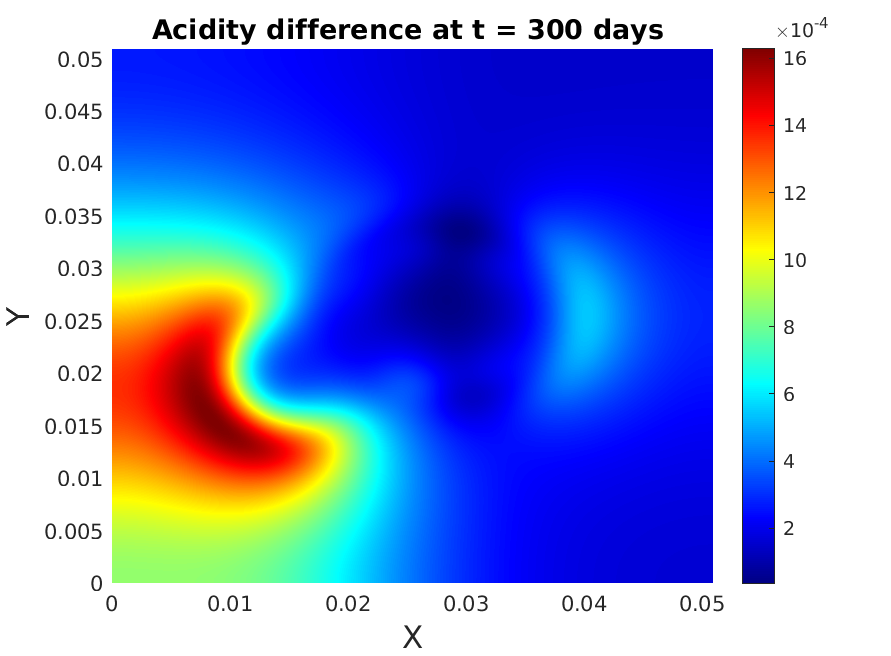}}\\
		{\includegraphics[width=1\linewidth, height = 3cm]{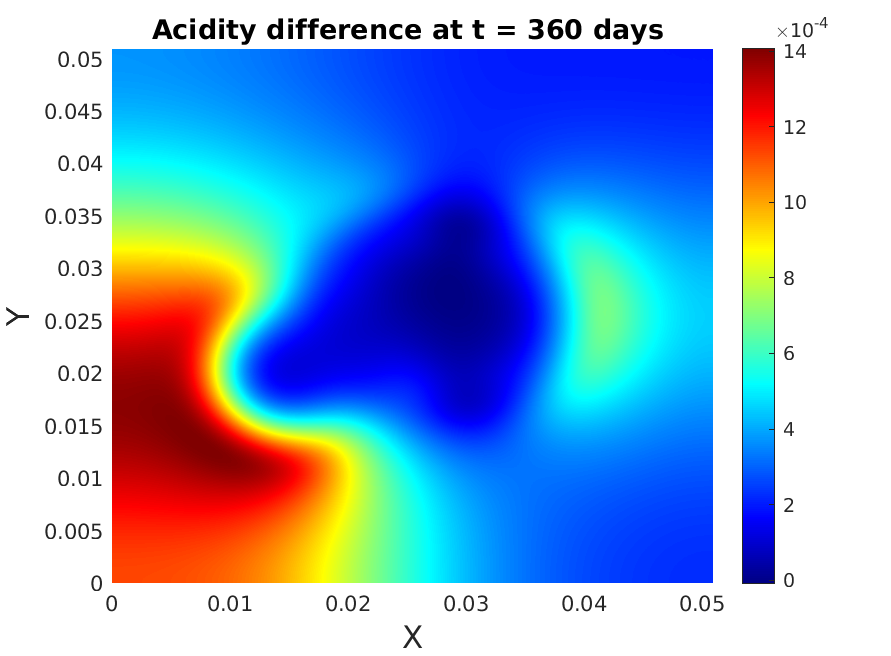}}\\
		{\includegraphics[width=1\linewidth, height = 3cm]{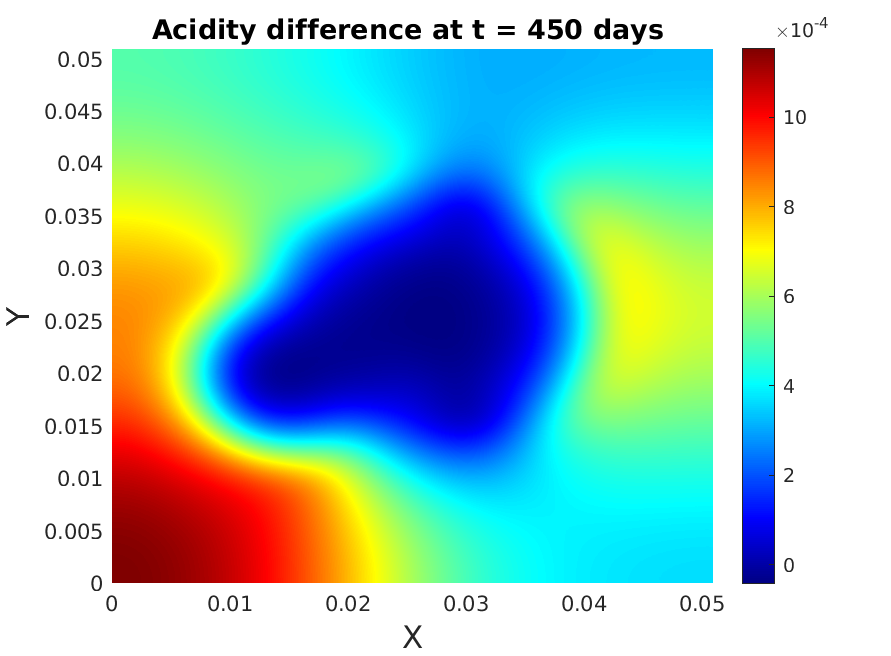}}%
		\subcaption{\scriptsize Acidity difference}
	\end{minipage}%
	\hspace{0.01cm}
	\begin{minipage}[hstb]{.24\linewidth}
		{\includegraphics[width=1\linewidth, height = 3cm]{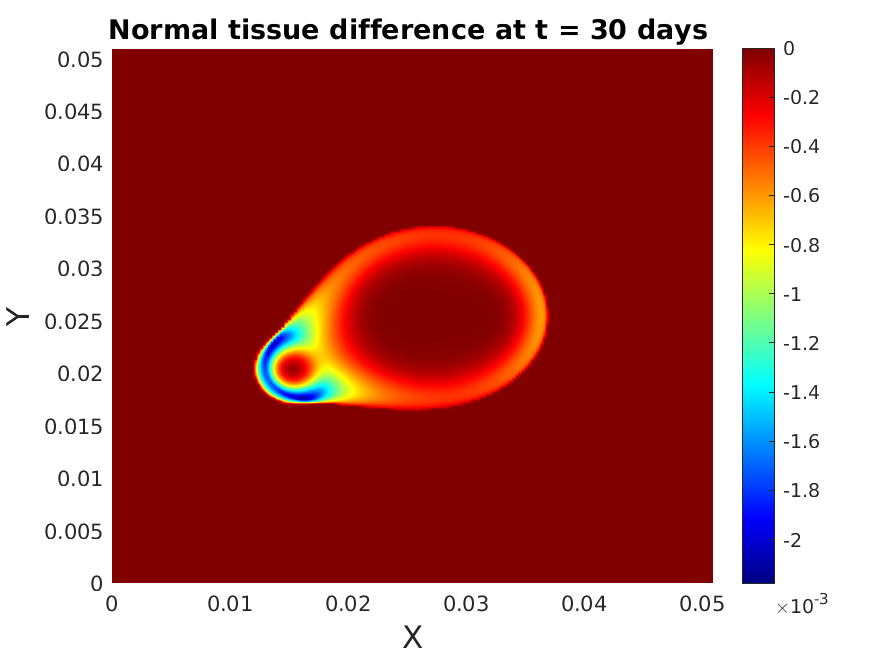}}\\
		{\includegraphics[width=1\linewidth, height = 3cm]{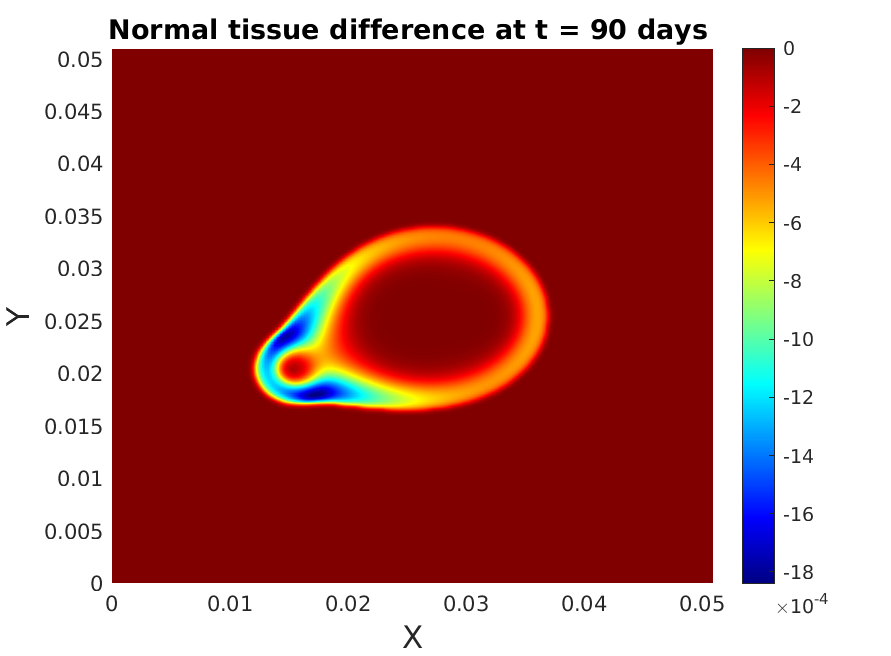}}\\
		{\includegraphics[width=1\linewidth, height = 3cm]{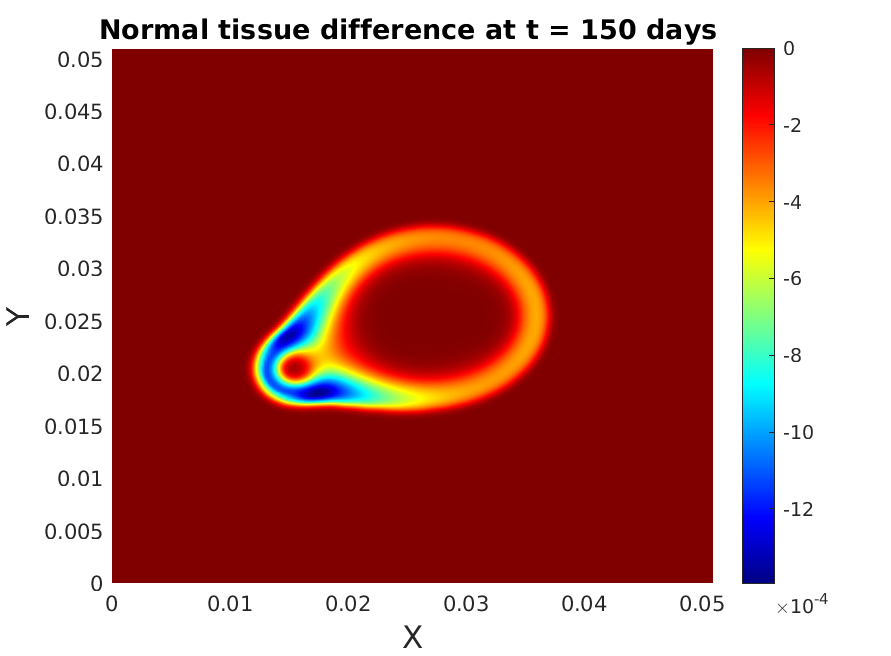}}\\
		{\includegraphics[width=1\linewidth, height = 3cm]{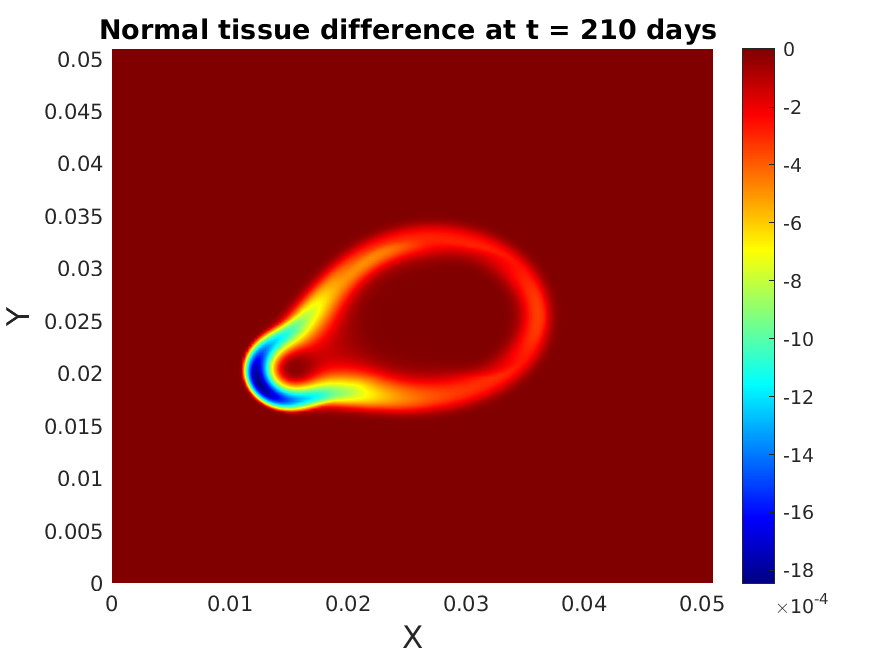}}\\
		{\includegraphics[width=1\linewidth, height = 3cm]{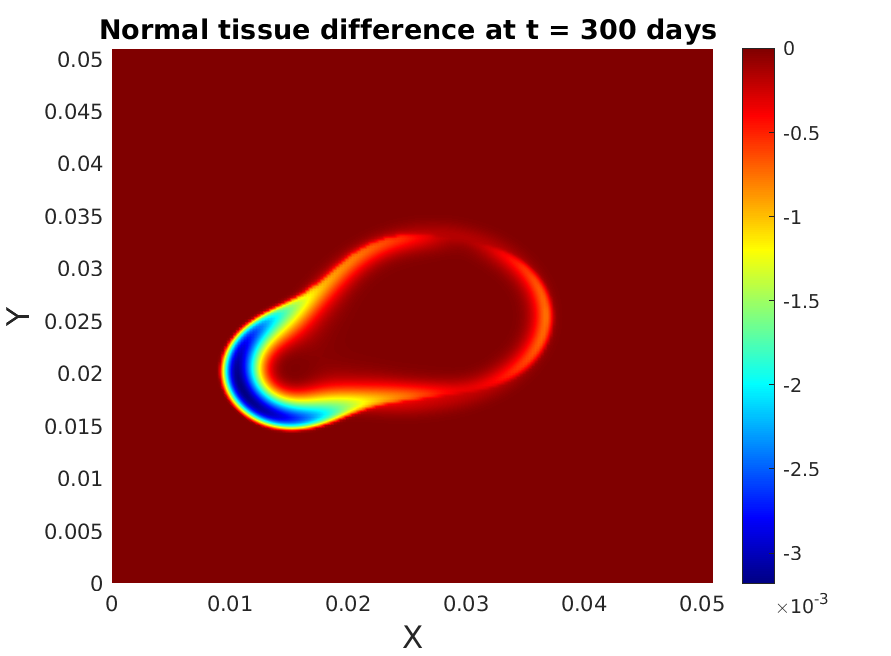}}\\
		{\includegraphics[width=1\linewidth, height = 3cm]{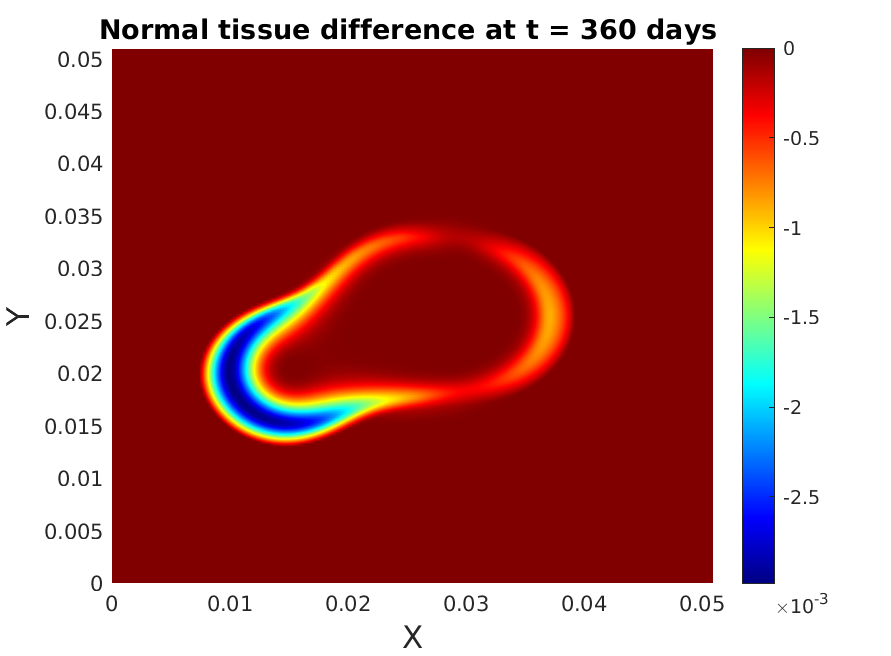}}\\
		{\includegraphics[width=1\linewidth, height = 3cm]{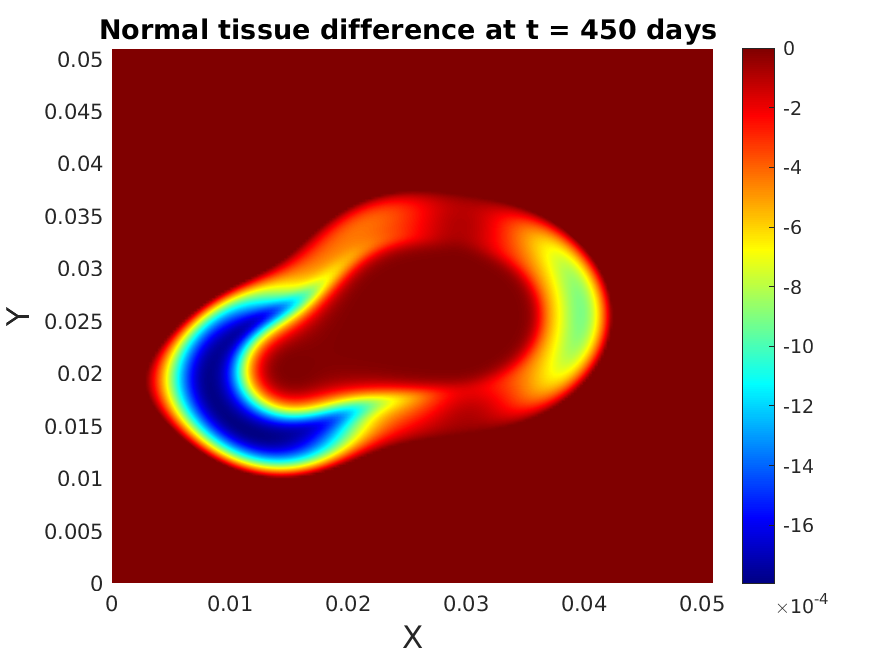}}%
		\subcaption{\scriptsize Normal tissue difference}
	\end{minipage}%
	\hspace{0.01cm}
	\begin{minipage}[hstb]{.24\linewidth}
		{\includegraphics[width=1\linewidth, height = 3cm]{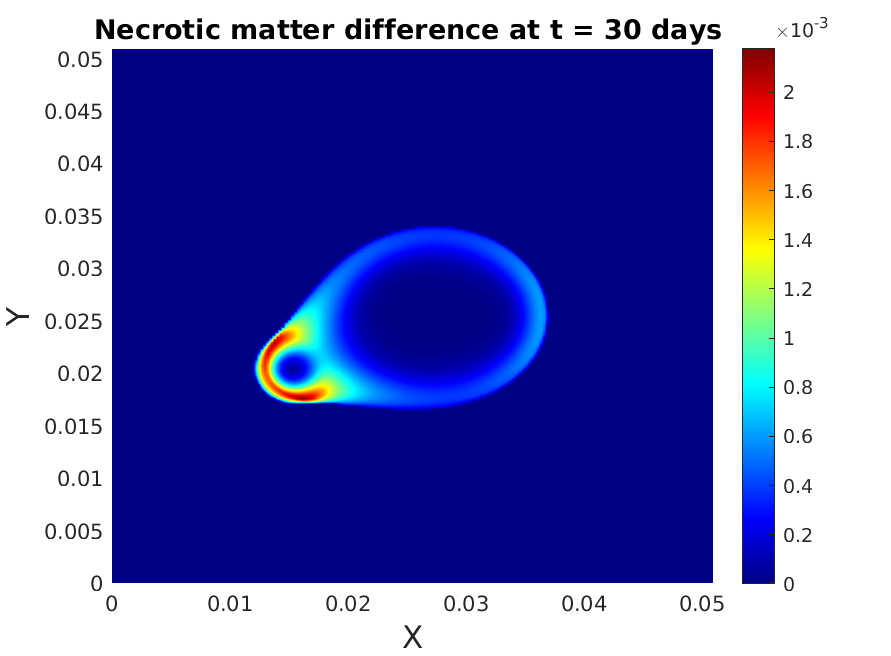}}\\
		{\includegraphics[width=1\linewidth, height = 3cm]{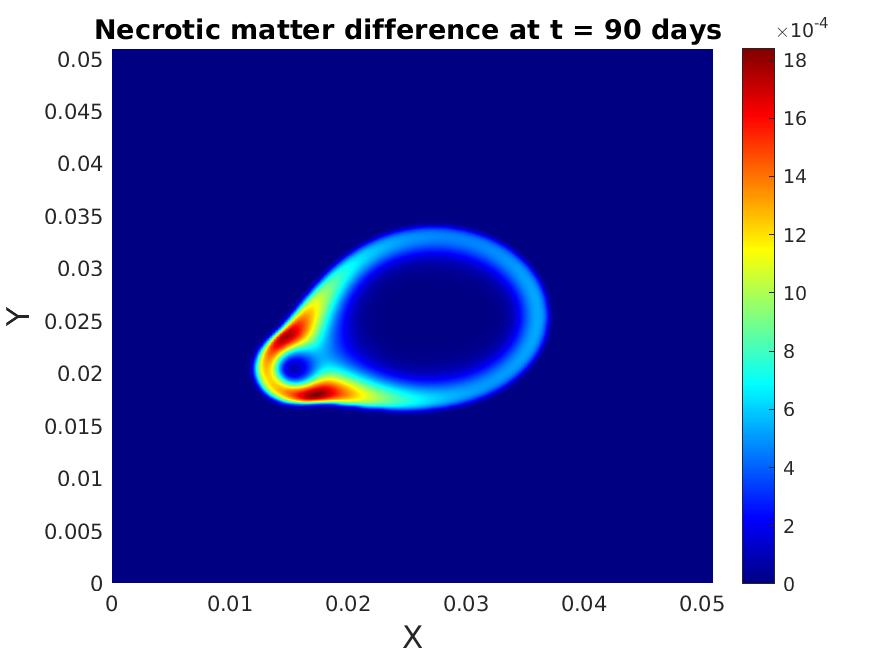}}\\
		{\includegraphics[width=1\linewidth, height = 3cm]{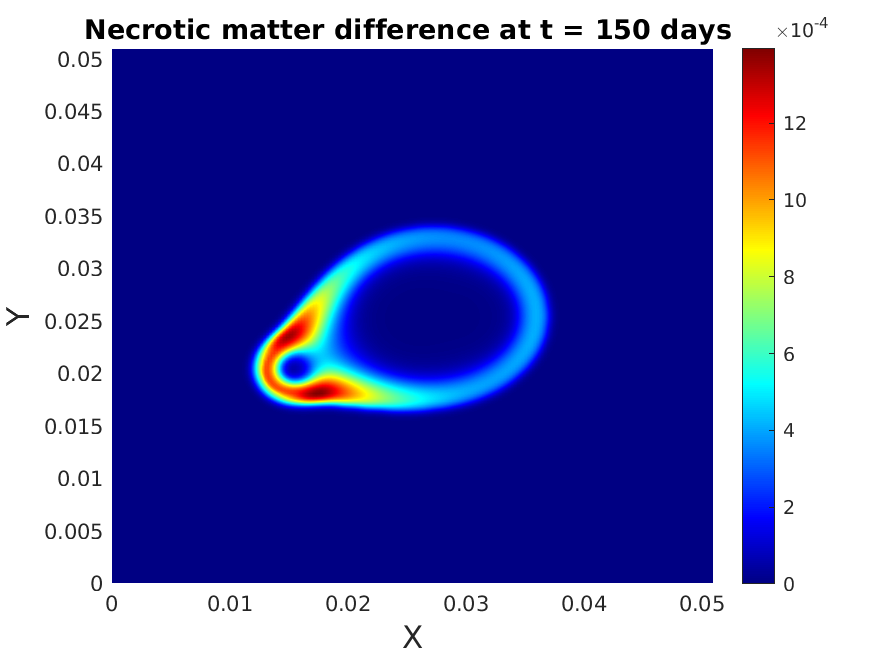}}\\
		{\includegraphics[width=1\linewidth, height = 3cm]{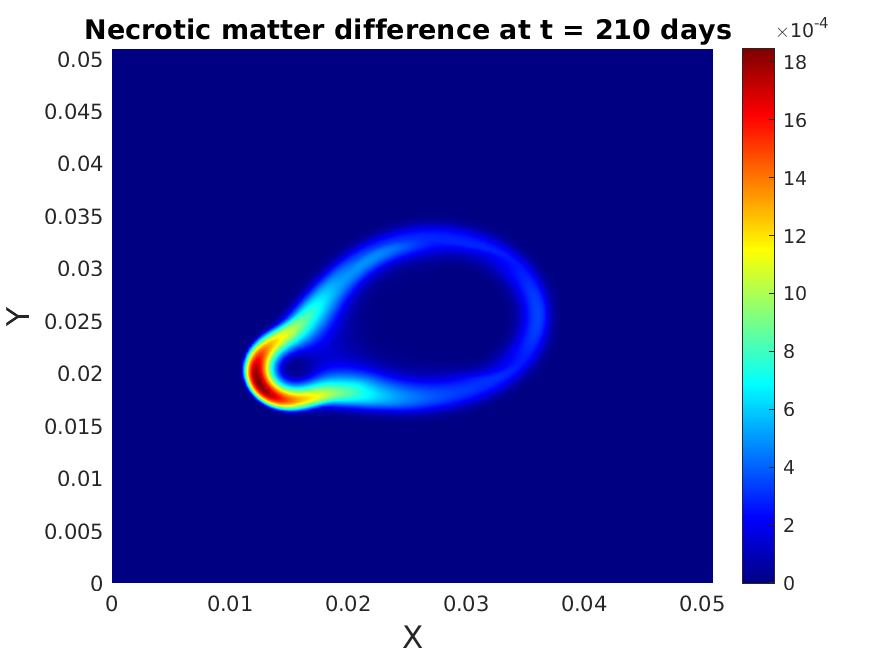}}\\
		{\includegraphics[width=1\linewidth, height = 3cm]{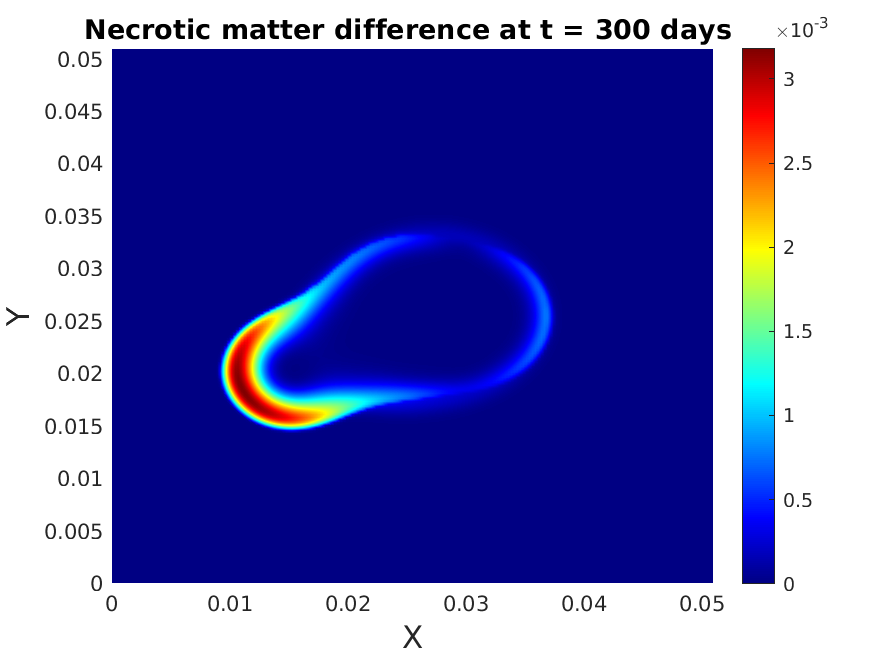}}\\
		{\includegraphics[width=1\linewidth, height = 3cm]{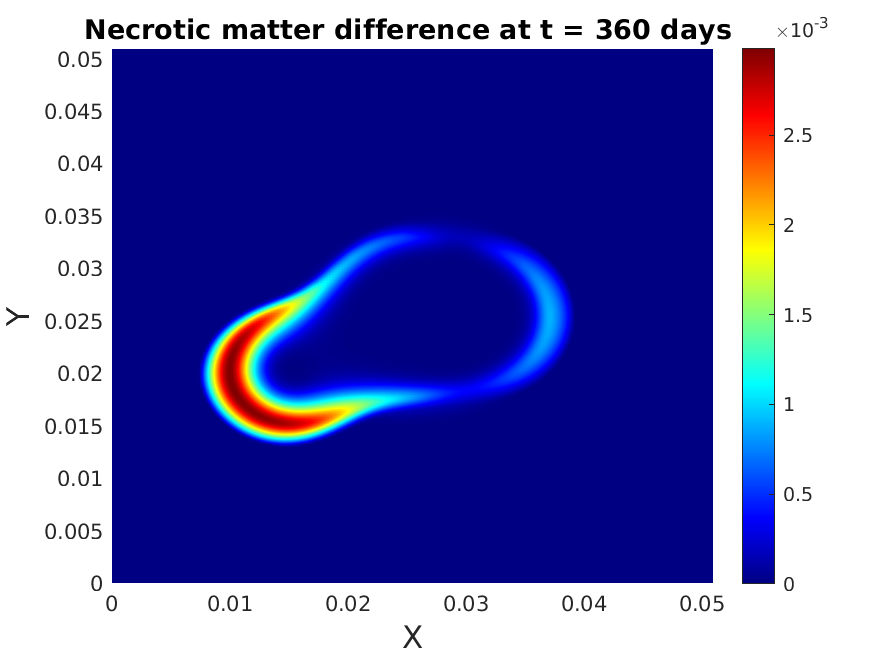}}\\
		{\includegraphics[width=1\linewidth, height = 3cm]{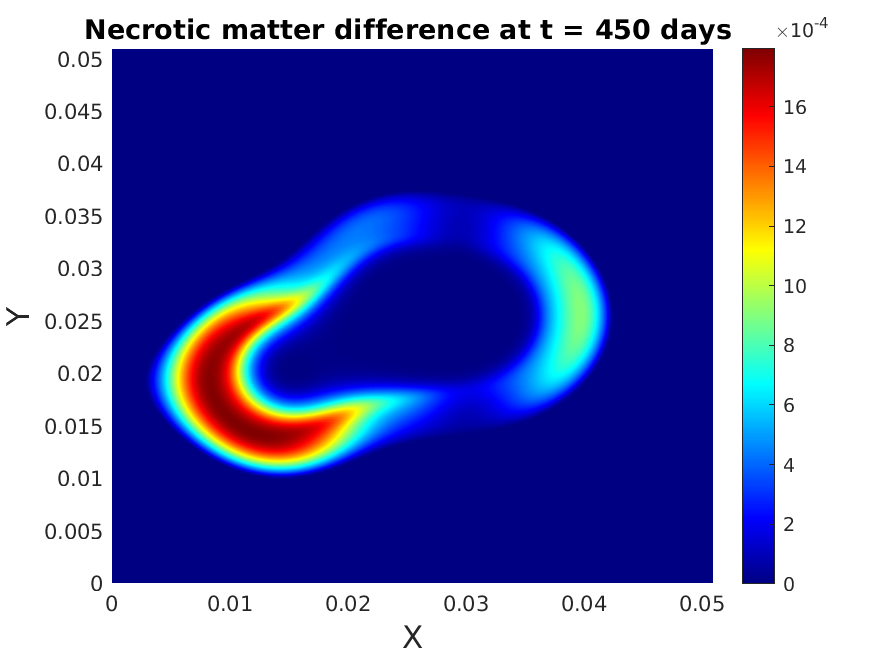}}%
		\subcaption{\scriptsize Necrotic matter difference}
	\end{minipage}%
	\caption{Difference between solution components obtained with $K_{cm}<K_{mn}<K_{cn}$ and those computed with $K_{cm}=K_{mn}=K_{cn}$, all at the previously intermediate value $K_{mn}$.}
	\label{fig:comp-Kmn}
\end{figure}

\section*{Appendix}

\subsection*{Non-dimensionalisation}
The following rescalings of the dimensional quantities involved in the obtain equations are considered:
\begin{align*}
&\tilde{h} := \frac{h}{h_{max}},\qquad \tilde{t} := t a_3,  
\qquad \tilde{x} := x\sqrt{\frac{a_3 }{D_H}}\qquad \tilde{c}_1 := \frac{c_1}{a_3}, \qquad \tilde{c}_2 := \frac{c_2}{a_3},\\
& \tilde{c}_3 := \frac{c_3}{a_3}, \qquad \tilde{K}_{mn} := \frac{K_{mn} D_H}{\alpha_c}, \qquad \tilde{K}_{cm} := \frac{K_{cm} D_H}{\alpha_c}, \qquad \tilde{K}_{cn} := \frac{K_{cn} D_H}{\alpha_c},\\
&\tilde{\chi} := \frac{\chi h_{max}}{\alpha_c}, \qquad \tilde{\alpha}_{m}:= \frac{\alpha_{m}}{\alpha_c},  \qquad \tilde{a}_2:= \frac{a_2}{a_3 h_{max}}, \qquad \tilde{g}(\tilde{h}): = \frac{\tilde{\chi} \tilde{h}}{1+ \tilde{h}},\\
&\tilde{D}_{cc}(u_c,u_m,u_n) :=(\tilde{K}_{cm}u_c+\tilde{K}_{mn}(1-u_c))(1-u_c)-(\tilde{K}_{cm} - \tilde{K}_{cn})u_c u_m,\\
&\tilde{D}_{cm}(u_c,u_m,u_n):=-(\tilde{K}_{cm}u_c+\tilde{K}_{mn}(1-u_c))u_c+(\tilde{K}_{cm}-\tilde{K}_{cn})u_c(1-u_m),\\
&\tilde{D}_{mc}(u_c,u_m,u_n):=-(\tilde{K}_{cm}u_m+\tilde{K}_{cn}(1 - u_m))u_m+(\tilde{K}_{cm}-\tilde{K}_{mn})u_m(1 - u_c),\\
&\tilde{D}_{mm}(u_c,u_m,u_n):=(\tilde{K}_{cm}u_m + \tilde{K}_{cn}(1 - u_m))(1-u_m)-(\tilde{K}_{cm}-\tilde{K}_{mn})u_c u_m,\\
& \tilde{S}(u_c,u_m,u_n) := \tilde{K}_{cm}\tilde{K}_{cn}u_c + \tilde{K}_{cm}\tilde{K}_{mn}u_m + \tilde{K}_{cn}\tilde{K}_{mn}(1-u_c-u_m).
\end{align*}
Dropping the tildes for simplicity, system \eqref{model2}, \eqref{EqProt} writes in the nondimensional form as
\begin{subequations}\label{eq:2D_nondim}
	\begin{align}
	\partial_t u_c  &= \nabla \cdot \left( \frac{1}{S(u_c,u_m,u_n)} \left[  D_{cc} \left( 2u_c + g(h) \right) + D_{cm} \alpha_m u_m^2 \theta  \right] \nabla u_c   \right) \nonumber\\
	&+ \nabla \cdot \left( \frac{1}{S(u_c,u_m,u_n)} \left[ D_{cc} u_c \nabla\left(g(h)\right) + 2D_{cm} \alpha_m u_m \left( 1 + \theta u_c \right) \nabla u_m  \right]  \right) -c_1 u_c u_n \left(h-1\right), \label{eq:mpm_final_c}\\
	\partial_t u_m &= \nabla \cdot \left( \frac{1}{S(u_c,u_m,u_n)} \left[2D_{mm}\alpha_m u_m \left(1 + \theta u_c\right)   \right] \nabla u_m ,   \right) \nonumber\\
	&+ \nabla \cdot \left( \frac{1}{S(u_c,u_m,u_n)} \left[ D_{mc} \nabla \left( u_c^2 + g(h)u_c  \right) + D_{mm} \alpha_m \theta u_m^2 \nabla u_c   \right]    \right) - \frac{c_2 u_m(h-1)_{+}}{1 + u_m + h} - c_3 u_m, \label{eq:mpm_final_m} \\
	\partial_t h &= \Delta h + a_2 u_c - h, \label{eq:mpm_final_h}\\
	u_n &= 1 - u_c - u_m. \label{eq:mpm_final_n}
	\end{align}
\end{subequations}

\subsection*{Parameters}

The following table summarises the parameters used in the simulations of Section \ref{SecNum}.

\begin{table}[!ht]	
	\centering
	\begin{tabular}{p{1.3cm} p{6.0cm} p{4.2cm} p{2.8cm}}
		\toprule[1.5pt]
		\textbf{Parameter} & \begin{minipage}{6cm} \centering \textbf{Meaning} \end{minipage} & \textbf{Value} & \textbf{Reference}\\
		\hline \\[-1.5ex]
		$h_{max}$ & acidity threshold for cancer cell death & $10^{-6.4}$ mol/L & \cite{Webb2011}\\
		$a_2$ & proton production rate & $10^{-9}$ mol /(L$\cdot $s)  & \cite{kumar2020flux}, \cite{Martin2} \\
		$a_3$ & proton removal rate  & $10^{-11}$ /s& \cite{kumar2020multiscale}\\
		$D_h$ & acidity diffusion coefficient & \begin{minipage}{6cm}$5\cdot 10^{-14}- 10^{-11}$ m$^2$/s\end{minipage} & \cite{kumar2020flux}, \cite{kumar2020multiscale} \\
		$c_1$ & glioma growth/decay rate & $2.31\cdot 10^{-6}$ /s & \cite{stein2007mathematical, eikenberry2009virtual} \\
		$c_2$ & normal cells decay rate& $3.47\cdot 10^{-6}$ /s& this work\\
		$c_3$ & normal cells decay rate& $3.47\cdot 10^{-8}$ /s& this work \\
		$K_{cm}$ &drag force coefficient between glioma and normal cells & $1.62 \cdot 10^{10}$ N m$^{-4}$s&this work \\
		$K_{mn}$ & drag force coefficient between normal and necrotic cells& $1.66 \cdot 10^{10}$ N m$^{-4}$s&this work \\
		$K_{cn}$ & drag force coefficient between glioma and normal cells&$1.7 \cdot 10^{10}$ N m$^{-4}$s&this work\\
		$\chi$ &coefficient in the additional pressure on glioma due to acidity interaction& $2.98$ N m$^{-2}$ L/mol &this work \\
		$\theta$ &dimensionless coefficient in the additional pressure on normal cells due to glioma &1 & this work\\
		$\alpha_m$ &coefficient in the additional pressure on normal cells due to glioma & $1.49 \cdot 10^{-7}$ N m$^{-2}$ & this work\\
		$\alpha_c$ & coefficient in the additional pressure on glioma due to their crowding& $1.49 \cdot 10^{-7}$ N m$^{-2}$   & this work\\
		\bottomrule[1.5pt]
	\end{tabular}
	\caption{Parameters (dimensional quantities) involved in \eqref{model2}.}\label{table_3} 
\end{table}

\newpage
\phantomsection
\printbibliography

\end{document}